\documentclass[aps,10pt,twocolumn,showpacs,pra,citeautoscript,floatfix,superscriptaddress,longbib]{revtex4-1}
\pdfoutput=1

\usepackage[toc,title]{appendix}
\usepackage{amsmath}
\usepackage{cancel}
\usepackage{amssymb}
\usepackage{amsthm}
\usepackage{bbold}
\usepackage{enumitem}
\usepackage[hidelinks,breaklinks]{hyperref}
\usepackage{bm}
\usepackage{comment}
\usepackage[normalem]{ulem}
\usepackage{thm-restate}
\usepackage{thmtools,nameref,cleveref}

\usepackage{chngcntr}
\usepackage{apptools}
\usepackage{graphicx}
\usepackage{mathtools}

\usepackage{amsmath}

\usepackage{amssymb}
\usepackage{epsfig}
\usepackage{color}
\usepackage{amsthm,bbold}
\usepackage[top=50pt,bottom=50pt,left=60pt,right=60pt]{geometry}
\usepackage{mathtools}
\usepackage[hidelinks]{hyperref}

\usepackage{soul}
\usepackage{verbatim}
\usepackage{wasysym}
\usepackage{subcaption} 

\usepackage{tikz,subfloat}
\usetikzlibrary{decorations.pathreplacing,calligraphy}
\usetikzlibrary{shapes}	

\usepackage{thm-restate}
\usepackage{thmtools,nameref,cleveref}

\usepackage{array}[=2016-10-06] 
\usepackage{tabularx}
    \newcolumntype{L}{>{\raggedright\arraybackslash}X}

\newlength{\NOTskip}

\usepackage{dsfont}
\usepackage{bbm} 
\usepackage[normalem]{ulem} 
\usepackage{color} 
\usepackage{braket} 
\usepackage{float}  

\def\ba#1\ea{\begin{align}#1\end{align}} 

\newtheorem{definition}{Definition}[section]

\newtheorem{remark}{Remark}[section]

\newtheorem{corollary}{Corollary}[section]

\newcommand{\Mref}[1]{\mbox{\Cref{#1}}}

\newcommand\mpwS[1]{{\let\helpcmd\sout\parhelp#1\par\relax\relax} }
\long\def\parhelp#1\par#2\relax{%
	\helpcmd{#1}\ifx\relax#2\else\par\parhelp#2\relax\fi%
}

\newcommand{\nocontentsline}[3]{}
\newcommand{\tocless}[2]{\bgroup\let\addcontentsline=\nocontentsline#1{#2}\egroup}

\newcommand{\ketbra}[2]{\ket{#1}\!\!\bra{#2}}
\newcommand{\proj}[1]{\ketbra{#1}{#1}}
\newcommand{\tr}{\textup{tr}}

\newcommand{\idd}[1]{1}%
\newcommand{\id}{{1}} 

\newcommand{\tb}[1]{$\overline{\mbox{#1}}$}
\newcommand{\mb}[1]{\overline{#1}}

\newcommand{\Ss}{\textup{S}}
\newcommand{\A}{\textup{A}}
\newcommand{\M}{\textup{M}}
\newcommand{\F}{\textup{F}}
\newcommand{\B}{\textup{B}}
\newcommand{\R}{\textup{R}}

\newcommand{\W}{\textup{W}}
\newcommand{\U}{\textup{U}}
\newcommand{\Lab}{\textup{L}}
\newcommand{\Lbar}{\textup{\tb {L}}}

\makeatletter
\def\p@paragraph{\thesection\,}
\makeatother

\usepackage{dutchcal}
\DeclareMathAlphabet{\matholdcal}{OMS}{cmsy}{m}{n}
\newcommand{\aaa}{\mathcal{a}}
\newcommand{\bbb}{\mathcal{b}}
\newcommand{\xxx}{\xi}
\newcommand{\kkk}{\mathcal{k}}

\makeatletter 
    
\renewcommand\onecolumngrid{
\do@columngrid{one}{\@ne}%
\def\set@footnotewidth{\onecolumngrid}
\def\footnoterule{\kern-6pt\hrule width 1.5in\kern6pt}%
}

\renewcommand\twocolumngrid{
        \def\footnoterule{
        \dimen@\skip\footins\divide\dimen@\thr@@
        \kern-\dimen@\hrule width.5in\kern\dimen@}
        \do@columngrid{mlt}{\tw@}
}%

\makeatother    
\begin{document}
\setlength{\parskip}{0pt}
	\title{A general quantum circuit framework for Extended Wigner's Friend Scenarios: logically and causally consistent reasoning without absolute measurement events}

	\begin{abstract}
Extended Wigner’s Friend Scenarios (EWFSs) model agents themselves as unitarily evolving quantum systems, going beyond the standard use of quantum theory where agents are treated classically. Several no-go results suggest deep challenges in this setting: Frauchiger and Renner (FR) argued that quantum agents reasoning using quantum theory would arrive at logical paradoxes, while the Local-Friendliness theorem raises difficulties for maintaining objective notions of measurement events and causal reasoning in EWFSs. This prompts a key question: can the ability to make and test scientific predictions and reason consistently still be maintained if quantum theory applies universally, even to agents and their memories?
We answer this positively by developing a general quantum circuit framework for EWFSs without assuming absolute measurement outcomes. We formalize Heisenberg cuts in terms of distinct choices of channels in the circuit, labelled by “settings,” and show that FR-type paradoxes are fully resolved by making explicit the conditioning on these channels/settings. We provide concrete reasoning rules for quantum agents that ensure logical, probabilistic and causal consistency. Our framework captures all perspectives in an EWFS within a single causal structure while allowing events to be fundamentally subjective, and we show how objective events and predictions nevertheless emerge in current real-world quantum experiments.
Thus, we demonstrate a relational yet operational framework that overcomes key challenges in EWFSs without altering the Born rule, unitarity, or classical logic and probability for outcomes. This establishes a formal platform to analyze and compare different EWFS arguments and extend quantum information methods consistently into the Wigner’s Friend domain.
	\end{abstract}	
	
	\author{V. Vilasini}
	\email{vilasini@inria.fr}
	\affiliation{Institute for Theoretical Physics, ETH Zurich, 8093 Z\"{u}rich, Switzerland}
\affiliation{Université Grenoble Alpes, Inria, 38000 Grenoble, France}
    \author{Mischa P. Woods}
    \email{mischa.woods@gmail.com}
    \affiliation{Institute for Theoretical Physics, ETH Zurich, 8093 Z\"{u}rich, Switzerland}
	\affiliation{Université Grenoble Alpes, Inria, 38000 Grenoble, France}
	\affiliation{ENS Lyon, Inria, LIP, 69342 Lyon, France}
\maketitle

\tableofcontents
 
\section{Introduction and paper overview}
\subsection{Introduction}\label{Introduction}

Quantum theory is among the most successful physical theories, accurately describing microscopic phenomena. Experiments continue to probe quantum effects in increasingly large systems, motivated by scalable quantum computing and tests of fundamental physics. If quantum theory is universally valid, it should, at least in principle, be able to model observers or agents as physical systems within the theory.

Wigner was among the first to consider this in the 1960s through his intriguing thought experiment \cite{Wigner1967}. In it, an agent (the Friend) measures a quantum system and perceives a definite classical outcome, while an external “super-agent” (Wigner) models the Friend’s entire lab as a closed, unitarily evolving quantum system and could, in principle, reverse the Friend’s measurement. This tension lies at the heart of the measurement problem, revealing how quantum theory applied to observers can have surprising empirical consequences (see also \cite{Baumann_2018}).

More recently, there has been a surge of interest in exploring Extended Wigner’s Friend Scenarios (EWFSs) involving more agents, some of which suggest radical implications for physics via no-go theorems. This began with an extension by Frauchiger and Renner (FR) \cite{Frauchiger2018}, where they argued that quantum agents who also reason about each other’s knowledge using quantum theory (e.g., Born rule, unitarity) in combination with standard axioms of classical logic (e.g., distributive axiom, consistent inheritance of knowledge) applied to their observed outcomes, would reach logical contradictions (see also \cite{Nurgalieva2018}). This led to their widely cited claim that ``Quantum theory cannot consistently describe the use of itself''.

Other no-go results, such as Brukner's \cite{brukner2015,Brukner2018} and the Local-Friendliness (LF) theorem \cite{Bong2020}, address more ontological aspects of quantum theory through EWFSs. The LF theorem shows that the assumptions of local agency (LA) and absoluteness of observed events (AoE) cannot be simultaneously satisfied in any theory that explains the predictions of a certain EWFS. This implies that if one wishes to preserve free choice of measurement settings and relativistic causality principles (both entailed by LA) in EWFS, then one can no longer have absolute observed events for all agents. 

The non-absoluteness of events thus raises concerns over whether experimental outcomes can be regarded as objective facts about the world. Meanwhile, FR’s claim, taken at face value, questions the consistent usability of quantum theory by agents in EWFSs. More generally, these issues make it imperative to unambiguously and rigorously address the following:\\

\noindent$\bold{Question}$: 
\textit{How can we reliably reason about the world and make consistent predictions if quantum theory applies universally to agents with full quantum control over each other’s labs, without assuming absolute measurement events?}\\

This is relevant because a sufficiently large quantum computer could act as an agent in such EWFS arguments (see also \cite{Wiseman_2023,Zeng2025}). EWFS no-go theorems have been met with several responses, with different interpretations of quantum theory suggesting distinct resolutions. In the FR case, responses span from conceptual clarifications \cite{Pusey2018, Nurgalieva2018, Nurgalieva2020}, to new reasoning rules for avoiding paradoxes \cite{Renes2021, Losada2019, Alexios2022}, to criticisms identifying various hidden assumptions in FR's reasoning \cite{ScottAronson, Healey2018, Araujo, Sudbery2019}. Yet no general framework has emerged that can both identify which assumptions are truly necessary for the paradoxes and provide a solution to preserve logical and causally consistent reasoning. 

Causality issues are central in the FR and LF scenarios, yet these have largely been analyzed independently without a clear unifying perspective. Some proposed solutions (e.g., \cite{Renes2021, Alexios2022}) to the FR paradox suggest keeping track of future operations by super agents in order to determine the validity of a measurement or its predictions made by a friend at a certain time, raising concerns: how much of the future must be tracked to reason about present outcomes, how can agents reason about space-like separated measurements in EWFS consistently with relativistic causality principles? The LF theorem and subsequent no-go results imply that existing quantum causal models cannot explain the predictions of such EWFSs, as they assume absoluteness of measurement outcomes \cite{Cavalcanti2021,Ying2023}. This applies even when allowing exotic scenarios with cyclic causation, raising the challenge of whether a consistent quantum causal modeling formalism can be constructed without assuming AoE, while preserving the other fundamental relativistic and operational principles such as LA.

At the core of these discussions is an ambiguity, underpinning our very understanding of what constitutes a measurement event: should a measurement be modelled as an irreversible process yielding a classical record (an outcome), or as a reversible unitary on a closed system? In quantum theory this choice, linked to the Heisenberg cut, determines which parts of an experiment are treated as classical vs quantum. Wigner’s scenario shows the pressing need to make this modelling explicit, which is not accounted for in existing quantum frameworks. Despite extensive work on the topic, there is still no consensus on how to consistently model EWFSs, which we argue, is primarily because no prior framework exists that meets all of the following desiderata. 

\begin{itemize}[left=0pt]
    \item {\bf Explicitly formulate predictions for EWFSs}, stating clearly the measurement modelling assumptions and knowledge conditions used, and clarifying the role of Heisenberg cuts.
    \item {\bf Apply to general EWFS} where agents’ labs are modeled as quantum systems which can be fully quantum-controlled by other agents, without assuming an absolute notion of events shared by all agents. 
    \item {\bf Ensure logical and causal consistency}, by providing explicit reasoning rules for such quantum agents, which are compatible with relativistic and causality principles like no-superluminal-signalling.
    \item {\bf Explain emergent objectivity of outcomes} in current real-world experiments and recover standard quantum predictions therein.
    \item {\bf Remain interpretation-independent}, so that identifying relevant assumptions and resolving paradoxes does not depend on a specific interpretation of quantum theory.
\end{itemize}

The standard quantum circuit formalism satisfies such desiderata when agents are not themselves quantum systems, providing clear causal semantics, unambiguous Born-rule predictions, and interpretation-independent usage. The EWFS no-go arguments pose challenges to extending this standard formalism to Wigner's Friend Scenarios, which we overcome here.  

\subsection{Overview of contributions and paper layout}\label{sec:Overview of contributions and paper layout}

We extend the quantum circuit framework to EWFS enabling a relational and operational formalism for quantum measurements that answers the question posed in \Cref{Introduction}. A main contribution is that FR-type contradictions \cite{Frauchiger2018} do not arise in this new framework, and it thus resolves a challenge posed in \cite{Renner_Challenge}. We explicitly identify a previously implicit assumption (which we term Assumption \hyperref[def: I_assump]{$\mathbf{I}$}) responsible for the inconsistency and show that quantum mechanics remains self-consistent when this assumption is abandoned. A key ingredient of our framework, whose central importance will recur throughout the paper, is the consistent labeling of Heisenberg cuts, which we call \emph{settings}. We now summarize our full contributions and paper layout:

\begin{itemize}[left=0pt]

\item {\bf Quantum circuit framework for EWFSs:} In \Cref{sec:background}, we review Wigner’s thought experiment, and intuitively motivate our framework through this example. In \Cref{sec:gen_framework}, we introduce our augmented quantum circuit framework for EWFSs with any number and configuration of agents, formalizing Heisenberg cuts through \emph{settings}, providing a means to treat measurement events as relative rather than absolute. This enables computing consistent probabilities, demonstrating \emph{completeness}, \emph{consistency}, and \emph{causality} without assuming AoE (\Cref{sec: consistency_pred}).

\item {\bf Universal resolution and identification of root cause of EWF paradoxes:} In \Cref{sec: consistency_reasoning}, we show our augmented circuit enables consistent reasoning with quantum rules and classical logic. We identify assumption \hyperref[def: I_assump]{$\mathbf{I}$} (independence of predictions from settings or Heisenberg cuts) as the failing assumption causing apparent paradoxes (\Cref{sec: I_assump}). Applied to the FR paradox, our framework clearly resolves it without restricting agents' reasoning (\Cref{sec: resolution_entanglement}), and clarifies the role of assumption \hyperref[def: I_assump]{$\mathbf{I}$} in relation to other EWFS theorems (\Cref{sec:The choices of Heisenberg cuts do matter in FR}).

\item {\bf Objective measurement events and causality:} \Cref{sec: standard_QT} distinguishes standard vs Wigner's Friend type quantum experiments by introducing a criterion for ``super-agency" based on causal structure. We show that standard scenarios recover objective measurement events independent of Heisenberg cuts.

\item {\bf Multi-agent reasoning and interpretations:} \Cref{sec: reasoning_rules}, compares EWFS paradoxes with classical reasoning pitfalls, proposing a paradigm for consistent reasoning even with partial knowledge. \Cref{sec: setting_interpretation} discusses the meaning of settings and their Bayesian-like updates, while \Cref{sec: interpret_indep} covers the interpretation-independence of our results.

\item {\bf Unified perspective:} Our framework provides a unified platform for multiple EWFS arguments. \Cref{appendix:relation_prev_works} relates our approach to previous responses and interpretations. In follow-up work \cite{LF_Vilasini_Woods} we apply our formalism to the LF scenario, linking AoE and \hyperref[def: I_assump]{$\mathbf{I}$} assumptions further.

\end{itemize}

\section{Motivating example: simplest Wigner's Friend Scenario}
\label{sec:background}

In this section, we review Wigner's original scenario and argument \cite{Wigner1967}, while highlighting the salient features and insights which are core to understanding our general framework and solution to EWFS paradoxes.

\begin{figure*}[ht!]
    \centering
    \includegraphics[scale=0.9]{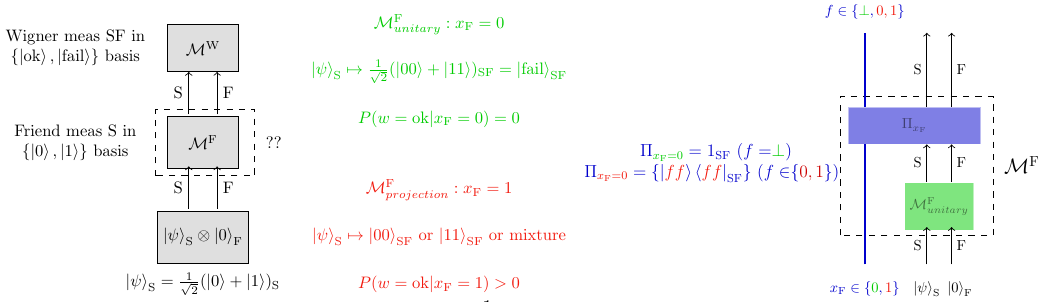}
    \caption{Time flows from bottom to top in this figure.  (Left) The Friend measures a state $\ket{\psi}_{\Ss}=\frac{1}{\sqrt{2}}(\ket{0}+\ket{1})_{\Ss}$ in the computational basis, storing the outcome in their memory $\F$ (initially in state $\ket{0}_{\F}$), while Wigner then measures $\Ss\F$ in the basis $\ket{\textup{ok/fail}}=\frac{1}{\sqrt{2}}(\ket{00}\mp \ket{11})_{\Ss\F}$, yielding outcome $w$. If the Friend's measurement is modelled unitarily, the post-measurement state is a pure entangled superposition (\textcolor{green!80!black}{green}) while if we apply the projection postulate, the post-measurement states are different (\textcolor{red!80!brown}{red}). Clearly, the probability of $w$ depends on this choice (which models the Heisenberg cut), which we differentiate by the labels \textcolor{green!80!black}{$x_{\F}=0$} and \textcolor{red!80!brown}{$x_{\F}=1$}. Ignoring the conditioning on \textcolor{blue!80!black}{$x_{\F}$} would lead to an apparent contradiction: $P(w=\textup{ok})=0$ and $P(w=\textup{ok})>0$. (Right) Our framework defines an augmented circuit which incorporates the choices and resolves all EWFS paradoxes, also in more complex set-ups. Depicted here is an augmented version of $\matholdcal{M}^{\F}$, where a non-trivial outcome \textcolor{red!80!brown}{$f\in \{0,1\}$} is only possible for \textcolor{red!80!brown}{$x_{\F}=1$} while \textcolor{green!80!black}{$x_{\F}=0$} implies there are no non-trivial classical records (\textcolor{green!80!black}{$\perp$}).}
    \label{fig:WFexample_augcirc}
\end{figure*}

The postulates of quantum theory propose two types of evolutions: unitary evolution of closed quantum systems and the projection postulate for the evolution associated with observing a measurement outcome. However, the theory does not specify when to apply each type of evolution. While this ambiguity does not affect our ability to apply quantum theory successfully in real-world experiments (see \Cref{sec: standard_QT} for further discussion on this point), Wigner's 1967 thought experiment highlights beautifully why we should be concerned about this ambiguity. 

The thought experiment assumes that quantum theory is universally applicable to measurement devices, agents performing the measurements, and their laboratories. Here an agent need not be conscious or human: a sufficiently advanced quantum computer that can measure another system, store the outcome in quantum memory, and perform basic computations can serve as an agent in Wigner's Friend no-go theorems (see also \cite{Nurgalieva2020,Wiseman_2023,Zeng2025}).

Suppose Alice (Wigner's friend) measures a quantum system $\Ss$ prepared in the state $\ket{\psi}=\alpha\ket{0}_\Ss+\beta\ket{1}_\Ss$ in the computational basis, obtaining a classical outcome $a$ with values $\aaa \in {0,1}$. She stores this outcome in her memory $\A$, initialised to the state $\ket{0}_\A$. Treating $\A$ as representing the rest of Alice's lab, $\Ss\A$ represents Alice's entire lab. If Alice's lab is a closed quantum system, it evolves unitarily according to the initial premise of the universal validity of quantum theory. Thus, according to the unitarity postulate, we would describe the evolution of Alice's lab as

\begin{equation}
\label{eq:M_unitary}
    \matholdcal{M}^{\A}_{unitary}: \ket{\psi}_{\Ss}\otimes \ket{0}_\A \mapsto \alpha \ket{00}_{\Ss\A}+\beta \ket{11}_{\Ss\A},
\end{equation}

where $\ket{\aaa\aaa}_{\Ss\A}$ represents the effective state of Alice's lab when she observes the outcome $a=\aaa$. When the memory is initialised to $\ket{0}_\A$, the unitary evolution $\matholdcal{M}^\A_{unitary}$ for the computational basis measurement on $\Ss$ (storing the outcome value in memory) is a CNOT gate with $\Ss$ as control and $\A$ as target.

Alternatively, applying the projection postulate, when Alice obtains outcome $a=0$, her lab's evolution is $\ket{\psi}_\Ss\ket{0}_\A \mapsto \ket{0}_\Ss\ket{0}_\A$, and for $a=1$, it is $\ket{\psi}_\Ss\ket{0}_\A \mapsto \ket{1}_\Ss\ket{1}_\A$. These are trace-decreasing evolutions, but summing over all possible outcomes gives the trace-preserving evolution:

\begin{align}
\label{eq:M_mixture}
\begin{split}
\matholdcal{M}^{\A}_{projection}: &\proj{\psi}_{\Ss}\otimes \proj{0}_\A \mapsto \\ &|\alpha|^2 \ketbra{00}{00}_{\Ss\A}+|\beta|^2 \ketbra{11}{11}_{\Ss\A}
\end{split}	
\end{align}

This describes the evolution of Alice's lab when using the projection postulate without considering specific outcomes, such as if Alice forgets the outcome after measurement.

Thus, treating a measurement as producing a classical record or as a purely unitary evolution yields two distinct trace-preserving channels for the same initial state. 
Wigner's thought experiment shows that this ambiguity can have observable consequences if quantum theory is universally valid (as also emphasized in the analysis of \cite{Baumann_2018}).

Consider an outside agent, Wigner, who has full quantum control over $\Ss\A$ (Alice's lab). Such a Wigner, a ``super-agent'' can perform arbitrary quantum operations on $\Ss\A$. Since these evolutions result in distinct states on $\Ss\A$, Wigner can operationally distinguish them by measuring $\Ss\A$ in a suitable basis. That is, the probability of Wigner's measurement outcome $w$, will generally depend on how the friend, Alice's measurement is modelled: using the unitarity (\cref{eq:M_unitary}) or the projection postulate (\cref{eq:M_mixture}). This is illustrated in \cref{fig:WFexample_augcirc}, which also summarises how the scenario would be modelled in our framework as an augmented circuit which incorporates this choice in how the Friend's measurement is modelled (which can be seen as a formalisation of the Heisenberg cut \Cref{sec:Discussion}).

Several extensions of Wigner’s original thought experiment, including Deutsch \cite{Deutsch1985}, Frauchiger-Renner (FR) \cite{Frauchiger2018}, and Local-Friendliness (LF) \cite{Bong2020}, add further agents and/or information-processing steps. These works suggest further radical features of quantum theory exemplified by Wigner’s Friend Scenarios; while FR have claimed that logical contradictions arise, the LF theorem presents a case for a fundamental break down of absoluteness of observed events in such scenarios, despite our experience of absolute observations in existing physical experiments. In the next sections, we develop a comprehensive framework that can describe arbitrary EWFSs in quantum theory, apply it to resolve all FR type paradoxes as well as address the emergence of absolute events.

\section{A general circuit framework for EWFS}
\label{sec:gen_framework}

We now develop our general circuit framework for EWFS that meets all the desiderata motivated in the introduction. We do so by carefully distinguishing between and connecting: the predictions of quantum theory, agents' knowledge and their statements. This is performed while defining what it means for a theoretical model to make consistent predictions in an EWFS.

\subsection{Extended Wigner's friend scenarios and quantum predictions}

The term Extended Wigner's Friend Scenario (EWFS) is commonly used for cases such as Frauchiger-Renner and Local-Friendliness, but a general definition has been lacking. We define EWFSs as finite multi-agent quantum protocols in which agents' memories, or equivalently labs, are modelled as quantum systems and one agent can have full quantum control over another's lab.
 The formal definition is provided below, and its generality is justified in \cref{appendix: generality_EWFS}.

\begin{definition}[Extended Wigner's Friend Scenario (EWFS)]
\label{def:LWFS}
    An EWFS is a quantum protocol that consists of
    \begin{enumerate}[left=0pt]
        \item A finite set $\mathtt{S}:=\{\Ss_1,...,\Ss_m\}$ of systems,
        \item A finite set $\mathtt{A}:=\{\A_1,...,\A_N\}$ of agents,
        \item A set $\mathtt{M}:=\{\M_1,...,\M_N\}$ of memory systems, one for each agent $\A_i$ where they store the outcome of their measurement.
        \item For each agent $\A_i$, a subset $\mathtt{S}_i\subseteq \mathtt{S}\cup \mathtt{M} \backslash \{\M_i\}$ of systems (which can include the memories of other agents) that they measure at time $t_i$ according to a measurement $\matholdcal{M}^{\A_i}$, obtaining a measurement outcome $a_i$ that they store in their memory system $\M_i$. Moreover, $t_i<t_j$ for $i<j$.
        \item For each agent $\A_i$, a finite set $\mathtt{O}_i:=\{0,1,...,d_{\mathtt{S}_i}-1\}$ in which the value $\aaa_i$ of their outcome $a_i$ belongs.
        \item For each agent $\A_i$, a fixed operation $\matholdcal{E}^i$ acting on their measured system and memory $\mathtt{S}_i\cup \{\M_i\}$ that captures the possibility of performing a further operation on the system $\mathtt{S}_i$ after the measurement, depending on the measurement outcome. 
        \item A joint initial state $\rho_{\Ss_1,...,\Ss_m}$ of all the systems $\mathtt{S}$ at time $t_0<t_i$ for all $i\in\{1,...,N\}$
    \end{enumerate}
   The protocol takes the form of \cref{fig: genform_circuit}, and we can consider projective measurements $\matholdcal{M}^{\A_i}=\{\pi_{\aaa_i}^{\mathtt{S}_i}=\proj{\aaa_i}_{\mathtt{S}_i}\}_{\aaa_i\in\mathtt{O}_i}$ without loss of generality.
\end{definition}

The general form of an EWFS is illustrated in \cref{fig: genform_circuit}. It is useful to note that we can assume without loss of generality that each agent performs only one measurement. Any scenario where an agent performs multiple measurements can be transformed into this form by modeling it as multiple agents, each performing a single measurement. Similarly, scenarios where an agent can choose between multiple measurements can be modelled as a single measurement by encoding the measurement choice in an initial state, as we show explicitly in an upcoming work \cite{LF_Vilasini_Woods} for the LF scenario, which includes different measurement choices.

\begin{figure*}[ht!]
    \centering
    \begin{tikzpicture}
    \draw[->] (6,-8)--(9,-8); \node at (9.5,-8) {time};
    
\draw [ultra thick,decorate,
    decoration = {calligraphic brace}] (-0.2,-2.2)--(-0.2,0.2); \node at (-1,-1) {$\rho_{\Ss_1,\ldots,\Ss_m}$}; \node at (-0.5,-3) {$\ket{0}_{\M_1}$};  \node at (-0.5,-4) {$\ket{0}_{\M_2}$};  \node at (-0.5,-5) {$\ket{0}_{\M_3}$};  \node at (-0.5,-7) {$\ket{0}_{\M_N}$};

    \draw (0,0)--(1,0); \node at (0.5,0.2) {$\Ss_1$}; \draw (0,-2)--(1,-2); \node at (0.5,-1.8) {$\Ss_m$}; \node at (0.5,-1) {$\cdot$}; \node at (0.5,-1.3) {$\cdot$}; \node at (0.5,-0.7) {$\cdot$}; \draw (1,0.3) rectangle node{$\matholdcal{M}^{\A_1}$} (2,-3.2);
     \draw (0,-3)--(1,-3); \node at (0.5,-2.8) {$\M_1$};
     
     \draw (2,0)--(3,0); \node at (2.5,0.2) {$\Ss_1$}; \draw (2,-2)--(3,-2); \node at (2.5,-1.8) {$\Ss_m$}; \node at (2.5,-1) {$\cdot$}; \node at (2.5,-1.3) {$\cdot$};
      \draw (2,-3)--(3,-3); \node at (2.5,-2.8) {$\M_1$}; \draw (3,0.3) rectangle node{$\matholdcal{E}^1$} (4,-3.2);
      
       \draw (4,0)--(5,0); \node at (4.5,0.2) {$\Ss_1$}; \draw (4,-2)--(5,-2); \node at (4.5,-1.8) {$\Ss_m$}; \node at (4.5,-1) {$\cdot$}; \node at (4.5,-1.3) {$\cdot$};
      \node at (4.5,-0.7) {$\cdot$}; \draw (4,-3)--(5,-3); \node at (4.5,-2.8) {$\M_1$}; 
      \draw (0,-4)--(5,-4); \node at (2.5,-3.8) {$\M_2$}; \draw (5,0.3) rectangle node{$\matholdcal{M}^{\A_2}$} (6,-4.2);
      
       \draw (6,0)--(7,0); \node at (6.5,0.2) {$\Ss_1$}; \draw (6,-2)--(7,-2); \node at (6.5,-1.8) {$\Ss_m$}; \node at (6.5,-1) {$\cdot$}; \node at (6.5,-1.3) {$\cdot$}; \node at (6.5,-0.7) {$\cdot$};
      \draw (6,-3)--(7,-3); \node at (6.5,-2.8) {$\M_1$};  \draw (6,-4)--(7,-4); \node at (6.5,-3.8) {$\M_2$}; \draw (7,0.3) rectangle node{$\matholdcal{E}^2$} (8,-4.2);
      
       \draw (8,0)--(9,0); \node at (8.5,0.2) {$\Ss_1$}; \draw (8,-2)--(9,-2); \node at (8.5,-1.8) {$\Ss_m$}; \node at (8.5,-1) {$\cdot$}; \node at (8.5,-1.3) {$\cdot$}; \node at (8.5,-0.7) {$\cdot$};
      \draw (8,-3)--(9,-3); \node at (8.5,-2.8) {$\M_1$};  \draw (8,-4)--(9,-4); \node at (8.5,-3.8) {$\M_2$}; \draw (0,-5)--(9,-5);\node at (4.5,-4.8) {$\M_3$};

       \draw (11,0)--(12,0); \node at (11.5,0.2) {$\Ss_1$}; \draw (11,-2)--(12,-2); \node at (11.5,-1.8) {$\Ss_m$}; \node at (11.5,-1) {$\cdot$}; \node at (11.5,-1.3) {$\cdot$}; \node at (11.5,-0.7) {$\cdot$};
      \draw (11,-3)--(12,-3); \node at (11.5,-2.8) {$\M_1$};  \draw (11,-4)--(12,-4); \node at (11.5,-3.8) {$\M_2$}; \draw (11,-5)--(12,-5); \node at (11.5,-4.8) {$\M_3$}; 
      
     \node at (10,-3.5) {$\dots$};
      
      \node at (0.5,-6) {$\cdot$}; \node at (0.5,-6.3) {$\cdot$}; \node at (0.5,-5.7) {$\cdot$};
      \draw (0,-7)--(12,-7); \node at (5.5,-6.8) {$\M_N$};  \node at (11.5,-6) {$\cdot$}; \node at (11.5,-6.3) {$\cdot$}; \node at (11.5,-5.7) {$\cdot$};
      \draw (12,0.3) rectangle node{$\matholdcal{M}^{\A_N}$}(13,-7.3);
      
       \draw (13,0)--(14,0); \node at (13.5,0.2) {$\Ss_1$}; \draw (13,-2)--(14,-2); \node at (13.5,-1.8) {$\Ss_m$}; \node at (13.5,-1) {$\cdot$}; \node at (13.5,-1.3) {$\cdot$}; \node at (13.5,-0.7) {$\cdot$};
      \draw (13,-3)--(14,-3); \node at (13.5,-2.8) {$\M_1$};  \draw (13,-4)--(14,-4); \node at (13.5,-3.8) {$\M_2$}; \draw (13,-5)--(14,-5); \node at (13.5,-4.8) {$\M_3$}; 
      \node at (13.5,-6) {$\cdot$}; \node at (13.5,-6.3) {$\cdot$}; \node at (13.5,-5.7) {$\cdot$}; \draw (13,-7)--(14,-7); \node at (13.5,-6.8) {$\M_N$};  \draw (14,0.3) rectangle node{$\matholdcal{E}^N$}(15,-7.3);
    \end{tikzpicture}
    \caption{General form of an EWFS involving $N$ agents and $m$ systems described in the main text. All agents agree on the initial state $\rho_{\Ss_1,\ldots,\Ss_m}$ of the systems and initialise their memories to $\ket{0}$. Each agent $\A_i$ performs a measurement on some subset of the systems and on a subset of the memories of all agents who acted before, and they store the outcome of the measurement in their own memory $\M_i$. After the measurement, each agent may perform a fixed transformation $\matholdcal{E}^i$ (that is previously agreed upon by all agents) before the next agent measures. However, the agents need not necessarily agree on how each measurement is modelled, depending on their perspective some agents may describe a measurement as a purely unitary evolution as in \cref{eq:M_unitary_gen} while others may describe the same measurement as a decoherent process by assigning projectors as in \cref{eq:M_mixture_gen_proj}. Due to this ambiguity how a measurement is modelled, this diagram is not a fully specified quantum circuit.}
        \label{fig: genform_circuit}
\end{figure*}

In an EWFS, the unitary description $\matholdcal{M}^{\A_i}_{unitary}$ of a measurement $\matholdcal{M}^{\A_i}$ corresponds to a CNOT from $\mathtt{S}_i$ to the memory $\M_i$ (chosen to be of appropriate dimensions) in the basis $\{\ket{\aaa_i}_{\mathtt{S}_i}\}_{\aaa_i \in \mathtt{O}_i}$ of the measurement. For any state $\ket{\psi}_{\mathtt{S}_i} = \sum_{\aaa_i} c_{\aaa_i} \ket{\aaa_i}_{\mathtt{S}_i}$ expressed in this basis, we have \cite{Note1}

\begin{equation}
\label{eq:M_unitary_gen}
    \text{$\matholdcal{M}$}^{\A_i}_{unitary}: \ket{\psi}_{\text{$\mathtt{S}_i$}}\otimes \ket{0}_{\M_i} \mapsto \sum_{\aaa_i} c_{\aaa_i}\ket{\aaa_i\aaa_i}_{\mathtt{S}_i \M_i}.
\end{equation}

As highlighted by Wigner's thought experiment, in Wigner's Friend Scenarios (WFS), the ambiguity in modeling a measurement—whether through unitarity or the projection postulate—leads to observably different predictions. Quantum theory itself does not offer a clear set of rules to resolve this ambiguity or determine the correct predictions in an EWFS. Consequently, various sets of rules can be proposed for making predictions or statements about observed outcomes in a given EWFS. Different interpretations of quantum theory yield different predictions for observable correlations in an EWFS, even though they agree on predictions for currently realisable quantum experiments.

Given this ambiguity and the potential for multiple approaches to making predictions and reasoning in an EWFS, we provide general definitions of predictions and statements. These definitions will offer a unified framework for discussing the wide range of previous results and responses to EWFS arguments, and their relation to our main results. In the following, whenever we have a subset $\matholdcal{M}^{\A_{j_1}}$,...,$\matholdcal{M}^{\A_{j_p}}$ of measurements, it will be useful to denote the corresponding set of outcomes and set of values using vectors
\begin{align}
    \begin{split}
       \vec{a}_j&:=a_{j_1},...,a_{j_p}\\
        \vec{\aaa}_j&:=\aaa_{j_1},...,\aaa_{j_p}.
    \end{split}
\end{align}

Then a value assignment to a set of outcomes, such as $a_{j_1}=\aaa_{j_1}$,...,$a_{j_p}=\aaa_{j_p}$ becomes $\vec{a}_j=\vec{\aaa}_j$.

\begin{definition}[Prediction and scenario parameters]
\label{definition: prediction}

Consider an EWFS along with a set $\matholdcal{R}$ of rules for computing predictions in the scenario. A prediction in the EWFS is conditional probability $P(\vec{a}_j=\vec{\aaa}_j|\vec{a}_l=\vec{\aaa}_l,k=\kkk)$, where $\vec{a}_j$ and $\vec{a}_l$ represent the set of outcomes associated with any two disjoint subsets $\matholdcal{M}^{\A_{j_1}}$,...,$\matholdcal{M}^{\A_{j_p}}$ and $\matholdcal{M}^{\A_{l_1}}$,...,$\matholdcal{M}^{\A_{l_q}}$ (latter possibly empty) of measurements in the EWFS, and $k$ is a (possibly empty) set of random variables whose values $\kkk$ encode additional information about the scenario (the exact description of which is to be specified by the rules $\matholdcal{R}$). 
Whenever $P(\vec{a}_j=\vec{\aaa}_j|\vec{a}_l=\vec{\aaa}_l,k=\kkk)$ takes values in $\{0,1\}$, we refer to it as a \emph{logical prediction}.

\end{definition}

An example of the $k$ parameters would be the systems, states, transformations and measurements that one must generally condition on when computing a probability of a measurement outcome e.g., if a state $\rho_S$ is prepared on a system $S$, evolved by a transformation $\matholdcal{U}$ and measured using a measurement $\matholdcal{M}$ to obtain an outcome $\aaa$, then the relevant outcome probability is a conditional probability $P(\aaa|\rho_S,\matholdcal{U},\matholdcal{M})$ . We will see more examples later.

\begin{remark}
In probability theory, a conditional probability is by definition ``well-defined''. This means it is uniquely determined by the event space and is normalised, ensuring no contradictions arise when applying them. In the above definition, however, we do not require conditional probabilities to be uniquely defined based on the rules for calculating them (e.g., if the rules are ambiguous) or to sum to one (e.g., if the rules are inconsistent). We merely assume they are numbers in $[0,1]$, which can lead to ``paradoxes" under rules that do not yield correctly normalised probabilities. We will see examples where the rules produce both well-defined and not well-defined conditional probabilities.
\end{remark}

\begin{definition}[Statements associated with predictions]
\label{definition:statement_prediction}

Every prediction $P(\vec{a}_j=\vec{\aaa}_j|\vec{a}_l=\vec{\aaa}_l,k=\kkk)$ can be associated with a corresponding statement.\\

``If the outcomes $\vec{a}_l$ take values $\vec{\aaa}_l$ and the additional parameters of the scenario take the value $k=\kkk$, then the outcomes $\vec{a}_l$ take values $\vec{\aaa}_l$ with a probability $P(\vec{a}_j=\vec{\aaa}_j|\vec{a}_l=\vec{\aaa}_l,k=\kkk)$.''\\

 The set of all statements associated with predictions made using a set of reasoning rules in a given EWFS will be denoted as $\Sigma$ (when the scenario and rules are evident from context).

\end{definition}

\begin{definition}[Logical statements]
\label{def: logic_statement_gen}
Statements associated with logical predictions are called \emph{logical statements}.
When we have a logical prediction $P(\vec{a}_j=\vec{\aaa}_j|\vec{a}_l=\vec{\aaa}_l,k=\kkk)=1$ (or $=0$), the corresponding statement has the same conditional part as above, and the latter part of the statement becomes ``...,then it is certain that the outcomes $\vec{a}_l$ (do not) take values $\vec{\aaa}_l$.''
For logical predictions, we can also express the statements using logical operators $\land$ (and), $\Rightarrow$ (implies) and $\neg$ (negation). Specifically, the statements associated with $P(\vec{a}_j=\vec{\aaa}_j|\vec{a}_l=\vec{\aaa}_l,k=\kkk)=0$ and $P(\vec{a}_j=\vec{\aaa}_j|\vec{a}_l=\vec{\aaa}_l,k=\kkk)=1$ would respectively be
\begin{align}
\label{eq: logic_statement}
   \begin{split}
       \vec{a}_l=\vec{\aaa}_l \land k&=\kkk \Rightarrow \neg (\vec{a}_j=\vec{\aaa}_j),\\
       \vec{a}_l=\vec{\aaa}_l \land k&=\kkk \Rightarrow \vec{a}_j=\vec{\aaa}_j.
   \end{split}
\end{align}   
The set $\Sigma_L$ of all such logical statements of an EWFS is a subset of $\Sigma$.
\end{definition}

In examples, we may have $\vec{a}_l$ and $k$ being empty sets, for instance we can have a prediction $P(a_i=\aaa_i)=1$. Then the associated logical statement is simply $a_i=\aaa_i$ or in words ``it is certain that the outcome $a_i$ takes the value $\aaa_i$.''

We then have the following definition of consistency for any set of predictive statements.

\begin{definition}[Consistency for $\Sigma$]
\label{def: consistency_pred}
A set $\Sigma$ of statements obtained in an EWFS from a set of reasoning rules is said to be consistent iff $\Sigma$ contains no pairs of statements $S$ and $S'$ associated with predictions $P(\vec{a}_j=\vec{\aaa}_j|\vec{a}_l=\vec{\aaa}_l,k=\kkk)$ and $P'(\vec{a}_j=\vec{\aaa}_j|\vec{a}_l=\vec{\aaa}_l,k=\kkk)$ respectively where $P\neq P'$. 
    
\end{definition}

\begin{restatable}{lemma}{ConsistencyBasic}
 If $\Sigma$  is a set of consistent predictive statements, then 
 \begin{equation}
     S\in \Sigma \quad \Rightarrow \quad \neg S\cap \Sigma =\emptyset,
 \end{equation}
 where $\neg S$ denotes the negation of the statement $S$.
\end{restatable}

There are two distinct ways of being certain of an outcome value, illustrated by the following simple example.
\begin{itemize}[left=0pt]
    \item {\bf Scenario 1:} Alice has a fair coin yielding outcome $c$ taking values in $\{heads,tails\}$. She flips the coin and observes $c=heads$ in a round and is then certain that $c=heads$. 
    \item {\bf Scenario 2:} Alice has a biased coin that reads heads on both sides, $P(c=heads)=1$ and she is certain that $c=heads$ in every round (without observing the outcome). 

\end{itemize}

Clearly, these scenarios are different. To ensure reliability and consistency when communicating with other agents, it is necessary to distinguish these cases. Our previous definition of statements associated with predictions only covers Scenario 2. To cover Scenario 1, we need to include a set of statements associated with observations that agents make in a given experimental run, denoted as $\Sigma_{obs}$. We define this case in \Cref{appendix: observational_statements} and also discuss a more refined definition of consistency that accounts for both predictive and observational statements there. For all the main results and discussions, it will be sufficient to consider the predictive statements $\Sigma$ and the generalisation of the results to include $\Sigma_{obs}$ is presented in the appendix. This is because a theory only permits observation of certain outcomes when its predictions assign a non-zero probability to that outcome, thus the main features of the theory can be understood by studying its predictions and associated statements as defined here.

\subsection{Formulating an EWFS within a single quantum circuit}
\label{sec: mapping_augcircuit}

As stressed before, in EWFS, there is ambiguity in applying the postulates of quantum theory to compute probabilities. The postulates can be applied in different ways (using chosen ``rules'') to obtain different probabilities for the same outcomes, while distinct rules could still lead to the same probabilities in certain cases. In previous literature on EWFS, there is a common pattern in the probabilities used in EWFS arguments, such as in the FR and LF results. However, the rules for arriving at these probabilities are often not explicitly specified.

We consider the probabilities conventionally used in previous literature when they refer to ``predictions of quantum theory"  in the context of EWFS, where it is assumed that unitary evolutions can be applied to agents' labs \cite{Note2}. These probabilities typically do not condition on additional variables $k$ or channels in the circuit, their expressions only refer to the outcomes, e.g., $P(a=0,b=1)$ denotes the probability of the measurement outcomes $a=0$ and $b=1$ (which can be associated with agents Alice and Bob) in a given scenario. Below we formalise one way to arrive at these type of probability expressions in an EWFS, which recovers the probabilities used in the FR and LF scenarios.

\begin{definition}[Conventional predictions in an EWFS]
\label{def: conv_prediction}
A conventional prediction in an EWFS is a conditional probability $P_{conv}(\vec{a}_j=\vec{\aaa}_j|\vec{a}_l=\vec{\aaa}_l)$ evaluated by applying the following rules.
\begin{enumerate}[left=0pt]
    \item In the given EWFS (see also \cref{fig: genform_circuit}), all the measurements $\matholdcal{M}^{\A_i}\not\in \{\matholdcal{M}^{\A_{j_1}},...,\matholdcal{M}^{\A_{j_p}}\} \cup \{\matholdcal{M}^{\A_{l_1}},...,\matholdcal{M}^{\A_{l_q}}\}$ (those not appearing in the prediction) are modelled as unitary evolutions of their labs (according to the channel \cref{eq:M_unitary_gen}). 
    \item For each  $\matholdcal{M}^{\A_i}\in \{\matholdcal{M}^{\A_{j_1}},...,\matholdcal{M}^{\A_{j_p}}\} \cup \{\matholdcal{M}^{\A_{l_1}},...,\matholdcal{M}^{\A_{l_q}}\}$, the projective measurement $\{\pi_{\aaa_i}^{\mathtt{S}_i}=\proj{\aaa_i}_{\mathtt{S}_i}\}_{\aaa_i\in\mathtt{O}_i}$ is applied on the corresponding system $\mathtt{S}_i$, followed by a CNOT in the same basis from the system $\mathtt{S}_i$ to the memory $\M_i$ which models the procedure of making the measurement on the system and storing the outcome in the memory.
    \item This fixes all the channels in the circuit, and the joint probability $P_{conv}(\vec{a}_j=\vec{\aaa}_j,\vec{a}_l=\vec{\aaa}_l)$ is then calculated by applying the Born rule to this circuit, using the measurement projectors given above \cite{Note3}.
    \item The prediction $P_{conv}(\vec{a}_j=\vec{\aaa}_j|\vec{a}_l=\vec{\aaa}_l)$ is then obtained through the usual conditional probability rule.
\end{enumerate}

\end{definition}

A conventional prediction corresponds to the case where $k$ is the empty set, meaning the prediction is not conditioned on any additional information about the scenario (such as states and channels used). Such predictions often lead to FR-type paradoxes, highlighting that the above rules for computing probabilities in EWFS are not generally consistent and do not yield valid joint probabilities for measurement outcomes in EWFSs.

In standard quantum theory outside the context of WFS, where agents' labs are not typically modelled as quantum systems in a circuit, all predictions for a given experiment are defined relative to a single circuit, ensuring consistent and well-defined joint probabilities. In EWFS, it is initially unclear which circuit is associated with a given scenario due to the ambiguity in the channels modelling the measurements (as discussed in \Cref{sec:background}), which has observable consequences in these scenarios. We now explain how every EWFS can be mapped to a single quantum circuit which we call the \emph{augmented circuit} of the EWFS, from which all the predictions of the EWFS can be computed and used consistently, while allowing measurements to be modelled as unitary evolutions of agents' labs and also preserving the validity of the Born rule.

For this, recall that for a measurement $\text{$\matholdcal{M}$}^{\A_i}$ in a basis defined by the projectors $\{\pi_{\aaa_i}^{\mathtt{S}_i}=\proj{\aaa_i}_{\mathtt{S}_i}\}_{\aaa_i\in\mathtt{O}_i}$ acting on a state $\ket{\psi}_{\mathtt{S}_i}=\sum_{\aaa_i} c_{\aaa_i} \ket{\aaa_i}_{\mathtt{S}_i}$, the associated unitary evolution $\matholdcal{M}^{\A_i}_{unitary}$ is as given in \cref{eq:M_unitary_gen}. On the other hand, applying the projection postulate for the measurement on $\Ss_i$ and the fact that the outcome is copied to the memory $\A_i$, one would obtain the final state $\proj{\aaa_i\aaa_i}_{\mathtt{S}_i\M_i}$ when the outcome is $\aaa_i$. If we consider the view that such a classical outcome is obtained but one lacks knowledge of its value, we would obtain the following trace-preserving evolution.

\begin{equation}
\label{eq:M_mixture_gen}
 \text{$\matholdcal{M}$}^{\A_i}_{projection}(\proj{\psi}_{\mathtt{S}_i}\otimes \proj{0}_{\M_i})=\sum_{a_i} |c_{\aaa_i}|^2 \proj{\aaa_i\aaa_i}_{\mathtt{S}_i\M_i}.
\end{equation}

Notice that $\text{$\matholdcal{M}$}^{\A_i}_{projection}(\proj{\psi}_{\mathtt{S}_i}\otimes \proj{0}_{\M_i})$ above can be equivalently expressed as
\begin{equation}
\label{eq:M_mixture_gen_proj}
 \sum_{a_i} \pi_{\aaa_i}^{\mathtt{S}_i \M_i}  \Big( \text{$\matholdcal{M}$}^{\A_i}_{unitary}\big(\proj{\psi}_{\mathtt{S}_i}\otimes \proj{0}_{\M_i}\big)  \text{$\matholdcal{M}$}^{\A_i \dag}_{unitary}\Big) \pi_{\aaa_i}^{\mathtt{S}_i \M_i},
\end{equation}
where $\{\pi_{\aaa_i}^{\mathtt{S}_i \M_i}:=\proj{\aaa_i\aaa_i}_{\mathtt{S}_i \M_i}\}_{\aaa_i\in\mathtt{O}_i}$, and we still have a one-to-one correspondence between the measurement outcomes $\aaa_i$ and an element from the above set of projectors. Therefore the two possible (trace-preserving) evolutions associated with a measurement $\text{$\matholdcal{M}$}^{\A_i}$ are now given by \cref{eq:M_mixture_gen_proj,eq:M_unitary_gen}. They differ is whether or not one views the measurement as having produced classical records.

\begin{figure*}[ht!]
    \centering
    \begin{tikzpicture}[scale=0.9]
    \draw[->] (6,-6)--(9,-6); \node at (9.5,-6) {time};
    
\draw [ultra thick,decorate,
    decoration = {calligraphic brace}] (-0.2,-2.2)--(-0.2,0.2); \node at (-1,-1) {$\rho_{\Ss_1,\ldots,\Ss_m}$}; \node at (-0.5,-3) {$\ket{0}_{\M_1}$};  \node at (-0.5,-5) {$\ket{0}_{\M_N}$};  \node at (0.5,-4) {$\cdot$}; \node at (0.5,-4.3) {$\cdot$}; \node at (0.5,-3.7) {$\cdot$};

    \draw (0,0)--(1.3,0); \node at (0.5,0.2) {$\Ss_1$}; \draw (0,-2)--(1.3,-2); \node at (0.5,-1.8) {$\Ss_m$}; \node at (0.5,-1) {$\cdot$}; \node at (0.5,-1.3) {$\cdot$}; \node at (0.5,-0.7) {$\cdot$};  \draw (0,-3)--(1.3,-3);\node at (0.5,-2.8) {$\M_1$};
    \draw[dashed] (1,1.4) rectangle  (5,-3.5); \draw (1.3,0.2) rectangle node{$\matholdcal{M}^{\A_1}_{unitary}$} (2.8,-3.2);
     \draw (0,-3)--(1,-3); \draw[blue] (3.8,1.2) rectangle node{$\mathtt{\Pi}^{\A_1}_{x_1}$} (4.8,-3.2); \draw[blue] (0,1)--(3.8,1); \node[blue] at (-0.9,1) {$x_1\in\{0,1\}$}; \draw[blue] (4.8,1)--(5.8,1); \node[blue] at (6.5,1.3) {$a_1\in \mathtt{O}_1 \cup \perp$};
    \node at (3,1.8) {$\matholdcal{M}^{\A_1}$};

     \draw (2.8,0)--(3.8,0); \node at (3.3,0.2) {$\Ss_1$}; \draw (2.8,-2)--(3.8,-2); \node at (3.3,-1.8) {$\Ss_m$}; \node at (3.3,-1) {$\cdot$}; \node at (3.3,-1.3) {$\cdot$};    \node at (3.3,-0.7) {$\cdot$}; \draw (2.8,-3)--(3.8,-3); \node at (3.3,-2.8) {$\M_1$};

       \draw (4.8,0)--(6,0); \node at (5.5,0.2) {$\Ss_1$}; \draw (4.8,-2)--(6,-2); \node at (5.5,-1.8) {$\Ss_m$}; \node at (5.5,-1) {$\cdot$}; \node at (5.5,-1.3) {$\cdot$};
      \node at (5.5,-0.7) {$\cdot$}; \draw (4.8,-3)--(6,-3); \node at (5.5,-2.8) {$\M_1$}; \draw (6,0.3) rectangle node{$\matholdcal{E}^1$} (7,-3.2);
      
      \draw (7,0)--(8,0); \node at (7.5,0.2) {$\Ss_1$}; \draw (7,-2)--(8,-2); \node at (7.5,-1.8) {$\Ss_m$}; \node at (7.5,-1) {$\cdot$}; \node at (7.5,-1.3) {$\cdot$};
      \node at (7.5,-0.7) {$\cdot$}; \draw (7,-3)--(8,-3); \node at (7.5,-2.8) {$\M_1$};

    \node at (9,-1.5) {$\dots$};  

\draw (0,-5)--(11.3,-5); \node at (5.5,-4.8) {$\M_N$};  \draw (10,0)--(11.3,0); \node at (10.5,0.2) {$\Ss_1$}; \draw (10,-2)--(11.3,-2); \node at (10.5,-1.8) {$\Ss_m$}; \node at (10.5,-1) {$\cdot$}; \node at (10.5,-1.3) {$\cdot$};
      \node at (10.5,-0.7) {$\cdot$}; \draw (10,-3)--(11.3,-3); \node at (10.5,-2.8) {$\M_1$}; \node at (10.5,-4.3) {$\cdot$}; \node at (10.5,-3.7) {$\cdot$}; \node at (10.5,-4) {$\cdot$};

    \draw[dashed] (11,2.9) rectangle  (15,-5.5); \draw (11.3,0.2) rectangle node{$\matholdcal{M}^{\A_N}_{unitary}$} (12.8,-5.2);
\draw[blue] (13.8,2.7) rectangle node{$\mathtt{\Pi}^{\A_N}_{x_N}$} (14.8,-5.2); \draw[blue] (0,2.5)--(13.8,2.5); \node[blue] at (-0.9,2.5) {$x_N\in\{0,1\}$}; \draw[blue] (14.8,2.5)--(15.8,2.5); \node[blue] at (16.5,2.8) {$a_N\in \mathtt{O}_N \cup \perp$};
    \node at (13,3.3) {$\matholdcal{M}^{\A_N}$};
    
      \draw (12.8,0)--(13.8,0); \node at (13.3,0.2) {$\Ss_1$}; \draw (12.8,-2)--(13.8,-2); \node at (13.3,-1.8) {$\Ss_m$}; \node at (13.3,-1) {$\cdot$}; \node at (13.3,-1.3) {$\cdot$}; \node at (13.3,-4.3) {$\cdot$}; \node at (13.3,-3.7) {$\cdot$}; \node at (13.3,-4) {$\cdot$};

      \node at (13.3,-0.7) {$\cdot$}; \draw (12.8,-3)--(13.8,-3); \node at (13.3,-2.8) {$\M_1$}; \draw (12.8,-5)--(13.8,-5); \node at (13.3,-4.8) {$\M_N$};
      
       \draw (14.8,0)--(16,0); \node at (15.5,0.2) {$\Ss_1$}; \draw (14.8,-2)--(16,-2); \node at (15.5,-1.8) {$\Ss_m$}; \node at (15.5,-1) {$\cdot$}; \node at (15.5,-1.3) {$\cdot$}; \node at (15.5,-4.3) {$\cdot$}; \node at (15.5,-3.7) {$\cdot$}; \node at (15.5,-4) {$\cdot$};
      \node at (15.5,-0.7) {$\cdot$}; \draw (14.8,-3)--(16,-3); \node at (15.5,-2.8) {$\M_1$}; \draw (14.8,-5)--(16,-5); \node at (15.5,-4.8) {$\M_N$}; \draw (16,0.3) rectangle node{$\matholdcal{E}^N$} (17,-5.2);
    
    \node[blue] at (-0.8,1.75) {$\cdot$};  \node[blue] at (-0.8,2.05) {$\cdot$};  \node[blue] at (-0.8,1.45) {$\cdot$};
    \end{tikzpicture}
    \caption{Augmented circuit for the general form of an EWFS illustrated in \cref{fig: genform_circuit}
    that makes explicit the implicit setting choices needed to model the measurement of each agent. This circuit makes it clear that each agent has a choice in how they describe each measurement in the scenario, when the setting $x_i=0$, the corresponding measurement $\text{$\matholdcal{M}$}^{\A_i}$ is modelled as a unitary evolution $\text{$\matholdcal{M}$}^{\A_i}_{unitary}$ (in this case the projector in the blue box implements an identity operation and $a_i=\perp$ deterministically) and when $x_i=1$, the same measurement is modelled as the unitary $\text{$\matholdcal{M}$}^{\A_i}_{unitary}$ followed by non-trivial projectors that identify the classical measurement outcome $a_i\in\{0,1,...d_{\mathtt{S}_i}-1\}:=\mathtt{O}_i$. In order for any agent to reason about the measurement outcome $a_i$ of an agent $\A_i$, they must necessarily choose $x_i=1$ in order to identify and calculate probabilities for the classical outcome that they are reasoning about. They may however choose to model all other agents $\A_j$ unitarily by choosing $x_j=0$.}
      \label{fig: genform_circuit_settings}
\end{figure*}

We explicitly account for this choice by introducing a classical binary variable for each measurement that takes values $x_i\in\{0,1\}$ (which we call the \emph{setting}), such that $x_i=0$ and $x_i=1$ correspond to the evolutions \cref{eq:M_unitary_gen,eq:M_mixture_gen_proj} respectively.

This is captured through the following enlarged set of projectors, obtained by appending a trivial outcome value $\perp$ to the outcome set $\mathtt{O}_i$.

\begin{equation}
\label{eq: proj_settings}
\mathtt{\Pi}^{\A_i}_{x_i}:=\begin{cases}
\{\pi_{\aaa_i,x_i}^{\A_i}\}_{\aaa_i\in\{ \perp\}}= \{\pi_{\perp,0}^{\A_i}:=\id_{\mathtt{S}_i \M_i}\}  &\text{ if } x_i=0\\
\{\pi_{\aaa_i,x_i}^{\A_i}\}_{\aaa_i\in\mathtt{O}_i}= \{\pi_{\aaa_i,1}^{\A_i}:=\pi_{\aaa_i}^{\mathtt{S}_i \M_i}\}_{\aaa_i\in\mathtt{O}_i} &\text{ if } x_i=1.
\end{cases}
\end{equation}
That is, for $x_i=0$ the value of the outcome $a_i$ is fixed uniquely to $a_i=\perp$ and the corresponding projector is the identity operator $\id_{\mathtt{S}_i \M_i}$.
Meanwhile, for $x_i=1$, $a_i$ takes values $\aaa_i$ in the set $\mathtt{O}_i$ (reflecting the different possible classical outcomes of the measurement, see \Cref{def:LWFS}). We will only use the terminology \emph{outcome}, to refer to the value of $a_i$ when $x_i=1$. In the case $x_i=0$, $a_i=\perp$ and we will not regard this as an outcome (we will sometimes refer to this case as the \emph{trivial outcome}).

 Then we can model each measurement $\text{$\matholdcal{M}$}^{\A_i}$ as the corresponding unitary $\text{$\matholdcal{M}$}^{\A_i}_{unitary}$, followed by a (setting-dependent) projective measurement $\mathtt{\Pi}_{x_i}^{\A_i}$. The pure unitary picture is always recovered by setting $x_i=0$, in which case only $\text{$\matholdcal{M}$}^{\A_i}_{unitary}$ \phantom{$\text{$\matholdcal{M}$}^{{\A_i}^A}_{{unitary}_y}$ }\hspace{-17mm}unitary is applied. The set of choices of settings for all $N$ agents can be represented by a vector $\vec{x}=(x_1,\ldots,x_N)$. With this explicit model, we can transform any EWFS into a standard temporally ordered quantum circuit parametrised by the settings $\vec x$ such that all the predictions made by different agents having different perspectives can be derived from the single circuit by fixing the settings $\vec x$ (as shown in \Cref{theorem: main}). We call this the \emph{augmented circuit} of the EWFS. This is illustrated in \cref{fig: genform_circuit_settings}, and summarised in the following definition.

\begin{definition}[Augmented circuit of an EWFS]
\label{def:LWFS_aug}
Given an EWFS of the form of \cref{fig: genform_circuit}, we can associate with it an augmented circuit of the form of \cref{fig: genform_circuit_settings} which is obtained from the EWFS by replacing each measurement $\text{$\matholdcal{M}$}^{\A_i}$ by the corresponding unitary description $\text{$\matholdcal{M}$}^{\A_i}_{unitary}$  (\cref{eq:M_unitary_gen}), followed by the setting-dependent enlarged set of projectors $\mathtt{\Pi}^{\A_i}_{x_i}$ (\cref{eq: proj_settings}). The setting takes binary values $x_i\in \{0,1\}$, where we have only the trivial outcome $\perp$ whenever $x_i=0$ and the non-trivial measurement outcome $\aaa_i\in \mathtt{O}_i$ whenever $x_i=1$. We will refer to this as an augmented EWFS, in short. 
\end{definition}

We now make a crucial observation. While we have so far given full freedom in choosing how to model the measurements, with both $x_i=0$ (pure unitary evolution) and $x_i=1$ (also assigning non-trivial projectors) being allowed, we notice that the form of the prediction being computed using the Born rule fixes some of these choices. For example, consider a prediction $P(a_i=\aaa_i|\vec{x})$. To compute this using the Born rule, we must apply the projector associated with the outcome $a_i=\aaa_i$, which is needed to identify that outcome. Therefore for this prediction, $x_i=1$. The other components of the setting vector $\vec{x}$ can generally vary according to the rules of reasoning being applied. 

When we say ``applying a projector'', this is not to be conflated with ``collapsing'' the state. Even when we do not explicitly refer to the post-measurement state of an agents' measurement, when reasoning about that agents' outcome using the Born rule, we apply knowledge of the measurement basis (given here by the projectors which identify the outcome). This is further discussed in \Cref{remark: no_collapse}, where we show how our framework can also be applied to EWFSs where the Born is rule is applied without need for a particular state update rule. The main observation above is that given a prediction, the measurements whose outcomes appear in that prediction are always regarded as producing classical records (associated with $x_i=1$ in our framework) when computing the prediction.

Applying this simple observation about quantum probabilities to agents' reasoning, we can consider an example. If $\A_i$ reasons about $\A_j$'s outcome based on their own outcome, they set $x_i=x_j=1$ in order to identify the classical outcomes they are reasoning about, but they may choose $x_k=0$ for all $k\neq i,j$ and model all other agents unitarily. 
How these remaining settings are to be chosen will need to be specified by a set of reasoning rules $\matholdcal{R}$, which we do not yet fix. This generality will become relevant when discussing the different responses to FR's no-go theorem (\Cref{appendix:relation_prev_works}).

Having defined the augmented circuit, we explain how one can compute predictions here. This can be done by applying the Born rule along with the conditional probability rule, as long as we are given a choice of initial values for the setting vector $\vec{x}$ (as this is required for the full specification of the circuit).

Different interpretations of quantum theory would generally propose different rules $\matholdcal{R}$ to fully specify $\vec{x}$ and can arrive at different predictions for the same EWFS. As we will see in \Cref{theorem: main}, formalising conventional predictions (\Cref{def: conv_prediction}) in our augmented circuit yields an explicit rule for choosing the settings which will be relative to a subset of systems of measurements being considered, and can therefore allow subjective choices of settings when we consider reasoning agents (see also \Cref{sec: reasoning_rules}).

Until then, all our results regarding the augmented circuit apply to all possible rules for choosing the settings  and hence apply to several interpretations of quantum theory (both relational and non-relational ones).

\begin{definition}[Setting-conditioned prediction]
\label{definition: setting_prediction}
Given an EWFS, a setting-conditioned prediction is a prediction $P(\vec{a}_j=\vec{\aaa}_j|\vec{a}_l=\vec{\aaa}_l,\vec{x}=\vec{\xxx})$ (\Cref{definition: prediction}) where the setting vector $\vec{x}$ represents the set $k$ of scenario parameters. A setting-conditioned prediction is computed by mapping the EWFS to an augmented circuit, then applying the quantum Born rule together with the rule for conditional probabilities for the given choice $\vec{x}=\vec{\xxx}$ of the settings. This is explicitly shown in \Cref{appendix: prob_rule_general}. Whenever $P(\vec{a}_j=\vec{\aaa}_j|\vec{a}_l=\vec{\aaa}_l,\vec{x}=\vec{\xxx})\in \{0,1\}$, we will refer to it as a \emph{logical setting-conditioned prediction}.

\end{definition}

\begin{definition}[Statements associated with setting-conditioned predictions]
\label{definition:logical_setting_prediction}

Every setting-conditioned prediction $P(\vec{a}_j=\vec{\aaa}_j|\vec{a}_l=\vec{\aaa}_l,\vec{x}=\vec{\xxx})$ can be associated with a corresponding statement, in the same way that \Cref{definition:statement_prediction} assigns statements to predictions $P(\vec{a}_j=\vec{\aaa}_j|\vec{a}_l=\vec{\aaa}_l,k=\kkk)$, but for setting-conditioned predictions, the scenario parameter values $k=\kkk$ are replaced with the setting values $\vec{x}=\vec{\xxx}$. Statements associated with logical setting-conditioned prediction can be expressed using logical operators as given in \Cref{def: logic_statement_gen} but with  $k=\kkk$ replaced by $\vec{x}=\vec{\xxx}$.

\end{definition}

\begin{definition}[Set of all statements in an augmented EWFS]
\label{def: all_aug_statements}
Consider the set of all setting-conditioned predictions that can be made in an augmented circuit of an EWFS, obtained under all possible setting choices. The set of all statements $\Sigma^{aug}$ in that EWFS is obtained by mapping each setting-conditioned prediction to a corresponding statement as per \Cref{definition:logical_setting_prediction}. Similarly, restricting to the set of all logical predictions, we have the corresponding set of all logical statements of the scenario which will be denoted as $\Sigma_L^{aug}\subseteq \Sigma^{aug}$.
\end{definition}
\begin{remark}
In setting-conditioned predictions, the setting variables are fixed to some specific value deterministically. More generally, we can consider arbitrary prior distribution over the settings $\vec{x}$ in our framework. This allows one to compute a probability distribution such as $P(\vec{a}_j=\vec{\aaa}_j)$ which does not explicitly feature any settings, which is obtained by choosing some prior $P(\vec{x})$ and averaging over the settings, $P(\vec{a}_j=\vec{\aaa}_j)=\sum_{\vec{\xxx}} P(\vec{a}_j=\vec{\aaa}_j|\vec{x}=\vec{\xxx})P(\vec{x}=\vec{\xxx})$.

We have established that for statements made using setting-conditioned predictions, the specific setting value assumed must be explicitly specified. Analogously, our framework and results can be generalised to accommodate arbitrary priors by ensuring that when using such priors to compute outcome probabilities or predictions, the corresponding statements about the measurement outcomes specify the chosen prior over the settings, as this is an additional choice that an agent must make in the reasoning process. However, we will not delve into the case of general prior distributions further, as it is not pertinent to the main results and insights of our work. 
\end{remark}

\section{Completeness, consistency and causality without absolute events}
\label{ssec: general_results}

In this section (and in \Cref{appendix: properties_general}), we formalise and prove several key properties of our framework that are pertinent to logical and causal reasoning in EWFS, without imposing an absolute notion of measurement events.

\subsection{Properties of general quantum predictions in EWFS}
\label{sec: consistency_pred}

We begin by emphasizing that our framework does not assume an absolute, objective notion of measurement events, unlike most existing quantum information frameworks. This is because it does not require a global joint probability distribution $P(a_1, \dots, a_N)$ over the non-trivial outcomes $a_i \in \mathtt{O}_i$ of all measurements. Instead, by introducing settings $\vec{x}$, which as discussed in \Cref{sec: setting_interpretation}, model choices of Heisenberg cuts, the framework allows outcome probabilities to be fundamentally relational in Wigner's Friend Scenarios. 

Consider an EWFS with two measurements $\matholdcal{M}^{\A_1}$ and $\matholdcal{M}^{\A_2}$. If the joint system $\Ss_1\M_1$ (modeling agent $\A_1$'s lab) is treated as recording classical outcomes, while $\Ss_2\M_2$ evolves unitarily, then $x_1 = 1$ and $x_2 = 0$, and we can compute probabilities only for $a_1 \in \mathtt{O}_1$. If both measurements are treated as classical, $x_1 = x_2 = 1$, enabling computation of joint probabilities for $a_1$ and $a_2$. These probabilities may differ, as different $x_2$ values are used. Nonetheless, both are derived from the same quantum circuit via different choices of the input variable $x_2$.

Although our formalism does not impose event absoluteness, its formulation as a well-defined quantum circuit (\Cref{sec:gen_framework}) immediately implies that several key properties underlying scientific reasoning and important physical principles remain valid in Wigner’s Friend scenarios. Full definitions and theorems are provided in \Cref{appendix: completeness_causality}, while the main ideas are outlined here at a more intuitive level.

In \Cref{appendix: completeness_causality} and in \Cref{theorem: main}, we formally define and prove that our augmented circuit framework respects the following three properties, implying that these properties can be rigorously ensured in any EWFS in quantum theory.

\begin{itemize}[left=0pt]
    \item {\bf Completeness} For any given EWFS, the corresponding augmented circuit uniquely specified by \Cref{sec:gen_framework} encodes all the conventional predictions of that EWFS: for each conventional prediction there exists a corresponding setting-conditioned prediction in the augmented circuit that is equal to it. 
    \item {\bf Consistency} For any EWFS, the set of all statements $\Sigma^{aug}$ obtained in the corresponding augmented circuit (\Cref{def: all_aug_statements}) are consistent according to \Cref{def: consistency_pred}.
    \item {\bf Causality} For any EWFS, the corresponding augmented circuit is such that the associated probabilities of any set of outcomes does not depend on setting choices made outside its past. Here the past is defined relative to the connectivity of the circuit which is consistent with the temporal order: one operation causally precedes another when there is a path in the circuit connecting an output of the first to an input the second, and this is only possible when the second appears later in time than the first (\Cref{definition: causal_str}).
\end{itemize}

These properties confirm that our framework preserves consistency in any EWFS without violating basic physical principles like causality (neither in the quantum information nor the relativistic sense of the word).

\subsection{Application to agents' reasoning}
\label{sec: consistency_reasoning}

So far our results have been about the properties of predictions (probabilities) that can be computed in an EWFS and showing that our formalism yields a complete, logical and causally consistent way to make predictive statements in such scenarios. We now apply our general formalism to reveal properties of agents' reasoning that hold in general EWFS. One immediate corollary of our general consistency result of \Cref{theorem: main}, for agents' reasoning is the following.
\begin{restatable}[]{corollary}{SameChoicesSamePredictions}
\label{lemma: samechoices_samepredictions}
If any two agents use the same choice of settings $\vec x$ for all measurements $\{\text{$\matholdcal{M}$}^{\A_i}\}_i$ in an augmented EWFS then they make all the same predictions in that scenario.
\end{restatable}

This corollary concerns a rather restricted case where all agents have a fixed and common choice of Heisenberg cut (here formalised through the settings). More generally, our framework allows agents to choose different settings depending on their perspective or prediction they wish to compute. Below we establish that our formalism allows quantum agents in an EWFS to reason consistently even when they apply typical axioms of classical logic to observed classical outcomes, the Born rule to compute probabilities, can model each others' labs as unitarily evolving quantum systems and have fundamentally subjective perspectives. In particular, we provide a formalisation of the assumptions $Q$, $U$, $C$, $D$ and $S$ used in FR's arguments \cite{Frauchiger2018, Nurgalieva2018} which FR have claimed to be inconsistent in a particular EWFS \cite{Note4}. 
Contrary to FR's conclusions, by formalising these assumptions in our framework as $\mathbf{Q}$, $\mathbf{U}$, $\mathbf{C}$, $\mathbf{D}$ and $\mathbf{S}$, we prove that all of them can be consistently satisfied in any EWFS. 
The formal definitions are provided in \Cref{appendix: QUCDS}, we summarise the intuition and state the main results here. 

\begin{itemize}[left=0pt]
    \item $\mathbf{Q}$: Agents involved in any EWFS take statements of the form $S:=$``If the outcomes $\vec{a}_l$ take values $\vec{\aaa}_l$ and the settings take the value $\vec{x}=\xxx$, then the outcomes $\vec{a}_j$ take values $\vec{\aaa}_j$ with a probability $P$.'' to be true when the probability that 
    $\vec{a}_j=\vec{\aaa}_j$ conditioned on $\vec{x}=\xxx$ and  $\vec{a}_l=\vec{\aaa}_l$ computed using the quantum mechanical Born rule in that EWFS equals $P$.
    \item $\mathbf{U}$: An agent $A_i$ can model the measurement of another agent $A_j$ as a unitary evolution of $A_j$'s lab (i.e., $A_i$ can choose the setting $x_j=0$ in certain predictions they make).
    \item $\mathbf{C}$: For any given EWFS, and any statement $S\in \Sigma^{aug}$ possible therein, if an agent $A_j$ knows that an agent $A_i$ knows the statement $S$, then $A_j$ can conclude that they know the statement $S$. In other words, agents can freely inherit each others' knowledge of statements derived from the augmented circuit in any EWFS.
     \item $\mathbf{D}$: For any given EWFS, and any logical statements $S_1,S_2\in \Sigma^{aug}_L$ possible therein, if an agent $A_i$ knows that $S_1$ is true and that $S_1$ implies $S_2$ is true, then they can conclude that $S_2$ is true.
   \item $\mathbf{S}$: If a subset $\vec{a}_j$ of measurement outcomes in an EWFS is known with certainty to take the values $\vec{\aaa_j}$, then it cannot be known with certainty that $\vec{a}_j$ takes any other values $\vec{\aaa_j}'\neq \vec{\aaa_j}$. 
\end{itemize}

We then obtain the following corollary, which follows by construction of our framework along with the general consistency result of \Cref{theorem: main}. We nevertheless include a proof in \Cref{appendix: proofs} for completeness.
\begin{restatable}[]{corollary}{QUCDSassumptions}
\label{corollary: QUCDSassumptions}
    If agents in an EWFS reason
about each other’s knowledge using the augmented
circuit for the scenario, then they can never arrive at a logical contradiction even if they reason
using all five assumptions \hyperref[def: Q_assump]{$\mathbf{Q}$}, \hyperref[def: U_assump]{$\mathbf{U}$}, \hyperref[def: assump_C]{$\mathbf{C}$}, \hyperref[def: assump_D]{$\mathbf{D}$} and \hyperref[def: assump_S]{$\mathbf{S}$}.
\end{restatable}

We have given a formalisation \hyperref[def: Q_assump]{$\mathbf{Q}$}, \hyperref[def: U_assump]{$\mathbf{U}$}, \hyperref[def: assump_C]{$\mathbf{C}$}, \hyperref[def: assump_D]{$\mathbf{D}$}, \hyperref[def: assump_S]{$\mathbf{S}$} of the FR assumptions $\textup{Q}$, $\textup{U}$, $\textup{C}$, $\textup{D}$ and $\textup{S}$ within our framework, showing our version of the 5 assumptions to be perfectly consistent in an EWFS, despite FR's claim that the original version of these assumptions lead to contradictions. FR's assumptions, although motivated as capturing ``quantum theory'' and ``logical axioms'' do not appear to be fully and rigorously formalised, allowing room for interpretation  (see also \Cref{appendix:relation_prev_works} for references to previous responses to FR's arguments). 
See \Cref{sec:The choices of Heisenberg cuts do matter in FR} for further discussion on the interpretation of FR's claims in light of our results.

Physically, our assumptions \hyperref[def: Q_assump]{$\mathbf{Q}$}, \hyperref[def: U_assump]{$\mathbf{U}$}, \hyperref[def: assump_C]{$\mathbf{C}$}, \hyperref[def: assump_D]{$\mathbf{D}$}, \hyperref[def: assump_S]{$\mathbf{S}$} still encompass the same essential physical requirements highlighted by FR: the universal applicability of quantum theory (Born rule + unitarity), the inheritance of agents' knowledge, and the validity of classical logic applied to knowledge of measurement outcomes. However, mathematically, within the Kripke structure of modal logic, the set of statements $\Sigma$ to which our assumptions apply differs from that considered in the original modal logic formulation of FR's result (as given in \cite{Nurgalieva2018}), due to the additional structure provided by the setting labels in our framework.

In the forthcoming section, we identify an additional assumption within our framework necessary to reproduce apparent inconsistencies akin FR's. This underscores that an implicit assumption concerning the independence of predictions from the choice of Heisenberg cuts is required to establish an FR-type no-go theorem in quantum theory.

\subsection{Reason for apparent inconsistencies}
\label{sec: I_assump}

Despite satisfying a formal version of each of the FR assumptions (\Cref{theorem: main} and \Cref{corollary: QUCDSassumptions}), our framework remains logically consistent. This raises the question: what additional assumption is needed to reproduce logical contradictions such as the apparent FR paradox within our framework?

The main difference between our framework and previous analyses of EWFSs is in the explicit introduction of the settings. Here we say ``explicit'' since the settings are indeed present in the conventional computations of quantum predictions, as a choice about how measurements are modelled has to be made when computing the probabilities. The difference between these conventional predictions and statements, and ours is that the former do not specify this in the probability expressions or at the level of statements being made while our formalism does so. Dropping this choice corresponds to an assumption regarding the ability to ignore the setting choice or the Heisenberg cut. We formalise this assumption below within our framework, and show that this is necessary to recover a paradox. This in turn yields a more precise and refined interpretation of FR type apparent paradoxes (as discussed in \Cref{sec:The choices of Heisenberg cuts do matter in FR} after analysing the example of the FR scenario). We also discuss its relation to other results in the literature and interpretations of quantum theory in~\Cref{appendix:relation_prev_works}.

\begin{definition}[Assumption $\mathbf{I}$: Independence]\label{def: I_assump}
A setting-conditioned prediction $P(\vec{a}_j=\vec{\aaa}_j|\vec{a}_l=\vec{\aaa}_l,\vec{x}=\vec{\xxx})$ is said to be setting-independent if $P(\vec{a}_j=\vec{\aaa}_j|\vec{a}_l=\vec{\aaa}_l,\vec{x}=\vec{\xxx})=P(\vec{a}_j=\vec{\aaa}_j|\vec{a}_l=\vec{\aaa}_l,\vec{x}=\vec{\xxx'})$ for all allowed values $\xxx$ and $\xxx'$ of the settings \cite{Note5}. Then the prediction can be consistently represented by dropping the settings, as $P(\vec{a}_j=\vec{\aaa}_j|\vec{a}_l=\vec{\aaa}_l,\vec{x}=\vec{\xxx})=P(\vec{a}_j=\vec{\aaa}_j|\vec{a}_l=\vec{\aaa}_l)$.
\end{definition}

We can then immediately obtain the following corollary of our results. 
\begin{corollary}
\label{corollary:setting_independence}
    In order to obtain an apparent logical contradiction (i.e., a violation of \hyperref[def: assump_S]{$\mathbf{S}$}) in an augmented EWFS where agents reason using assumptions \hyperref[def: Q_assump]{$\mathbf{Q}$}, \hyperref[def: U_assump]{$\mathbf{U}$}, \hyperref[def: assump_C]{$\mathbf{C}$} and \hyperref[def: assump_D]{$\mathbf{D}$}, it is necessary to assume \hyperref[def: I_assump]{$\mathbf{I}$} on at least one logical setting-conditioned prediction that is not setting-independent.
\end{corollary}

This corollary follows because firstly, using \Cref{corollary: QUCDSassumptions} we have that the augmented EWFS is perfectly consistent with all five assumptions \hyperref[def: Q_assump]{$\mathbf{Q}$}, \hyperref[def: U_assump]{$\mathbf{U}$}, \hyperref[def: assump_C]{$\mathbf{C}$}, \hyperref[def: assump_D]{$\mathbf{D}$} and \hyperref[def: assump_S]{$\mathbf{S}$}. Further, the consistency of our framework captured by \Cref{theorem: main} and \Cref{lemma: samechoices_samepredictions} implies that any apparent violation of \hyperref[def: assump_S]{$\mathbf{S}$} (i.e., $P(\vec{a}_j=\vec{\aaa}_j)=1$ and $P(\vec{a}_j=\vec{\aaa}'_j)=1$ for $\vec{\aaa}'_j\neq \vec{\aaa}_j$) that might be obtained in an EWFS (for instance, the violation obtained by FR \cite{Frauchiger2018}), necessarily arises by computing the probability of the outcomes under two distinct choices of settings and then ignoring this choice by identifying the two predictions as the same (which is equivalent to applying \hyperref[def: I_assump]{$\mathbf{I}$} to the two setting-conditioned predictions).

In other words, an apparent violation of \hyperref[def: assump_S]{$\mathbf{S}$} such as $P(\vec{a}_j=\vec{\aaa}_j)=1$ and $P(\vec{a}_j=\vec{\aaa}'_j)=1$ when formulated explicitly within our framework, always translates to $P(\vec{a}_j=\vec{\aaa}_j|\vec{x}=\vec{\xxx})=1$ and $P(\vec{a}_j=\vec{\aaa}'_j|\vec{x}=\vec{\xxx'})=1$ (where $\vec{\xxx}$ and $\vec{\xxx}'$ are two distinct setting values), which is not paradoxical as it refers to two distinct conditional probability distributions. Note however that this necessarily violates \hyperref[def: I_assump]{$\mathbf{I}$} since the predictions are indeed dependent on the setting, which is why the apparent paradox is recovered when \hyperref[def: I_assump]{$\mathbf{I}$} is imposed: the contradiction arises from imposing \hyperref[def: I_assump]{$\mathbf{I}$} in a scenario where it fails to hold.

\section{A simple resolution to the FR apparent paradox}
\label{sec: resolution_entanglement}

Having developed a fully general framework, we now apply it to a concrete example, to demonstrate a simple resolution of the original FR paradox. We focus on the entanglement version of the FR protocol, illustrated in \cref{fig: FR_ent_circuit} and reviewed in detail in \Cref{sec: FR_ent_review}, highlighting the core idea behind the resolution both at the conceptual and formal level. This version, proposed in talks by Luis Masanes and Matthew Pusey, is analyzed in full detail in \Cref{appendix: ent_FR}, which includes all supporting calculations.

While many view the entanglement version as equivalent to the original prepare-and-measure protocol, the authors of the FR paper have emphasized that it omits key subtleties involving timing, of the measureents and reasoning steps, which were central to the original formulation. Our framework is fully general and resolves the apparent paradoxes in both versions in the same way. In \Cref{ssec: resolution_prep}, we address the original FR paradox directly, providing a step-by-step analysis showing how the paradox entirely disappears within our framework. This holds even when all agents reason freely using unitary quantum theory, the standard Born rule, and classical logic, and all statements from FR's original argument are preserved in a more explicit form, in our framework.

{\bf Key takeaway} FR's paper argued that “quantum theory cannot consistently justify the use of itself,” based on a paradox involving quantum agents reasoning about others’ measurement outcomes. Our resolution reveals a key hidden assumption in this argument: assumption \hyperref[def: I_assump]{$\mathbf{I}$} (setting-independence). The contradiction arises only when quantum theory is used in a restricted way that imposes \hyperref[def: I_assump]{$\mathbf{I}$} \cite{Note6}. Observing that the actual protocol violates this assumption i.e., predictions do depend on settings (which describe whether a measurement is modelled purely unitarily or as producing classical records), we relax \hyperref[def: I_assump]{$\mathbf{I}$} and show that the paradox vanishes. By tracking the settings in the augmented circuit associated with each agent’s reasoning, we provide explicit versions of their statements, which are consistent with quantum theory and classical logic, but their combination no longer yields any contradictions.

\begin{figure*}[t]
	\centering
	\begin{tikzpicture}
		\draw (-1.5,0)--(4,0);  \draw (-1.5,-1)--(4,-1);  \draw (-1.5,-3)--(4,-3); \draw (-1.5,-4)--(4,-4); \draw [ultra thick,decorate,
		decoration = {calligraphic brace}] (-1.7,-3.2)--(-1.7,-0.8); 
		
		\node at (-2,0) {$\ket{0}_\A$};
		\node at (-2.5,-2) {$
			\ket{\psi}_\textup{RS}$};
		\node at (-2,-4) {$\ket{0}_\B$};
		
		\fill (1.5,-1) circle [radius=3pt]; \draw (1.5,0) circle (3mm); \draw (1.5,-1)--(1.5,0.3);
		\fill (1.5,-3) circle [radius=3pt]; \draw (1.5,-4) circle (3mm); \draw (1.5,-3)--(1.5,-4.3);
		
		\draw[dashed] (1,-1.5) rectangle (2,0.5); \node at (1.5,1.2) {\shortstack{Alice's \\ measurement\\ of $\textup{R}$}};
		
		\draw[dashed] (1,-4.5) rectangle (2,-2.5); \node at (1.5,-5.2) {\shortstack{Bob's \\ measurement \\ of $\textup{S}$}};
		
		\draw (4,-1.5) rectangle node{\shortstack{Ursula's\\measurement}} (6,0.5); \draw[->] (6,-0.5)--(7,-0.5); \node at (8.2,-0.5) {$u\in \{\textup{ok},\textup{fail}\}$};
		\draw (4,-4.5) rectangle node{\shortstack{Wigner's\\measurement}} (6,-2.5);  \draw[->] (6,-3.5)--(7,-3.5); \node at (8.2,-3.5) {$w\in \{\textup{ok},\textup{fail}\}$};
		
		\node at (-2,-6.5) {$t=1$};  \node at (1.5,-6.5) {$t=2$};  \node at (5,-6.5) {$t=3$};
	\end{tikzpicture}
	\caption{Circuit that describes the entanglement version of the FR protocol, from the view of the super-agents Ursula and Wigner who describe the measurement of Alice and Bob as unitary evolutions. The protocol proceeds as follows: Alice and Bob share a bipartite state $\ket{\psi}_{\R\Ss}=\frac{1}{\sqrt{3}}(\ket{00}+\ket{10}+\ket{11})_{\R\Ss}$, Alice measures $\R$ and Bob measures $\Ss$, both in the computational basis, and the agents store the outcome of the measurement in their memories $\A$ and $\B$ respectively (the unitary description of these measurements is a CNOT). Ursula then measures $\R\A$ (Alice's lab) and Wigner measures $\Ss\B$ (Bob's lab), both agents measure in the $\{\ket{ok/fail}:=\frac{1}{\sqrt{2}}(\ket{00}\mp \ket{11}))\}$ basis
 to obtain the outcomes $u,w\in\{\textup{ok},\textup{fail}\}$. Note that this circuit alone does not allow the super-agents to reason about the classical measurement outcome of Alice and Bob as these are modelled purely unitarily, and hence no classical measurement outcomes $a$ and $b$ are identified. }
	\label{fig: FR_ent_circuit}
\end{figure*}

\subsection{Augmented circuit}
We first formulate the entanglement version of the FR protocol within our framework by giving its augmented circuit. A typical circuit associated with this scenario, also found in the previous literature, is shown in \cref{fig: FR_ent_circuit}. The caption of the figure gives a quick recap of the protocol sufficient to follow this discussion. It is one which models the measurements of the agents Alice and Bob purely unitarily. As these are computational basis measurements, the corresponding unitary is the CNOT. However, this circuit alone does not allow us to calculate the probabilities of Alice and Bob's classical outcomes. In the previous literature (see for instance \cite{Nurgalieva2020}), multiple different circuits are considered for calculating the probabilities from the perspective of each agent. In our framework, we have shown that all such reasoning can in fact be captured with a single circuit|the augmented circuit. This circuit includes a setting $x_i$ for each agent $\A_i$, which when set to $x_i=0$ models their measurement as the unitary $\text{$\matholdcal{M}$}^{\A_i}_{unitary}$ and when set to $x_i=1$ models the measurement as the unitary $\text{$\matholdcal{M}$}^{\A_i}_{unitary}$ followed by a set of projectors associated with the measurement outcome $a_i$.

In the present case, we have four agents. In the full augmented circuit, we would therefore have four setting variables, one for each agent. Moreover, in the FR protocol, Ursula and Wigner announce their classical outcomes $u$ and $w$ in each run of the protocol and halt when they obtain $u=w=\textup{ok}$. Therefore in the full augmented circuit this communication channel between Ursula and Wigner which captures this announcement will also be present. However, for the purpose of our discussion here and to see the resolution of the paradox, it suffices to work with a much simplified augmented circuit, where we fix Ursula's and Wigner's settings to 1 (as they are effectively classical from the perspective of all agents involved) \cite{Note7} and ignore the explicit communication channel for the announcement as this is a post-processing. Then we only need to assign settings $x_1$ and $x_2$ to the agents Alice and Bob.

We can now write the enlarged set of projectors $\mathtt{\Pi}^A_{x_1}$, and $\mathtt{\Pi}^B_{x_2}$ (c.f. \cref{eq: proj_settings}) as
\begin{align}
	\label{eq: FR_aug_proj}
	\begin{split}
		\mathtt{\Pi}^A_{0}&:=\{\pi_{0,\perp}^A=\id_{\R\A}\},\\
		\mathtt{\Pi}^A_{1}&:=\{\pi_{1,0}^A=\proj{00}_{\R\A},\pi_{1,1}^A=\proj{11}_{\R\A}\}\\
		\mathtt{\Pi}^B_{0}&:=\{\pi_{0,\perp}^B=\id_{\Ss\B}\},\\
		\mathtt{\Pi}^B_{1}&:=\{\pi_{1,0}^B=\proj{00}_{\Ss\B},\pi_{1,1}^B=\proj{11}_{\Ss\B}\}\\
	\end{split}
\end{align}
These capture the fact that when Alice's setting $x_1=0$, her measurement is modelled as a unitary evolution and the $a$ is set deterministically to the trivial value $\perp$. When $x_1=1$, Alice's measurement is first modelled unitarily and then her measurement outcomes $a=0$ and $a=1$ are identified by the projectors $\pi_{1,0}^A$ and $\pi_{1,1}^A$. The case for Bob is similar.

With this we obtain the (simplified) augmented circuit for the entanglement version of the FR protocol which is illustrated in \cref{fig: FR_ent_circuit_aug}. Note that here Alice and Bob act at the same time $t=2$ and Ursula and Wigner act at $t=3$, but we could always have the agents' operations occur at different times without affecting the circuit structure, such that this circuit fits within the general form of \cref{fig: genform_circuit_settings}. We do not do this here for simplicity, but it is easy to see that this transformation would not affect any of the arguments.

\begin{figure*}[t!]
	\centering
	\begin{tikzpicture}
		\draw (-1.5,0)--(4,0);  \draw (-1.5,-1)--(4,-1);  \draw (-1.5,-3)--(4,-3); \draw (-1.5,-4)--(4,-4); \draw [ultra thick,decorate,
		decoration = {calligraphic brace}] (-1.7,-3.2)--(-1.7,-0.8); 
		
		\node at (-2,0) {$\ket{0}_\A$};
		\node at (-2.5,-2) {$
			\ket{\psi}_\textup{RS}$};
		\node at (-2,-4) {$\ket{0}_\B$};
		
		\node at (-2,-6.5) {$t=1$};  \node at (2.5,-6.5) {$t=2$};  \node at (7,-6.5) {$t=3$};
		
		\fill (1.5,-1) circle [radius=3pt]; \draw (1.5,0) circle (3mm); \draw (1.5,-1)--(1.5,0.3);
		\fill (1.5,-3) circle [radius=3pt]; \draw (1.5,-4) circle (3mm); \draw (1.5,-3)--(1.5,-4.3);
		
		\draw[dashed] (1,-1.5) rectangle (2,0.5); \node at (1.5,1.2) {\shortstack{Alice's \\ measurement \\ of $\textup{R}$}};
		
		\draw[dashed] (1,-4.5) rectangle (2,-2.5); \node at (1.5,-5.2) {\shortstack{Bob's\\ measurement \\ of $\textup{S}$}};
		
		\draw (6,-1.5) rectangle node{\shortstack{Ursula's\\measurement}} (8,0.5); \draw[->] (8,-0.5)--(9,-0.5); \node at (10.2,-0.5) {$u\in \{\textup{ok},\textup{fail}\}$};
		\draw (6,-4.5) rectangle node{\shortstack{Wigner's\\measurement}} (8,-2.5);  \draw[->] (8,-3.5)--(9,-3.5); \node at (10.2,-3.5) {$w\in \{\textup{ok},\textup{fail}\}$};

		\draw[blue] (4,-1.5) rectangle node{$\mathtt{\Pi}^A_{x_1}$} (5,1);
		\draw[->, blue] (3.5,1)--(4,0.5);  \draw[->, blue] (5,0.5)--(5.5,1);\node[blue] at (3.3,1.3) {$x_1\in \{0,1\}$}; \node[blue] at (5.7,1.3) {$a\in \{\perp,0,1\}$};

		\draw[blue] (4,-5) rectangle node{$\mathtt{\Pi}^B_{x_2}$} (5,-2.5);
		\draw[->, blue] (3.5,-5)--(4,-4.5);  \draw[->, blue] (5,-4.5)--(5.5,-5);\node[blue] at (3.3,-5.3) {$x_2\in \{0,1\}$}; \node[blue] at (5.7,-5.3) {$b\in \{\perp,0,1\}$};
		
		\draw (5,0)--(6,0);  \draw (5,-1)--(6,-1);  \draw (5,-3)--(6,-3);  \draw (5,-4)--(6,-4);
	\end{tikzpicture}
	
	\caption{Augmented circuit for the entanglement version of the FR protocol. Unlike the circuit of \cref{fig: FR_ent_circuit}, this circuit allows agents to model the measurements of Alice and Bob as unitary evolutions and at the same time also reason about their classical measurement outcomes. To model Alice unitarily, one must set $x_1=0$, in which case the trivial outcome $a=\perp$ is obtained deterministically and the blue box acts as an identity operation. To reason about Alice's classical outcome, one must set $x_2=1$, in which case the blue box applies the projectors corresponding to the classical outcome $a\in\{0,1\}$. All possible reasoning in this FR set-up can be derived from this circuit under different choices of settings, but these statements can no longer be combined to yield a paradox as explained in the main text. }
	\label{fig: FR_ent_circuit_aug}
\end{figure*}

\subsection{Explicit version of the statements that resolve the paradox}
As reviewed in \Cref{sec: FR_ent_review}, the four logical statements involved in the entanglement version of the FR paradox (see \cref{fig: FR_ent_circuit} for a quick recap of the protocol) are as follows.

\begin{align}
\label{eq: FR_statements}
    \begin{split}
       u=\textup{ok} &\land w=\textup{ok} \\
       u=\textup{ok} &\Rightarrow b=1\\
       b=1 &\Rightarrow a=1\\
       a=1 &\Rightarrow w=\textup{fail}
    \end{split}
\end{align}
The paradox ensues since the last three statements can be combined using classical logic to yield $u=\textup{ok}\Rightarrow w=\textup{fail}$, which contradicts the first statement. 
These statements follow due the following four conventional predictions (\Cref{def: conv_prediction}) of the 4-agent EWFS specified by this protocol, while post-selecting on an experimental run where $u=w=\textup{ok}$ is obtained by the super-agents Ursula and Wigner. The first prediction in the list below guarantees this post-selection will eventually succeed.

\begin{align}
\label{eq: FR_conv_pred}
    \begin{split}
       P_{conv}(u=\textup{ok}, w=\textup{ok})&=\frac{1}{12} \\
       P_{conv}(b=1|u=\textup{ok})&=1\\
       P_{conv}(a=1|b=1)&=1\\
       P_{conv}(w=\textup{fail}|a=1)&=1
    \end{split}
\end{align}

From our main theorem, \Cref{theorem: main}, it follows that these conventional predictions are equivalent to the following four setting conditioned predictions (that can be computed from the augmented circuit of \cref{fig: FR_ent_circuit_aug}). 

\begin{align}
\label{eq: FR_setting_cond_pred}
    \begin{split}
       P(u=\textup{ok}, w=\textup{ok}| (x_1,x_2)=(0,0))&=\frac{1}{12} \\
       P(b=1|u=\textup{ok}, (x_1,x_2)=(0,1))&=1\\
       P(a=1|b=1, (x_1,x_2)=(1,1))&=1\\
       P(w=\textup{fail}|a=1, (x_1,x_2)=(0,1))&=1
    \end{split}
\end{align}

The 3 logical statements associated with the last three setting-conditioned predictions, together with the corresponding original 3 logical statements of FR are given in \cref{table: resolution_ent} for comparison.

\begin{table}[h!]
	\centering
	\begin{tabular}{c|c}
		original  formulation &  our explicit formulation \\
		\hline
		$u=\textup{ok}\Rightarrow b=1$     &  $(x_1,x_2)=(0,1) \land u=\textup{ok} \Rightarrow b=1$\\
		$b=1\Rightarrow a=1$     &  $(x_1,x_2)=(1,1) \land b=1 \Rightarrow a=1$\\
		$a=1\Rightarrow w=\textup{fail}$     &  $(x_1,x_2)=(1,0) \land a=1 \Rightarrow w=\textup{fail}$\\
	\end{tabular}
	\caption{Explicit versions of the FR statements as given in our framework, which provides a logically consistent resolution to the FR apparent paradox without giving up any of the original assumptions of FR. While the original statements can be chained together to yield $u=\textup{ok} \Rightarrow w=\textup{fail}$ which yields a contradiction along with the fact that $P(u=w=\textup{ok})>0$, our explicit version of the statements cannot be chained together to yield this conclusion even using the axioms of classical logic. \label{table: resolution_ent}}
\end{table}

Then we can immediately see that while the original statements can be chained together using classical logical rules to yield $u=w=\textup{ok}\Rightarrow w=\textup{fail}$ (\cref{eq: chain4}), the explicit version of the statements obtained in our framework cannot be chained together in the same manner, even under the standard rules of classical logic.

The fact that the probabilities of \cref{eq: FR_conv_pred,eq: FR_setting_cond_pred} indeed match can seen by explicitly writing out the expression of the conventional prediction using the Born rule, which we show in full detail in \Cref{appendix: ent_FR}. We illustrate this here for some of these cases. Consider the conventional prediction $P_{conv}(u=\textup{ok}, w=\textup{ok})$. This is computed in the circuit of \cref{fig: FR_ent_circuit} i.e., applying the two CNOT gates to the initial state of $\ket{\psi}_{\R\Ss}$ and memories initialised to $\ket{0}_{\A}$ and $\ket{0}_{\B}$, one obtains $\ket{\psi}_{\R\A\Ss\B}=\frac{1}{\sqrt{3}}(\ket{0000}+\ket{1100}+\ket{1111})_{\R\A\Ss\B}$. The measurements of Ursula and Wigner act on this state. Indeed our augmented circuit of \cref{fig: FR_ent_circuit_aug} for the case $(x_1,x_2)=(0,0)$ is equivalent to the circuit of \cref{fig: FR_ent_circuit}. Thus writing our the two predictions, we have

\begin{align}
    \begin{split}
   &P_{conv}(u=\textup{ok}, w=\textup{ok})\\
   =&P(u=\textup{ok}, w=\textup{ok}| (x_1,x_2)=(0,0))\\
   =&|\bra{\textup{ok}}_{\R\A}\otimes\bra{\textup{ok}}_{\Ss\B}.\ket{\psi}_{\R\A\Ss\B}|^2
    \end{split}
\end{align}

For comparison, consider the conventional prediction $P_{conv}(a=1|b=1)$. This is fully specified by $P_{conv}(a=1,b=1)$, and it will be more illustrative to consider that. This is obtained by applying the computational basis measurement on $\R$ and on $\Ss$ to the initial state $\ket{\psi}_{\R\Ss}$, which is equivalent to applying the $\{\ket{00},\ket{11}\}$ basis measurement to $\ket{\psi}_{\R\A\Ss\B}$. This is exactly what one would get when using the settings $(x_1,x_2)=(1,1)$ in the augmented circuit of \cref{fig: FR_ent_circuit_aug}. 
We have

\begin{align}
    \begin{split}
     &P_{conv}(a=1,b=1)\\
     =&P(a=1,b=1|(x_1,x_2)=(1,1))\\
     =&|\bra{11}_{\R\A}\otimes\bra{11}_{\Ss\B}.\ket{\psi}_{\R\A\Ss\B}|^2.   
    \end{split}
\end{align}

The equivalence of the remain predictions can be similarly shown, as detailed in \Cref{appendix: ent_FR}.

\begin{remark}[On the role of the projection postulate]
\label{remark: no_collapse}
As seen in \Cref{sec:background}, Wigner’s Friend Scenarios generally involve three aspects of quantum theory: unitary evolution, projection postulate and the Born rule. While all three aspects are present in textbook quantum theory, one may however question the role of the projection postulate in arguments based on measurement probabilities.

In the entanglement version of the FR paradox considered here, the apparent contradiction arises from combining the logical statements in \cref{eq: FR_statements}, which directly follow from the standard predictions in \cref{eq: FR_conv_pred}. These predictions (measurement probabilities) are derived using the Born rule and unitary modelling of measurements, without requiring the projection postulate, which specifies the post-measurement state. Importantly, our resolution, based on the setting-conditioned predictions in \cref{eq: FR_setting_cond_pred}, also operates entirely at the level of measurement probabilities and does not invoke the post-measurement state. Therefore, the resolution given in \cref{table: resolution_ent} also need not invoke the projection postulate.

An important but subtle point here is that even though the setting-conditioned predictions involve the setting 1 (which was described in the augmented circuit as associated with the state update of the projection postulate), one does not require this state update rule for computing the predictions. Nevertheless, the projectors associated with the setting 1 description are \emph{necessary} to identify the measurement outcome, and compute its probability via the Born rule. Moreover, they specify the measurement basis.

For example, to compute the joint probability of Ursula and Bob's outcomes $u$ and $b$, Alice’s measurement is modelled unitarily (setting $x_1 = 0$), and its basis is irrelevant for the prediction. In contrast, Bob’s measurement, described under setting $x_2 = 1$, requires the basis projectors ${\pi_{1,0}^\B = \proj{00}_{\Ss\B}, \pi_{1,1}^\B = \proj{11}_{\Ss\B}}$ from \cref{eq: FR_aug_proj} to identify $b$ and compute its probability. Thus, setting 1 need not imply an objective “collapse,” but rather encodes the information required to specify and interpret a classical record, independently of any post-measurement state. 
This can be perfectly consistent with a unitary description of the measurement by another agent, for whom the original measurement is regarded as an evolution of a closed system and the classical record is unknown (associated with trivial outcome $\perp$).

This illustrates that even versions of the FR argument that avoid invoking the projection postulate or collapse dynamics can be resolved within our framework, simply by carefully accounting for the knowledge and conditioning used in the reasoning.
\end{remark}

\subsection{Setting-dependence: a refined interpretation of the FR paradox}\label{sec:The choices of Heisenberg cuts do matter in FR}

In \Cref{sec: I_assump} we identified a new assumption \hyperref[def: I_assump]{$\mathbf{I}$} (setting-independence) and showed that inconsistent quantum predictions in EWFSs only arise when assuming  \hyperref[def: I_assump]{$\mathbf{I}$}  in a situation where predictions do depend on the setting. Analysing this assumption for the FR scenario sheds light on the root cause of such FR paradoxes yielding more refined physical interpretation.

In \Cref{appendix: setting_dep_FR}, we show the setting-dependence of predictions in the FR scenario (i.e., a violation of \hyperref[def: I_assump]{$\mathbf{I}$}), by explicitly computing them. 
For instance it follows from the detailed analysis of \Cref{appendix: ent_FR} that 
that the conventional prediction $P_{conv}(w=\textup{fail}|a=1)$ (used to derive the statement $a=1\Rightarrow w=\textup{fail}$) is given as
\begin{align}
    \begin{split}
        &P_{conv}(w=\textup{fail}|a=1)=\\
       & P(w=\textup{fail}|a=1,(x_1,x_2)=(1,0))=1
    \end{split}
\end{align}

$P(w=\textup{fail}|a=1,(x_1,x_2)=(1,0))$ is indeed setting dependent (violates \hyperref[def: I_assump]{$\mathbf{I}$}), as one can verify (see \Cref{appendix: setting_dep_FR}) that 
\begin{equation}
    P(w=\textup{fail}|a=1,(x_1,x_2)=(1,1))=\frac{1}{2}\neq 1.
\end{equation}

Similarly, the setting-dependence of other predictions of the FR scenario can also be verified in our framework. 

Using the general results of \Cref{ssec: general_results} and the above analysis of the FR scenario, we can readily prove the following theorem regarding our assumptions applied to the FR scenario. This applies to both the entanglement version (that was the focus here) and the prepare and measure versions (described in \Cref{ssec: resolution_prep}).

\begin{restatable}{theorem}{FRtheorem}
	\label{theorem: FR}
	There exists a consistent description of the FR protocol (both versions) that satisfies all five assumptions \hyperref[def: Q_assump]{$\mathbf{Q}$}, \hyperref[def: U_assump]{$\mathbf{U}$}, \hyperref[def: assump_C]{$\mathbf{C}$}, \hyperref[def: assump_D]{$\mathbf{D}$} and \hyperref[def: assump_S]{$\mathbf{S}$} but violates \hyperref[def: I_assump]{$\mathbf{I}$} for certain logical setting-conditioned predictions. Furthermore, when simultaneously assuming \hyperref[def: Q_assump]{$\mathbf{Q}$}, \hyperref[def: U_assump]{$\mathbf{U}$}, \hyperref[def: assump_C]{$\mathbf{C}$}, \hyperref[def: assump_D]{$\mathbf{D}$} and \hyperref[def: assump_S]{$\mathbf{S}$} in the FR protocol, additionally imposing \hyperref[def: I_assump]{$\mathbf{I}$} on at least one logical setting-conditioned prediction is a necessary condition for reproducing the apparent FR paradox, while imposing \hyperref[def: I_assump]{$\mathbf{I}$} on all logical setting-conditioned predictions is a sufficient condition for the same.
\end{restatable}

FR have claimed that their assumptions $\textup{Q}$, $\textup{U}$, $\textup{C}$, $\textup{D}$ and $\textup{S}$ lead to a contradiction in a physical theory that reproduces the quantum predictions of the FR scenario, while we have formalised a version \hyperref[def: Q_assump]{$\mathbf{Q}$}, \hyperref[def: U_assump]{$\mathbf{U}$}, \hyperref[def: assump_C]{$\mathbf{C}$}, \hyperref[def: assump_D]{$\mathbf{D}$} and \hyperref[def: assump_S]{$\mathbf{S}$} of these, showing that they can always be applied consistency even while reproducing the FR predictions. This calls for  closer examination of FR's claim to understand this apparent mismatch.

FR's work suggests that their assumptions $\textup{Q}$, $\textup{U}$, $\textup{C}$, $\textup{D}$ and $\textup{S}$ should be interpreted as capturing the validity of unitary ``quantum theory'' and ``classical logic'' applied to the knowledge of agents. 
However, the formalisation of the FR assumptions leave ambiguities in the definition of ``agents' knowledge'', and in how agents model measurements in each statement they make. Indeed, we have provided a formalisation of the same principles via \hyperref[def: Q_assump]{$\mathbf{Q}$}, \hyperref[def: U_assump]{$\mathbf{U}$}, \hyperref[def: assump_C]{$\mathbf{C}$}, \hyperref[def: assump_D]{$\mathbf{D}$} and \hyperref[def: assump_S]{$\mathbf{S}$}, which provably avoids FR-type paradoxes, and have shown that the additional assumption \hyperref[def: I_assump]{$\mathbf{I}$} (setting-independence) must be imposed to recover an apparent contradiction.

This contrast necessitates recognizing two versions of quantum theory in the context of EWFSs: (1) \emph{Heisenberg cut independent} and (2) \emph{Heisenberg cut dependent} versions.

In essence, FR's result can be interpreted as revealing a contradiction between version (1) of quantum theory and classical logic within a specific EWFS, while our findings establish the general consistency between version (2) of quantum theory and classical logic across all EWFSs. Then the two sets of results are mutually consistent, even though our results imply the broader claim of FR, regarding the inconsistency of quantum theory does not hold.

Moreover, the violation of \hyperref[def: I_assump]{$\mathbf{I}$}—the dependence of predictions on Heisenberg cuts in EWFSs—is not unexpected once the concept of such cuts is formalised in terms of different channels describing a measurement. Not only is the \hyperref[def: I_assump]{$\mathbf{I}$} assumption violated in the FR scenario as shown here, this is also the case in Wigner's original thought experiment. Indeed, Wigner clearly points to this effect in the original paper where the thought experiment was introduced, although evidently not in the language of the settings we use here. As we can see from the review of Wigner's experiment in \Cref{sec:background}, the core message is that the ambiguity in how a measurement is modelled, in light of the unitarity vs projection postulates (which in our framework is labeled by the settings), does have empirical consequences in Wigner's scenario. 

This highlights that in EWFSs, when assuming the universal validity of unitary quantum theory, it is natural to consider version (2) of quantum theory, which incorporates the Heisenberg cut dependence of predictions. This version can be formalised while remaining perfectly consistent in all EWFSs as we have shown. Both Wigner's original result as well as FR's result therefore serve to raise caution on the dangers of using version (1) of quantum theory in such EWFSs. However, FR's arguments differ from Wigner's as FR arrive at the apparent contradiction while only requiring agents performing compatible measurements to reason about each other in each statement \cite{Note8}, while the apparent paradox obtained in Wigner's experiment requires a super-agent (Bob) to issue statements about an agent (Alice) who performs an incompatible measurement.

\subsection{Link to (non)absoluteness of observed events}\label{sec: AoE}

Although FR’s arguments center on agents’ reasoning, our analysis here takes a step back and focuses on the more fundamental aspect of the predictions of quantum theory for the scenario. This allows us to draw insights on the failure of an absolute notion of observed events in EWFS (as highlighted in \cite{Brukner2018, Cavalcanti2021}), and its relation to the settings we have introduced in this work. 

Absoluteness of observed events (AoE) entails that the predictions of the scenario can be derived from a single joint probability distribution on the observed (non-trivial) outcomes of all agents. 
However, the quantum predictions of the FR scenario given in \cref{eq: FR_conv_pred} are not compatible with a single joint distribution $P(u,w,a,b)$ on the observed outcomes of all four agents.

We have shown that these conventional predictions of \cref{eq: FR_conv_pred} are equivalent to the setting-conditioned predictions of \cref{eq: FR_setting_cond_pred}. The setting choices are present but not made explicit in the conventional representation.

Once we account for the setting-dependence of these predictions using their explicit form given in \cref{eq: FR_setting_cond_pred}, their incompatibility with a single, well-defined joint distribution $P(u,w,a,b)$ independent of any settings, becomes immediate. This is because those predictions arise from different conditional probability distributions but by ignoring the conditioning information (on which the prediction depends), therefore whenever there is such setting-dependence in a scenario (violation of \hyperref[def: I_assump]{$\mathbf{I}$}), we cannot expect AoE to hold for the corresponding conventional predictions of that scenario. 

Due to space considerations, we leave a detailed discussion of other no-go results for EWFSs relating to the absoluteness of events \cite{Brukner2018, Bong2020}, particularly the Local-Friendliness (LF) theorem \cite{Bong2020}, and the relation between AoE and \hyperref[def: I_assump]{$\mathbf{I}$}  to a follow-up paper \cite{LF_Vilasini_Woods}. However, we note that our \hyperref[def: I_assump]{$\mathbf{I}$} assumption plays a very different role in the FR vs LF scenarios. While we have shown that the violation of \hyperref[def: I_assump]{$\mathbf{I}$} is sufficient to fully evade the conclusion of FR's paper that ``Quantum theory cannot consistently justify the use of itself'', the violation of \hyperref[def: I_assump]{$\mathbf{I}$} in the LF scenario cannot be used to evade their conclusions, rather it sheds deeper light on the structure and meaning of their absoluteness of events assumption.

\section{Emergence of absolute measurement events}
\label{sec: standard_QT}

In Wigner's Friend Scenarios, relational approaches typically propose to evade no-go theorems by relaxing the objectivity of measurement outcomes and demanding that measurement outcomes are to be defined relative to an agent, a context, a world or some other new concept introduced within the framework (see also the discussions in \cite{Frauchiger2018, Nurgalieva2018, Nurgalieva2020}). However, this leaves open the crucial question of how one can recover predictions and observations of existing real-world quantum experiments where measurement outcomes and probabilities are objective for all agents and do not appear to be relational for any given protocol.

In this section, we address this question and show how our formalism recovers the apparent objectivity of measurement outcomes and the standard predictions of quantum theory in existing real-world experiments, even though it provides a fundamentally relational and consistent measurement model for general EWFSs where an absoluteness notion of measurement events may not exist (as discussed in the preceeding section).

\subsection{Criteria for superagency: distinguishing standard and genuinely Wigner’s Friend scenarios}
\label{sec:super-agent}

To address the open question motivated above, we first develop formal criteria to distinguish genuinely Wigner’s Friend type scenarios from standard quantum scenarios where agents do not measure each other's memories/labs in a non-trivial manner (or more colloquially, when they do not ``Hadamard each others' brains''). 

Recall that the memory $\M_i$ of each agent $\A_i$ plays the role of ``their entire lab except the system that they measure'' i.e., the quantum system $\mathtt{S}_i$ which $\A_i$ measures together with their memory $\M_i$ constitute a minimal model of $\A_i$'s lab needed for EWFS arguments. We now proceed to formally define what is meant by ``acting non-trivially on an agent's memory'' and thereby establish a formal criterion distinguishing when one agent does not act as a super-agent to another in an EWFS. Considering this criterion for all pairs of agents in the EWFS we arrive at a concept that we term as the \emph{non-super-agent structure} ($n\matholdcal{SA}$) of the EWFS, which will allow to differentiate parts of a general EWFS as standard quantum sectors as opposed to genuine Wigner's Friend experiments. 

Here it is important to note that even in everyday scenarios, agents do act on the memories of each other in the sense that an agent can consult their memory, which stores a classical outcome, and communicate it to another agent. This may influence the operations and reasoning process of the second agent. Therefore, we need a definition that is not too restrictive to forbid this kind of ``trivial'' or ``standard'' way of acting on each others' memories, while still strong enough to identify ``non-trivial'' or ``non-standard'' ways in which the operation performed by one agent in a Wigner's Friend Scenario can act on the memory/lab of another agent.

A key preliminary concept required for developing this is the operational equivalence between two EWFSs.

\begin{definition}[Operationally equivalent EWFSs]
\label{definition: operational_equivalence}
We say that two $N$-agent EWFSs involving sets $\{\A_1',...,\A_N'\}$ and $\{\A_1,...,\A_N\}$ of agents are operationally equivalent if and only if the following hold
\begin{itemize}[left=0pt]
\item There is one-to-one identification between the systems $\mathtt{S}'=\{\Ss'_1,...,\Ss'_m\}$ and $\mathtt{S}=\{\Ss_1,...,\Ss_m\}$, agents $\{\A_1',...,\A_N'\}$ and $\{\A_1,...,\A_N\}$, memories $\{\M_1',...,\M_N'\}$ and $\{\M_1,...,\M_N\}$, measurements $\{\text{$\matholdcal{M}$}^{\A_1'},...,\text{$\matholdcal{M}$}^{\A_N'}\}$ and $\{\text{$\matholdcal{M}$}^{\A_1},...,\text{$\matholdcal{M}$}^{\A_N}\}$, subsets $\mathtt{S}'_i\subseteq \mathtt{S}'\cup \mathtt{M}'\backslash \{\M'_i\}$ and $\mathtt{S}_i\subseteq \mathtt{S}'\cup \mathtt{M}\backslash \{\M_i\}$ of systems on which each measurement acts non-trivially, and sets of operations $\{\matholdcal{E}_1',...,\matholdcal{E}_N\}$ and $\{\matholdcal{E}_1,...,\matholdcal{E}_N\}$, with outcome sets $\mathtt{O}_i'$ and $\mathtt{O}_i$ being equivalent.
\item The augmented circuits of the two EWFS yield the same setting-conditioned predictions, i.e., for all disjoint subsets $\vec{a}_j$ and $\vec{a}_l$ of outcomes in one EWFS and corresponding subsets $\vec{a}'_j$ and $\vec{a}'_l$ of outcomes in the other, as well as all settings $\vec{x}$ in one and corresponding settings $\vec{x}'$ in the other, we have

\end{itemize}
\begin{align}
    \begin{split}
    \label{eq: operational_equivalence}   &P(\vec{a}_j=\vec{\aaa_j}|\vec{a}_l=\vec{\aaa}_l,\vec{x}=\vec{\xxx})\\=&P(\vec{a}'_j=\vec{\aaa_j}|\vec{a}'_l=\vec{\aaa}_l,\vec{x}'=\vec{\xxx}).
    \end{split}
\end{align}

\end{definition}

The formal definition of the non-super-agent structure $n\matholdcal{SA}$ is given in the Appendix (\cref{definition: nSA}) as this combines the above definition with those of causal structure and directed paths from \cref{appendix: completeness_causality}. Below, we provide a detailed and illustrative intuition for the definition which will suffice to follow our main results. 

The main idea behind the $n\matholdcal{SA}$ is that for an agent $\A_j$ not to act as a super-agent to another agent $\A_i$, it suffices to know either that $\A_j$ does not act after $\A_i$ in the causal structure (in which case they do not act on $\A_i$'s memory at all \cite{Note9}), or that all operations in the augmented circuit occurring after $\A_i$'s measurement, including $\A_j$'s operations, can be simulated by equivalent operations that act trivially on $\A_i$'s memory. 

A simple example illustrates this concept, showing that it does not exclude usual scenarios where agents may communicate different information based on the outcomes stored in their memories. Notice that the post-measurement state of an agent's memory $M_i$ and measured system $S_i$ after a measurement $\matholdcal{M}^{\A_i}$ is symmetric in the exchange of $M_i$ and $S_i$ for both setting choices $x_i=0$ and $x_i=1$, as shown in \cref{eq:M_unitary_gen,eq:M_mixture_gen}. This symmetry exists because the memory is perfectly correlated with the system in the measurement basis, obtained by coherently or incoherently copying (depending on the setting) the system state in that basis. Therefore, any scenario where an agent prepares a new state for another agent based on the state stored in their memory can be perfectly mimicked by an equivalent operation that prepares the new state based on the state of the system, acting trivially (as an identity) on the memory.

However, a scenario where Wigner performs a Hadamard operation on a Friend's brain, or undoes a Friend's measurement, involves a non-trivial joint operation on the memory and system that cannot be emulated by an operation acting on the system alone. Therefore, when $\A_i$ acts before $\A_j$ in the causal order, the key property distinguishing whether or not $\A_j$ acts as a super-agent to $\A_i$ is whether $\A_j$'s operations necessarily act jointly on $\A_i$'s lab (system and memory). The non-super-agent structure $n\matholdcal{SA}$ of an EWFS (defined in \cref{definition: nSA}) includes an ordered pair $(A_i,A_j)$ of agents whenever 
$A_j$ does not act as a super-agent to $A_i$ in the EWFS.

Importantly, the $n\matholdcal{SA}$ of an EWFS is also a physical, objective and perspective-independent property of the protocol, as is the causal structure of an EWFS as defined in \cref{definition: causal_str}. Although the notion of operational equivalence in \cref{definition: operational_equivalence} refers to settings, the property is also independent of the settings (or Heisenberg cuts) as it is required to hold for all setting choices.

For instance, in Wigner's original experiment, where the agent Wigner $\W$ can measure the lab of the Friend $\F$ in the Bell basis while the Friend only acts on some quantum system (that does not include Wigner's memory/lab), applying this definition, we find that $(\F,\W)\not\in n\matholdcal{SA}$ and $(\W,\F)\in n\matholdcal{SA}$. This indicates that $\W$ can act as a super-agent to $\F$ but not vice-versa. More generally, if another agent Ursula $\U$ can ``ask'' Wigner about his observed measurement outcomes but does not have non-trivial quantum control over Wigner's full lab, we can simulate Ursula's operation of ``asking Wigner about his outcome" as an operation that acts directly on Wigner's system (Friend's lab) and not on his memory (as discussed earlier in this section). This simulation would produce the same setting-conditioned predictions in the augmented circuit as the original scenario.

Therefore $(\W,\U), (\U,\W)$ $\in n\matholdcal{SA}$ while $(\F,\U)\not\in n\matholdcal{SA}$ and $(\W,\U)\in n\matholdcal{SA}$, and of course $(\F,\F), (\W,\W), (\U,\U) \in n\matholdcal{SA}$ since no agent acts on their own memory through the channel $\matholdcal{E}_i$ that they implement after their measurement.

In the (entanglement version of the) FR protocol reviewed in \Cref{sec: FR_ent_review}, we have $n\matholdcal{SA}=\{(\A,\B),$ $(\B,\A),$ $(\U,\W),$ $(\W,\U),$ $(\U,\B),$ $(\B,\U),$ $(\W,\A),$ $(\A,\W),$ $(\U,\A),$ $(\W,\B),$ $(\A,\A),$ $(\B,\B),$ $(\U,\U),$ $(\W,\W)\}$, but this is not a standard quantum experiment as the pairs $(\A,\U)$ and $(\B,\W)$ don't appear in $n\matholdcal{SA}$, the latter agent acts as a super-agent to the former. The $n\matholdcal{SA}$ of the prepare and measure version of the FR protocol is identical.

We are now ready to formally define what constitutes a standard quantum experiment.

\begin{definition}[Standard quantum scenario]
\label{def:std_q_exp}
    An EWFS with $N$ agents $\{\A_1,...,\A_N\}$ is said to correspond to a standard quantum scenario if $(\A_i,\A_j)\in n\matholdcal{SA}$ for all $\A_i$, $\A_j \in \{\A_1,...,\A_N\}$. In case this holds for a subset of the first $k$ agents $\{\A_1,...,\A_k\}$ (who act at time steps $t_1<...<t_k$) in the context of a larger $N$-agent EWFS we call $\{\A_1,...,\A_k\}$ a standard quantum sector of the EWFS. Otherwise we call it a non-standard scenario or sector accordingly.
\end{definition}

For instance, in FR's protocol, the agents Alice and Bob who act first form a standard quantum sector, but once we include the super-agents Ursula and Wigner, the sector is no longer standard.

{\bf Non-standardness as a signature of genuine Wigner’s Friend-ness?} An interesting question is whether scenarios that are non-standard according to our \Cref{def:std_q_exp} can be considered as exhibiting a genuinely Wigner's Friend-type aspect. One intuition in favour of this is that by construction, non-standardness captures that there is at least one pair of agents such that one (say $\A_j$) acts on the memory of another (say $\A_i$) in a non-trivial manner that cannot be regarded as $\A_j$ simply ``asks'' $\A_i$ their outcome. However, this point needs to be further investigated both at a conceptual and technical level to make conclusive statements, after all there can be different equally well-motivated criteria for ``genuineness'' of a non-classical resource, as the vast literature on quantum non-locality and entanglement highlights.

Our formalism provides a first consistent, general and fully formal platform for formulating and investigating such questions for EWFS. This gives the potential to understand, in terms of fundamental principles and axioms, the quantum resource associated with Wigner's Friend type no-go results that fundamentally distinguish them from no-go results (such as Bell's theorem) where agents are not treated quantum mechanically. We leave this for future work.

\subsection{Recovering Heisenberg cut independence in standard quantum experiments}

In the previous sections of the paper, we have shown that the choice of settings (which provide a circuit-based formalisation of the choices of Heisenberg cuts) do affect the empirical predictions in Wigner's original scenario as well as its extensions such as FR and LF. That is, the \hyperref[def: I_assump]{$\mathbf{I}$} assumption (setting-independence) is violated here. This differs from our standard intuitions and usage of quantum theory to describe realistic experiments, where we do not have to consider any such settings or Heisenberg cuts, and where observed outcomes appear to be objective records independent of any such concepts.

In order to show that our framework correctly reproduces the known predictions and observations of standard quantum experiments conducted so far, we must show that the setting-variables can be safely dropped from the predictions and the augmented circuit when analysing such experiments. The following results formalise this intuition concretely using our definition of standard quantum scenarios. At a foundational level, this will shed light on how the perceived objectivity or non-relationalism of observed measurement events emerges within standard quantum sectors, where the agents involved do not have full quantum control over each other's labs/memories.

\begin{restatable}[Non-action on memory and setting-independence]{theorem}{SettingIndep}
\label{theorem: setting_independence}
Consider an EWFS and a subset $\A_{\matholdcal{K}}$ of agents therein. Suppose that $\A_i\not\in\A_{\matholdcal{K}}$  is another agent in the EWFS such that no agent in $\A_{\matholdcal{K}}$ acts as a super-agent to $\A_i$  i.e., $(\A_i,\A_k)\in n\matholdcal{SA}$  $\forall A_k\in A_{\matholdcal{K}}$. Then for every partition $A_{\matholdcal{K}}=\{\A_{j_1},...,\A_{j_p}\}\cup \{\A_{l_1},...,\A_{l_q}\}$ of $\A_{\matholdcal{K}}$, the setting-conditioned prediction $P(\vec{a}_j=\vec{\aaa}_j|\vec{a}_l=\vec{\aaa}_l,\vec{x}=\vec{\xxx})$ is independent of the setting $x_i$ that is,

\begin{align}
    \begin{split}
    \label{eq: setting_indep}
     &P(\vec{a}_j=\vec{\aaa}_j|\vec{a}_l=\vec{\aaa}_l,\vec{x}=\vec{\xxx})=\\   &P(\vec{a}_j=\vec{\aaa}_j|\vec{a}_l=\vec{\aaa}_l,\vec{x}\backslash x_i=\vec{\xxx}\backslash \xxx_i),  
    \end{split}
\end{align}

where $\vec{x}\backslash x_i=(x_1,...,x_{i-1},x_{i+1},...,x_N)$ is the setting vector with the i$^{th}$ component omitted, and similarly for the vector of setting values $\vec{\xxx}\backslash \xxx_i$. This expression is equivalent to the conditional independence given in \cref{eq: causality_indep}.
\end{restatable}

A proof of the above theorem can be found in \Cref{appendix: proofs}. As a corollary of \Cref{theorem: setting_independence}, we can immediately recover the full setting-independence of predictions in standard quantum experiments.
\begin{corollary}[Full setting-independence in standard quantum theory]
\label{corollary: std_QT_setting_indep}
    In any EWFS corresponding to a standard quantum scenario, every non-trivial setting-conditioned prediction, is setting indepndent. Formally, for any disjoint sets $\vec{a}_j$ and $\vec{a}_l$ of outcomes in the EWFS, $P(\vec{a}_j=\vec{\aaa}_j|\vec{a}_l=\vec{\aaa}_l, \vec{x}=\vec{\xxx})$  is independent of settings $x_i$ for all $i\not\in\{j_1,\dots,j_p,l_1,\dots,l_q\}:= \matholdcal{JL}$ i.e., \cref{eq: setting_indep} holds for all such $i$.
   Specifically, in such standard scenarios, all non-trivial predictions i.e., those where $a_i\neq \perp$ for all $i\in \matholdcal{JL}$, can equivalently be expressed in a fully setting-independent manner as given below. 
    \begin{align}
    \label{eq: std_QT_corollary}
        \begin{split}
&P(\vec{a}_j=\vec{\aaa}_j|\vec{a}_l=\vec{\aaa}_l, \vec{x}=\vec{\xxx}) \\=& P(\vec{a}_j=\vec{\aaa}_j|\vec{a}_l=\vec{\aaa}_l, x_i=1, \forall i\in \matholdcal{JL})\\:=&P(\vec{a}_j=\vec{\aaa}_j|\vec{a}_l=\vec{\aaa}_l). 
        \end{split}
    \end{align}

\end{corollary}

Going beyond predictions, and to the underlying circuit representation, we expect that standard quantum experiments involving some measurements and channels, can be represented in terms of quantum circuits with no ambiguity in how a measurement is modelled (i.e., circuits with no setting variables for measurements). Here, two equivalent types of circuit representations are possible, which are both commonly used within the standard quantum computing paradigm: we can either model each measurement as acting only on the measured system and yielding a (non-trivial) classical outcome at the time of the measurement, or we can model each measurement as a unitary interaction between a system and ancilla at the time of measurement and probabilities can be extracted by measuring all the ancillas at a later time at the end of the experiment. 

 We leave a formal definition of these to \Cref{appendix: standard_circuits}, where we prove that for EWFSs corresponding to standard quantum scenarios, the augmented circuit can equivalently be reduced to either of these expected forms. These two types of standard quantum circuits are illustrated in \cref{fig:circuit_sys_form} and \cref{fig:circuit_sys+anc_form} in the same appendix.

\section{Discussion}
\subsection{Sound scientific reasoning in EWFSs and analogies to classical multi-agent reasoning}\label{sec: reasoning_rules}
The results of \Cref{sec:gen_framework} guarantee that our formalism allows quantum agents to make predictions and to reason consistently about experiments and each other's knowledge, even if unitary quantum theory is universally valid. Here, we illustrate a general paradigm for scientific reasoning indicated by our results, which allows one to straightforwardly extend our resolution to multi-agent scenarios with partial knowledge. We do so while contrasting quantum and classical aspects and showing how consistency questions in EWFSs reduce to analogous issues in classical multi-agent reasoning (see also \cite{Jones2025} for comparisons between Wigner’s Friend (WF) and classical paradoxes).

{\bf Genuinely quantum and relational aspects} Our results (\Mref{def: conv_prediction}, \Mref{theorem: main}) provide a concrete default rule for choosing settings to compute predictions, which makes them equivalent to conventional quantum predictions for EWFSs used in the literature, while avoiding paradoxes by making the conditioning on settings explicit. Operationally, the rule says that when computing the outcome probabilities of a subset of measurements, exactly those measurements are treated as producing ``classical'' records (setting 1) and the rest are modeled purely unitarily (setting 0), which can be interpreted as a choice of Heisenberg cuts (outside vs.\ inside). When agents reason about others, the chosen cuts depend on the agent who is reasoning and the agents whose classical outcomes they are reasoning about: when $A$ reasons about $B$’s outcome (and not $C$’s), then $A$ will place $A$’s and $B$’s measurements on the classical side of the cut and $C$’s on the quantum side; but when $A$ reasons about $C$’s outcome, $C$’s measurement will be placed on the classical side. This captures relational, non-absolute events and the setting (Heisenberg-cut) dependence of predictions. Moreover, genuinely WF-like scenarios require nontrivial super-agent actions, often linked to measurement complementarity. This suggests a non-classical resource tied to ``genuinely WF-like aspects'' yet to be fully characterized in future work (see also \Cref{sec:super-agent,sec:Discussion}).

{\bf Aspects reducible to classical issues} By formalizing Heisenberg cuts as channel choices in a circuit, classical probability and logic already dictate conditioning on these choices unless (1) all agents employ the same fixed choices or (2) the choices are proven irrelevant. Even classically, two physicists who assume different noise models $N_A, N_B$ for the same known channel (and applied to the same input state) would generally reach different predictions unless they communicate and condition on their assumptions, and they could obtain FR-like apparent inconsistencies if they ignore such parameters (see \Cref{appendix: classical_example} for a classical scenario that reproduces certain aspects of the FR paradox). On the other hand, the conclusions are independent of whether the analysis was done during the day or night, hence one can safely ignore this information.

{\bf General reasoning paradigm} These examples yield a simple recipe for sound scientific reasoning in any multi-agent setting:
\begin{enumerate}[left=0pt]
    \item Identify common knowledge vs.\ variable parameters that are not common knowledge to all agents;
    \item Drop parameters shown to be redundant;
    \item Fix the remaining parameters (according to some well-defined rules) and condition on them;
    \item Test choices against data; update if falsified.
\end{enumerate}
In our framework from \Mref{sec:gen_framework}, the EWFS protocol description (\Cref{def:LWFS}) is common knowledge of all agents; the only variables outside common knowledge are the settings $\vec{x}$ (the case of partial knowledge of the protocol is discussed shortly). \Cref{theorem: main,theorem: setting_independence} give criteria—via operational causal structure and the non-super-agent structure—for when settings are redundant; the default rule (entailed by \Cref{def: conv_prediction}, \Cref{theorem: main}) then fixes the rest and conditions the predictions on them to ensure consistency per \Cref{def: consistency_pred}.

Finally, our discussion (\Cref{sec: setting_interpretation,appendix:relation_prev_works}) shows how setting-conditioned predictions can be falsified. The default rule for choosing the settings assumes universal validity of unitary quantum theory; if the existence of an objective-collapse mechanism is confirmed in some future loophole-free experiment, it would falsify predictions made under universal quantum theory and mandate updating certain settings from $0$ to $1$.

{\bf Efficiency of reasoning} Computing setting-conditioned predictions requires applying the Born rule to the same-size circuit as the underlying protocol, i.e., no worse than standard quantum reasoning. Processing settings adds no asymptotic overhead: one reads independences from the causal and non-super-agent structures (these are precomputable) \cite{Note10}. The only extra cost is communicating/storing a short binary vector of non-redundant settings—already necessary in analogous classical ambiguities. Our framework also satisfies \hyperref[def: assump_C]{$\mathbf{C}$}: communicated setting-conditioned statements need not be recomputed by recipients, while relationality appears through agents choosing different settings under the default rule.

{\bf Scope of our resolution and partial knowledge} We assumed that the EWFS protocol (cf. \Mref{def:LWFS}) is common knowledge of the agents involved, in order to isolate what is genuinely WF-like. Extending to cases where agents have partial knowledge of the protocol is directly analogous to how this would be handled in classical multi-agent reasoning, and is based on the same general reasoning paradigm outlined above: fewer elements of common knowledge mean more variables to condition on to avoid inconsistencies, not new quantum-specific pitfalls. The same rules apply; we omit the (straightforward) explicit algorithm.

{\bf Status of proposed reasoning challenges}
Our work fully resolves FR-type paradoxes arising from the genuinely non-classical aspect of AoE violation (\Cref{sec: resolution_entanglement,ssec: resolution_prep}), and suggests that any remaining logical challenges for quantum agents in EWFSs (including partial-knowledge scenarios) are analogous to classical challenges or inconsistencies arising from failing to properly condition on assumptions outside common knowledge. For instance, \cite{Renner_Challenge} proposed such a challenge, inspired by the FR arguments. However their arguments, like FR, involve ambiguities about Heisenberg cuts, measurement–channel modeling, and knowledge conditioning in computing probabilities. Similar ambiguities cause inconsistencies even classically (see also \Cref{appendix: classical_example} for analogies to classical paradoxes arising from ignoring conditioning on relevant knowledge). We believe that our work therefore also fully resolves the challenge of \cite{Renner_Challenge}. Indeed, to our knowledge, there exists no EWFS counterexample (including partial-knowledge scenarios) that entails a paradox which cannot be resolved by applying our framework and general reasoning paradigm.

\subsection{Interpretations of quantum theory}
\label{sec: interpret_indep}
We developed our framework in three levels of generality which enable flexibility across interpretations and conceptual viewpoints. (i) A general EWFS formalism that allows for different rules for computing probabilities $P(\vec{a}_j|\vec{a}_l,k)$ of some outcomes $\vec{a}_j$, $\vec{a}_l$, given certain parameters $k$ describing the assumptions made about the scenario (such as its states and channels, etc.). (ii) Specializing these rules such that $k$ refers to the settings $\vec{x}$ of the augmented circuit still leaves freedom to choose the setting values. Different interpretations suggest different $k$ and settings in the above two levels e.g., Bohmian mechanics can include the hidden variables $\lambda$ in the $k$ parameters; collapse theories would fix $x_i=1$ for all measurements, while many-worlds would allow $x_i=0$ (a purely unitary description of measurements). Our consistency theorem, \Cref{theorem: main}, applies to \emph{all} setting choices, so logical/probabilistic consistency is achievable across different interpretations. We discuss how many-worlds, collapse, hidden-variable theories, relational quantum mechanics, and QBism could consistently apply our framework in \Cref{appendix:relation_prev_works}.

(iii) At the last level, we prescribe reasoning rules assuming universal quantum theory (this excludes some interpretations such as collapse theories), but they remain compatible with others—such as many-worlds, QBism, or relational quantum mechanics—since applicability does not hinge on whether states or channels are ontic or epistemic. Related approaches such as consistent histories \cite{Griffiths1984,GriffithsCH}, decoherence \cite{Zukowski2021}, and the parsing rule of \cite{Renes2021} are compared in \Cref{appendix:relation_prev_works}.

\subsection{Physical interpretation of settings}
\label{sec: setting_interpretation}
Settings already play a necessary role in the quantum predictions of an EWFS, which are prerequisite for both FR- and LF-type arguments: one must specify whether a measurement is modeled as unitary or as producing classical records to apply the Born rule \cite{Note11}, and the ``predictions of quantum theory'' used in the EWFS literature exactly match particular setting-conditioned predictions (\Cref{sec: consistency_pred}). This applies to FR-type reasoning and to LF-type arguments about AoE (\Cref{sec: resolution_entanglement}); our framework makes these assumptions explicit. We discuss physical interpretations of these settings, noting that this can in turn depend on the interpretation of quantum theory to which one subscribes.

{\bf Heisenberg cuts and classical records}
Assuming quantum theory is complete (no hidden variables beyond the operational description of the quantum states, measurements, and channels in a protocol/experiment), setting choices correspond to choices of Heisenberg cuts: $x_i=0$ treats the lab $\M_i$ as purely quantum or inside the cut, $x_i=1$ treats it as producing a classical record (in the measurement basis) or outside the cut. This gives a relational model of quantum measurements, as the choice of Heisenberg cut or settings can generally depend on the prediction being considered, or the perspective of an agent (as also discussed in \Mref{sec: resolution_entanglement}, \Mref{sec: reasoning_rules}). This remains consistent with the universal validity of unitary quantum theory; $x_i=1$ does not mean an objective collapse: projectors enter only to compute probabilities (via the Born rule) about the referenced classical record (see also \Mref{remark: no_collapse}).

{\bf Falsifying setting choices through experiments}
Every EWFS equivalently specifies a unique augmented circuit (\Cref{sec:gen_framework}). The settings in the augmented circuit do not correspond to choices of physical operations actually performed in a protocol, but (i) if unitary quantum theory is universally valid, they encode choices and assumptions—such as how measurements are modeled—that enter predictions and agents' reasoning (these choices can be relational and agent-dependent); or (ii) in theories that reject universal unitary evolution, they encode fundamental (often absolute for all agents) macroscopic dynamics of the theory. For example, in Wigner’s original setup, collapse models imply $x_A=1$ for the friend (corresponding to a collapse mechanism producing classical records), while many-worlds suggests $x_A=0$ (purely unitary dynamics), leading to distinct predictions for Wigner's outcome $P(w|x_A)$. Which choice nature prefers is ultimately empirical, and the fact that these predictions differ leaves open the possibility of a future experiment that could falsify either $P(w|x_A=0)$ or $P(w|x_A=1)$. This highlights that setting choices do have physical and empirical consequences in EWFSs. By explicitly keeping track of the settings and their dependence, our framework provides a way to clearly test and understand which predictions in an EWFS can be falsified and under what assumptions.

{\bf Time dependence of setting choices and knowledge update}
The falsification of predictions in an EWFS prompts an update of our knowledge about the underlying setting choices. This requires setting dependence, since setting-independent predictions would not allow any inference about the settings. While knowledge updates in classical theories concern “absolute facts of the matter” that exist but about which one may lack knowledge, in our quantum EWFS formalism the update concerns an agent’s beliefs about the existence of classical measurement records, which may no longer be absolute facts. For example, at present we can safely treat all real-world measurements as setting 1, as there is no reason to believe there currently exists a super-agent capable of performing quantum operations on an entire lab, needed to erase another agent’s records. This suggests another operational interpretation: setting 1 reflects the belief that no super-agent can jointly act on all relevant quantum degrees of freedom (d.o.f.) of some agent’s (Alice’s) lab, ensuring the persistence of Alice's measurement records in the d.o.f. not being acted upon (see also \Cref{sec:super-agent}). But if such “super-agent” devices were developed, we would need to classify some measurements as setting 0 to explain experiments in which records can be undone (e.g., if a sufficiently loophole free violation of an LF inequality \cite{Bong2020} were achieved), and update previous beliefs about setting 1. This aligns with a decoherence-based interpretation, where classical records emerge from interactions with an environment. However, there is a crucial distinction, as those interpretations do not typically allow the possibility that the environment could later be under the quantum control of a super-agent, or the resulting time dependence with regard to the persistence of classical records. 

{\bf Analogy to Maxwell’s demon paradox}
Maxwell’s demon paradox \cite{Knott1911} in thermodynamics is resolved by clearly adopting a single perspective (demon vs.\ macroscopic) when computing entropies, accounting for information erasure, and keeping track of which perspective was used \cite{Landauer1961, Bennett1982}. Likewise, our \Cref{corollary:setting_independence} shows that in EWFSs, making the choice of settings clear and conditioning on them (which model Heisenberg cuts, and consequently different agents' perspectives) is key to obtaining consistent predictions and reasoning; combining statements made across different setting choices and then ignoring them by assuming \hyperref[def: I_assump]{$\mathbf{I}$} creates only an \emph{apparent} paradox that is easily resolved, as per the above analogy to Maxwell's demon.

\section{Conclusions and outlook}\label{sec:Discussion}

\subsection{Conclusions}

Here we provide a high-level overview of our main results, emphasizing the key conclusions (complementing the more technical summary in \Cref{sec:Overview of contributions and paper layout}).
In this work, we develop a comprehensive framework for Extended Wigner’s Friend Scenarios (EWFSs), overcoming a number of challenges posed by previous no-go arguments for EWFSs \cite{Frauchiger2018, Brukner2018, Bong2020, Cavalcanti2021}.

{\bf Full resolution of FR paradoxes and reasoning challenges in EWFSs}
Our main theorems show that the framework guarantees logical and probabilistic consistency in all EWFSs and preserves causality principles, while cleanly separating objective features (e.g., causal structure) from subjective ones (e.g., predictions and agents’ knowledge). As a central application, this resolves all FR-type logical paradoxes in quantum theory and the challenge posed in \cite{Renner_Challenge}, without requiring any departure from standard quantum or probabilistic theories. The resolution holds independently of one’s preferred interpretation of quantum theory (cf. \Cref{appendix:relation_prev_works}), and even without invoking the projection postulate or state-update rules—while still allowing measurements to be modeled unitarily and using the standard Born rule (cf. \Cref{remark: no_collapse}).

{\bf Identifying the root cause of paradoxes}
Predictions in EWFSs change with the choice of Heisenberg cut (i.e., whether a measurement is modeled as producing a classical record or as a purely unitary evolution, captured by our Assumption \hyperref[def: I_assump]{$\mathbf{I}$}). Mixing statements made under different choices leads to apparent contradictions and corresponds to imposing \hyperref[def: I_assump]{$\mathbf{I}$} in a scenario that violates it.

{\bf Reasoning rules for quantum agents}
To avoid such paradoxes, we provide physically motivated rules for agents’ reasoning that are compatible with (i) universal quantum mechanics, (ii) classical logic for recorded outcomes, and (iii) the fundamentally relational, perspective-dependent character of EWFSs. This constructively shows that quantum theory remains self-consistent in any EWFS—there is no threat to consistent scientific reasoning or to programming future quantum devices that act as agents—contrary to the broader claims made in \cite{Frauchiger2018}.

{\bf Emergence of objectivity}
Our framework explains how objective, Heisenberg-cut–independent predictions emerge in real-world experiments and within specific sectors of an EWFS (formalizing the idea of a “Wigner Bubble” \cite{Cavalcanti2021}), thereby accommodating the non-absoluteness of events and relationalism in general EWFSs while remaining consistent with the objective observations seen so far.

\subsection{Outlook}
We now outline some future directions based on our work.

{\bf Local-Friendliness and non-absolute events}
The Local-Friendliness (LF) theorem \cite{Bong2020} shows that no theory simultaneously satisfying the absoluteness of observed events (AoE), together with certain free choice and locality principles, can reproduce specific quantum predictions in an EWFS. Using our framework, we analyze the LF scenario in forthcoming work \cite{LF_Vilasini_Woods}, highlighting the distinct role of setting/Heisenberg-cut dependence (i.e., violation of \hyperref[def: I_assump]{$\mathbf{I}$}) relative to the FR scenario and linking it to AoE violation. Based on this, it would be interesting to probe the possibility of characterizing novel EWFSs and no-go theorems for AoE in combination with other fundamental principles and quantum phenomena—including complementarity, contextuality, relativistic causality, and protocols where agents perform quantum operations in superpositions of causal orders \cite{Note12}.

{\bf Quantum and relativistic causality in EWFSs}
Bell’s theorem shows that quantum correlations challenge explanations based on classical causal models \cite{Pearl2009}, if we wish to preserve relativistic causality and free choice—motivating the development of quantum causal modeling frameworks (e.g., \cite{Henson2014,Costa2016,Allen2017, Barrett2020A}). The LF theorem poses even deeper challenges \cite{Cavalcanti2021, Ying2023}, even for quantum causal modeling, since these formalisms typically assume AoE. Our framework allows one to relinquish AoE, yet still recovers all (possibly subjective) predictions in an EWFS within a single objective causal structure for the protocol: the augmented circuit is acyclic and embeddable in spacetime, preserving relativistic causality. The only subjective element concerns agents’ priors over settings (Heisenberg cuts). This paves the way for a relational, operational quantum causal modeling framework that does not assume AoE and remains compatible with free choice and relativistic constraints. Extending this to scenarios involving cyclic quantum circuits (and thus protocols with indefinite order of operations) \cite{Hardy2005, Oreshkov2012, Chiribella2013,Araujo_2017,Barrett2020, ferradini2025c, ferradini2025q} which have also been linked to relativistic causality \cite{VilasiniColbeckPRA,VilasiniColbeckPRL,VilasiniRennerPRA, VilasiniRennerPRL} could provide a route toward a unified treatment of quantum agents, quantum causality, and spacetime structure.

{\bf EWFSs beyond quantum theory}
EWFSs have also been formulated and studied beyond quantum theory \cite{Vilasini_2019, ormrod2023, NurgalievaVilasini}. In \cite{Vilasini_2019, NurgalievaVilasini}, agents’ memories/labs were modeled within general theories, distinguishing theory-independent analogues of projective vs. unitary-like (information-preserving) updates. FR-type apparent paradoxes can then arise with PR-box resources \cite{PopescuRohrlich1994, Vilasini_2019}, and more generally in any theory admitting certain forms of contextuality \cite{NurgalievaVilasini}. Moreover, \cite{ormrod2023} shows that any theory allowing information-preserving updates, Bell non-classicality, and a locality principle leads to AoE violations akin to the quantum LF scenario \cite{Bong2020}. Thus, the measurement problem and FR-type paradoxes are not uniquely quantum. Generalizing our framework beyond quantum theory—building on these works—could resolve multi-agent paradoxes and enable a consistent calculus for reasoning in post-quantum settings, while clarifying the status of AoE and the measurement problem therein.

{\bf Resource-theoretic characterization of EWFSs}
What informational and non-classical resources make FR- and LF-type results possible? What sets genuine Wigner’s Friend scenarios apart from standard experiments where agents are effectively classical? Our \Cref{theorem: main,theorem: setting_independence} gives structural conditions relevant to such questions—under which absolute and Heisenberg-cut–independent “classical” records emerge for all agents. Our work identifies the violation of \hyperref[def: I_assump]{$\mathbf{I}$} by a scenario’s predictions as a key signature of WF-like (as opposed to standard) quantum experiments. However, there is room to characterize necessary and sufficient conditions that depend on the properties of states and measurements (e.g., contextuality, complementarity, etc.) for violating \hyperref[def: I_assump]{$\mathbf{I}$} and consequently AoE (cf. \Cref{sec: AoE}). More broadly, the non-super-agent structure introduced here to distinguish standard quantum scenarios from genuine Wigner’s Friend scenarios motivates a new research program: developing a systematic resource theory of EWFSs (as discussed in \Cref{sec: standard_QT})—a direction that remains to be systematically explored.

\begin{acknowledgements}
  We thank Renato Renner, Victor Gitton, Y\`{i}l\`{e} Yīng, Marina Maciel Ansanelli, Joe Renes, Eric Cavalcanti and Nuriya Nurgalieva for insightful discussions.
  V.V. acknowledges support from an ETH Postdoctoral Fellowship, and from the PEPR integrated project EPiQ ANR-22-PETQ-0007 as part of Plan France 2030. M.P.W. acknowledges funding from the Swiss National Science Foundation (AMBIZIONE Fellowship, No.~PZ00P2\_179914). Both authors acknowledge support from NCCR QSIT. 
\end{acknowledgements}

\newpage
\onecolumngrid
\appendix
\newpage

\begin{center}
    {\bf \large{APPENDIX}}
\end{center}
\section{Generality of the definition of EWFSs}
\label{appendix: generality_EWFS}

Here we discuss the justification for the generality of \Cref{def:LWFS} of an Extended Wigner's Friend Scenario that we have proposed in this paper. The idea is that \Cref{def:LWFS} encompasses all finite multi-agent quantum protocols where agents' memories (in which they store the measurement outcome) are modelled as quantum systems, and where one agent can have full quantum control over the labs (measured system and memory) of other agents in the scenario. Here the finiteness applies both to the Hilbert space dimensions and the number of information-processing steps. 

Generically, we can model such scenarios by considering a set of $N$ agents $\mathtt{A}=\{\A_1,\ldots,\A_N\}$ and a set of $m$ systems $\mathtt{S}=\{\Ss_1,\ldots,\Ss_m\}$ under study and a set $\mathtt{M}:=\{\M_1,\ldots,\M_N\}$ of systems, one for each agent $\A_i$ which models their memory where they store the outcomes of measurements that they perform.

For simplicity but without loss of generality, we take the lab of each agent $\A_i$ to consist of the system $\mathtt{S}_i$ that they measure, along with their memory $\M_i$.  
Each agent $\A_i\in \mathtt{A}$ performs a measurement $\matholdcal{M}^{\A_i}$ on some subset $\mathtt{S}_i\subseteq \mathtt{S}\cup \mathtt{M} \backslash \{\M_i\}$ of systems, which can include the memories of other agents, and stores the outcome of the measurement in their memory $\M_i$ \cite{Note13}. Note that $\mathtt{S}_i$ is a subset of systems and memories of other agents, which allows each agent to possibly act as a super-agent to any subset of the other agents by measuring their memories. We have no loss of generality in assuming that each agent performs one measurement, because any scenario where one agent performs multiple measurements can be bought to this form by modelling them as multiple agents, each performing one measurement.

Next, also without loss of generality, we can assume that each agent $\A_i$ acts at a distinct time step $t^i$ with $t^i<t^j$ for $i<j$, since any physical scenario where the same agent acts at different times can equivalently be modelled in terms of multiple agents each acting at distinct times (since our definition allows for communication channels that carry the relevant information between the time steps). Further, we can model each measurement $\matholdcal{M}^{\A_i}$ as acting on the whole set $\mathtt{S}$ of systems and the set of all memories $\M_1,\ldots,\M_{i-1}$ of previous agents since any operation on a subset $\mathtt{S}_i$ of systems can be trivially enlarged into an operation on all systems by appending the identity on the complementary set of systems. Agent $\A_i$ may also perform certain fixed transformations (corresponding to quantum channel $\matholdcal{E}^i$ on $\mathtt{S}\cup \mathtt{M}$) in between time steps $t_i$ and $t_{i+1}$, for instance, agent $\A_1$ may measure the system $\Ss_1\in \mathtt{S}$ and  depending on their measurement outcome, could perform a different transformation on some initial state of $\Ss_2$ which they send to agent $\A_2$ who could then measure $\Ss_2$.

The circuit of \cref{fig: genform_circuit} illustrates this general form of an EWFS that we consider, where $\{\matholdcal{E}^i\}_{i=1}^N$ are fixed transformations that all agents agree on as part of the protocol while $\{\matholdcal{M}^{\A_i}\}_{i=1}^N$ denote the measurements, one for each agent.

Furthermore, consider that an agent $\A_i$ measures a subset $\mathtt{S}_i$ of the systems through a measurement $\matholdcal{M}^{\A_i}$. Let $a_i$ be the random variable associated with the measurement outcome, which takes values $\aaa_i$ in the set $\{0,1,...,d_{\mathtt{S}_i}-1\}$, where $d_{\mathtt{S}_i}$ is the Hilbert space dimension (assumed to be finite) of the system $\mathtt{S}_i$ measured by $\A_i$.
Without loss of generality, we can model this as a projective measurement, since any measurement on finite dimensional systems can be purified to a projective measurement involving rank 1 projectors on a larger set of systems through Neumark dilation; see e.g. \cite[Sec. 9-6]{AsherPeres}. Explicitly, we can use the projective measurement $\{\pi_{\aaa_i}^{\mathtt{S}_i}=\proj{\aaa_i}_{\mathtt{S}_i}\}_{\aaa_i\in\mathtt{O}_i}$, where $\{\ket{\aaa_i}_{\mathtt{S}_i}\}_{\aaa_i\in\mathtt{O}_i}$ forms an orthonormal basis of $\mathtt{S}_i$. While this looks like a computational basis due to the choice of outcome labels $\{0,1,...,d_{\mathtt{S}_i}-1\}$, the choice of which basis to associate with these labels is arbitrary and therefore the measurement may correspond to an arbitrary orthonormal basis.

Together with all these simplifications, we arrive at the general form of \Cref{def:LWFS}.

\section{Distinguishing predictive and observational statements: refining consistency}
\label{appendix: observational_statements}

In the main text, we used $\Sigma$ to refer to the set of all statements obtained from predictions in an EWFS, in \Cref{definition:statement_prediction}. We motivated the need to distinguish such predictive statements from observational statements through examples of simple classical scenarios. Here we define observational statements and discuss how our results immediately generalise to ensure consistency of predictive and observational statements together. To make the distinction clear, in this section, we will explicitly write $\Sigma_{pred}:=\Sigma$. 
\begin{definition}[Observational statements]
\label{def: obs_statement}
    Observational statements are statements in the set $\Sigma_{obs}:=\{$ ``Based on observation, I am certain that the outcomes $\vec{a}_j$ takes values $\vec{\aaa}_j$.'' $\}_{\vec{a}_j,\vec{\aaa}_j}$. In logical notation, we will denote elements of the set as $\vec{a}_j=\vec{\aaa}_j|_{obs}$.
\end{definition}

\begin{definition}[Consistency of observational statements]
\label{def: consistency_obs}
   A set $\Sigma_{obs}$ of observational statements is consistent iff $\vec{a}_j=\vec{\aaa}_j|_{obs}\in \Sigma_{obs}$ implies that $\vec{a}_j=\neg\vec{\aaa}_j|_{obs}\not\in \Sigma_{obs}$. 
\end{definition}
\begin{definition}[Global consistency]
\label{def: global_consistency}
    A set $\Sigma_{pred}\cup \Sigma_{obs}$ of predictive and observational statements obtained in an EWFS, is said to be globally consistent if $\Sigma_{pred}$ and $\Sigma_{obs}$ are consistent according to \Cref{def: consistency_pred,def: consistency_obs} respectively, \emph{and} additionally, $S:=\vec{a}_j=\vec{\aaa}_j|_{obs}\in \Sigma_{obs}$, then $S'\in \Sigma_{pred}$ where $S'$ is a statement associated with a prediction $P(\vec{a}_j=\vec{\aaa}_j|k=\kkk)>0$.
\end{definition}

Recall that within our framework the scenario parameters $k$ are instantiated by the setting vector $\vec{x}$, which give us setting-conditioned predictions e.g., $P(\vec{a}_j=\vec{\aaa}_j|\vec{x}=\vec{\xxx})$.

In any EWFS, a given choice of settings $\vec{x}=\vec{\xxx}$ allows one to fix all the channels involved the scenario. Generically, the choice of how these setting values must be chosen is given by some set of reasoning rules $\matholdcal{R}$. This rule may either provide an absolute choice of setting values $\vec{x}=\vec{\xxx}$ that must be applied to every prediction in the scenario, or it may provide a different choice of setting values relative to each prediction one wishes to compute. An example of the former would be collapse theories or any interpretation that rejects universal validity of unitary quantum theory, which would require all settings to be 1. An example of the latter is our reasoning rule given by the setting choices illustrated in the completeness result of \Cref{theorem: main}: in every prediction, e.g., $P(\vec{a}_j=\vec{\aaa}_j|\vec{x}=\vec{\xxx})$, the settings for all outcomes that appear in the probability are set to 1, and the settings for the remaining measurements are set to 0. 

Now, given an EWFS together with a set of rules $\matholdcal{R}$ for choosing the settings in the augmented circuit, we can consider what happens when we design an experiment to observe the outcome referred to in a prediction. Let us do so with a simple example, referring back to Wigner's original thought experiment, where Alice was the agent and Bob the super-agent. When considering the probability of Bob's outcome $b$, we have a choice for Alice's setting $x_A$, and can compute for instance $P(b|x_A=0)$. If it were possible to physically perform a Wigner’s Friend type experiment involving these agents, and if indeed unitary quantum theory were universally valid, then Bob's observations regarding $b$ would be consistent with this prediction in the sense of \Cref{def: global_consistency}: if these premises are satisfied, then Bob can observe $b=\bbb$ only if $P(b=\bbb|x_A=0)>0$. On the other hand, if unitary quantum theory was not universally valid, but there was an additional (yet undiscovered) physical mechanism for objective collapse, then Bob's physical observations would be consistent with $P(b=\bbb|x_A=1)$ (according to \Cref{def: global_consistency}), and need not be globally consistent relative to $P(b=\bbb|x_A=0)>0$. Note that in these discussions, we have omitted $x_B$, which will by default be $x_B=1$ here since these predictions refer to a non-trivial outcome of Bob. 

We have proven in our framework that (setting-conditioned) predictions obtained under any possible rule $\matholdcal{R}$ for choosing the settings are mutually consistent. Moreover, as the above example illustrates, when considering observational statements, it is important to consider the physical dynamics leading to the said observations. 

From the consistency of the predictions, and the definition of global consistency, it is straightforward to see that if the physical dynamics leading to the observations respects the channel choices specified by the reasoning rules $\matholdcal{R}$ (and assuming that the physical probabilities also respect the Born rule relative to those channel choices), then the predictions and observations will be globally consistent. On the other hand, if the physical dynamics leading to the observations differs from the $\matholdcal{R}$ used to compute the predictions (e.g., when there is physical collapse for all measurements but the rules model certain measurements as pure unitaries), then it is possible to violate global consistency. This provides a way to operationally falsify the rule $\matholdcal{R}$ (see also \Cref{sec: reasoning_rules}), assuming that the experiment was performed in a faithful way, in the sense that indeed Bob's operation acted on Alice's whole lab and there was no unexpected information leakage which is not accounted for in the scenario description.

Considering conventional predictions used in FR and LF type arguments, we have seen that these imply a particular default rule $\matholdcal{R}^{def}$ for selecting the settings \Cref{theorem: main} (as also discussed in \Cref{sec: reasoning_rules}). In this rule, only the measurements whose outcomes appear in a given prediction are assigned projectors (in order to compute said probability through the Born rule), while all other measurements are modelled as pure unitary evolutions. Therefore, if unitary quantum theory were indeed valid universally, our consistency result of \Cref{theorem: main} (which only refers to predictive statements) is sufficient to guarantee global consistency of predictive statements obtained through this rule for a given EWFS together with any observations made within a hypothetical experimental realisation of that EWFS. Therefore, we have no loss of generality in restricting only to predictive statements in the main text. 

{\bf Consistency and agents' knowledge}
This distinction between predictive and observational statements emphasised here, provides a precise understanding of what does not does not constitute an inconsistency. The consistency for predictive statements requires that one should not be able to obtain two different probability assignments $P$ and $P'$ for the same outcomes in a scenario, conditioned on the same information. This only refers to predictions i.e., probabilities associated with running the protocol for several rounds. 

It is important to note that even in a consistent theory (such as purely classical physics) agents can nevertheless assign different probabilities to events due to different knowledge based on observations made in a given round. If Alice and Bob have a fair coin at hand, where Alice knows that the outcome of the coin flip is $c=heads$ in one round, she would assign probability 1 to $c=heads$ in that round while Bob would assign a uniform probability to $c=heads$ if he does not know the outcome. This is not a contradiction because Alice's probability relates to an observation in a particular round, her certainty in this case would be captured by an observational statement $c=heads|_{obs}$. This example leads to consistent predictions since both agents would predict a uniform probability for $c=heads$ if asked what the probability of heads will be over many coin flips. 

Furthermore, Alice may wish to make a prediction about the outcome $b$ of a bet given that she observed $c=heads$ (and possibly other information $k=\kkk$ that she may know about the scenario). This would then correspond to a prediction $P(b=win|c=heads, k=\kkk)$. On the other hand, Bob who does not know the outcome of the coin flip but has the same background information $k=\kkk$ about the scenario, would make a prediction $P(b=win|k=\kkk)$ which can be different from $P(b=win|c=heads, k=\kkk)$. This is also not a contradiction, since the two predictions are conditioned on different knowledge, and this conditioning is important to ensure consistency. It is immediate to see that even in these simple classical examples, ignoring the conditioning on agents' knowledge and/or the background assumptions they make about the scenario at hand, one can obtain apparent inconsistencies quite easily.

\section{Details of formal definitions and key properties of our framework}
\label{appendix: properties_general}

In this section, we provide full formal definitions of certain concepts (such as causal and non-super-agent structure of an EWFS, and the assumptions \hyperref[def: Q_assump]{$\mathbf{Q}$}, \hyperref[def: U_assump]{$\mathbf{U}$}, \hyperref[def: assump_C]{$\mathbf{C}$}, \hyperref[def: assump_D]{$\mathbf{D}$} and \hyperref[def: assump_S]{$\mathbf{S}$}) that were motivated intuitively in \cref{sec: consistency_pred,sec: consistency_reasoning,sec: standard_QT} of the main text, and thereby formally present a theorem which proves completeness, consistency and causality properties of our framework.

\subsection{Completeness, consistency and causality}
\label{appendix: completeness_causality}
As explained in \cref{sec: consistency_pred}, even though our augmented circuit formalism does not a-priori impose absoluteness of events, it can be shown that it provides a complete representation of all predictions that can be made in an EWFS in a way that is consistent and respects causality principles. Before stating the formal theorem to establish this, we define what we mean by causality principles.  
For this, we begin by defining a directed acyclic graph (DAG) \cite{Note14} that corresponds to every EWFS. This DAG captures the potential information flow within the EWFS protocol, adhering to the time-ordered sequence of operations, and representing an objective and agent independent  property of that EWFS. We will refer to this as the causal structure of the EWFS.

Recall that by definition, an EWFS with $N$ agents is modelled as having $2N$ operations w.l.o.g., as each agent performs a measurement $\matholdcal{M}^{\A_i}$ followed by some quantum channel $\matholdcal{E}_i$, both associated with a time step $t_i$. Let $\matholdcal{O}$ denote any one of these $2N$ operations we will denote by $\mathtt{S}_{\matholdcal{O}}$ the subset of all systems and memories $\mathtt{S}\cup\mathtt{M}$ such that the operation $\matholdcal{O}$ acts non-trivially on all systems in $\mathtt{S}_{\matholdcal{O}}$ and as the identity on the rest of $\mathtt{S}\cup\mathtt{M}$.

\begin{definition}[Causal structure of an EWFS]
\label{definition: causal_str}
The causal structure of an EWFS with $N$ agents is a directed acyclic graph (DAG) $G$ with the following properties
\begin{enumerate}
    \item $G$ has $2N$ vertices given by the set Vert$(G):=\{V_i^{\matholdcal{M}},V_i^{\matholdcal{E}}\}_{i=1}^N$, where $V_i^{\matholdcal{M}}$ and $V_i^{\matholdcal{E}}$ are respectively associated with the operations $\matholdcal{M}^{\A_i}$ and $\matholdcal{E}_i$, and both pairs of such vertices are associated with the time step $t_i$, for each $i$.
    \item $G$ contains a directed edge $V\rightarrow V’$ for $V,V’\in$ Vert$(G)$ whenever $t_V<t_{V’}$ and $\mathtt{S}_{\matholdcal{O}_V}\cap\mathtt{S}_{\matholdcal{O}_{V’}}\neq \emptyset$, where $\matholdcal{O}_{V}$ and $\matholdcal{O}_{V’}$ are the operations, and $t_V$ and $t_{V’}$ are the time steps associated with the vertices $V$ and $V’$ respectively.

\end{enumerate}

\end{definition}

\begin{definition}[Directed paths and partial order]
\label{def: dir_path}
The DAG $G$ associated with an EWFS in \Cref{definition: causal_str} defines a partial order relation $\prec$ on agents in the EWFS, with $\A_i\prec \A_j$ if and only if there is a directed path from the measurement vertex $V_i^{\matholdcal{M}}$ of $\A_i$ to the measurement vertex $V_j^{\matholdcal{M}}$ of $\A_j$. Furthermore, we will write $\A_i\prec^S \A_j$ whenever there is such a directed path and the system $S$ is included in at least one of the set intersections $\mathtt{S}_{\matholdcal{O}_V}\cap\mathtt{S}_{\matholdcal{O}_{V’}}$ involved in that path. We use $\A_i\not\prec \A_j$ and $\A_i\not\prec^{S} \A_j$ to denote the absence of directed paths with the above-defined properties.
\end{definition}

Crucially, notice that the above graph $G$ and induced partial order $\prec$ represent objective properties inherent to the EWFS, independent of the specific settings chosen to model the measurement $\matholdcal{M}^{\A_i}$ in the augmented EWFS. This is because the definition only relies on the time order of operations and the set of systems on which the operations act. These aspects are included in the description of the original EWFS, \Cref{def:LWFS}, and the settings do not affect them.

Given this definition, a natural causality principle is that an outcome $\aaa_j$ should not depend on a setting $\xxx_i$ of a measurement whenever $\A_i\not\prec \A_j$. As our circuits are acyclic and operations therein have a clear time ordering, we have $\A_i\prec \A_j$ implies $t_i<t_j$, and it is easy to see that this ensures that $G$ is indeed a directed acyclic graph and that $\prec$ is a partial order relation.
Moreover, this causality principle ensures that there is no retrocausal dependence of outcomes on future settings, despite settings in our formalism representing Heisenberg cuts rather than actual experimental measurement choices (further elaborated in \Cref{sec: setting_interpretation}).

Finally we note that $t_i<t_j$  does not necessarily imply $\A_i\not\prec\A_j$ as it is possible to have two measurements acting at different times but on non-overlapping sets of systems, then $G$ will not contain any directed paths between the measurements of these agents. Thinking of different subsystems as being embedded at different ``spatial locations'',  we can regard such disjoint sets of systems as being space-like separated. Therefore, if we consider the circuit as being embedded in a spacetime, with an output of one operation connected to the input of another only if the former is in the past light-cone of the latter, then $\prec$ is compatible with the causal structure of the spacetime \cite{VilasiniRennerPRA, VilasiniRennerPRL}.

\begin{restatable}[]{theorem}{MainTheorem}
\label{theorem: main}\phantom{.}
    \begin{enumerate}
        \item Completeness: In any given EWFS, all conventional predictions in that EWFS can be derived within the single augmented circuit of that EWFS. More explicitly, each conventional prediction $P_{conv}(\vec{a}_j=\vec{\aaa}_j|\vec{a}_l=\vec{\aaa}_l)$ in the EWFS equals a particular setting conditioned prediction $P(\vec{a}_j=\vec{\aaa}_j|\vec{a}_l=\vec{\aaa}_l,\vec{x}=\vec{\xxx}^*)$ of the augmented circuit where the setting choice $\vec{x}=\vec{\xxx}^*$ is such that $x_i=1$ for all $i\in\{j_1,...,j_p,l_1,...,l_q\}$ and $x_i=0$ for all $i\not\in\{j_1,...,j_p,l_1,...,l_q\}$.
        \item Consistency: For any EWFS, the set of all statements $\Sigma^{aug}$ obtained in the corresponding augmented circuit (\Cref{def: all_aug_statements}) are consistent according to \Cref{def: consistency_pred}. 
        \item Causality: For every setting-conditioned prediction $P(\vec{a}_j=\vec{\aaa}_j|\vec{a}_l=\vec{\aaa}_l,\vec{x}=\vec{\xxx})$, and every $i$ such that $\A_i\not\prec \A_k$ for all $k\in \{j_1,...,j_p,l_1,...,l_q\}$, the prediction is independent of the setting $x_i$. That is, for all such $i$, we have the following, noting that $\vec{x}=\vec{\xxx}$ is short for $x_1=\xxx_1,....,x_N=\xxx_N$ and each $\xxx_i\in \{0,1\}$.
\begin{align}
\label{eq: causality_indep}
 \begin{split}
 &P(\vec{a}_j=\vec{\aaa}_j|\vec{a}_l=\vec{\aaa}_l,(x_1=\xxx_1,...x_i=0,...,x_N=\xxx_N)) =\\ &P(\vec{a}_j=\vec{\aaa}_j|\vec{a}_l=\vec{\aaa}_l,(x_1=\xxx_1,...x_i=1,...,x_N=\xxx_N)).
 \end{split}  
\end{align}

    \end{enumerate}
\end{restatable}

Whenever a setting-conditioned prediction is independent of a setting $x_i$, i.e., satisfies \cref{eq: causality_indep}, we will simply drop $x_i$ from that prediction and denote it as follows:
\begin{align}
    \begin{split}
   &P(\vec{a}_j=\vec{\aaa}_j|\vec{a}_l=\vec{\aaa}_l,\vec{x}=\vec{\xxx})=\\   &P(\vec{a}_j=\vec{\aaa}_j|\vec{a}_l=\vec{\aaa}_l,\vec{x}\backslash x_i=\vec{\xxx}\backslash \xxx_i),
    \end{split}
\end{align}
where $\vec{x}\backslash x_i=(x_1,...,x_{i-1},x_{i+1},...,x_N)$ is the setting vector with the i$^{th}$ component omitted, and similarly for the vector of setting values $\vec{\xxx}\backslash \xxx_i$.

A proof of this theorem is deferred to \cref{ssec: proofs_appendix} as it will use results from \cref{appendix: prob_rule_general}.
\subsection{Properties of agents' reasoning}
\label{appendix: QUCDS}

In \cref{sec: consistency_reasoning}, we motivated the assumptions  \hyperref[def: Q_assump]{$\mathbf{Q}$}, \hyperref[def: U_assump]{$\mathbf{U}$}, \hyperref[def: assump_C]{$\mathbf{C}$}, \hyperref[def: assump_D]{$\mathbf{D}$} and \hyperref[def: assump_S]{$\mathbf{S}$} which formalise a version of FR's assumptions $\textup{Q}$, $\textup{U}$, $\textup{C}$, $\textup{D}$ and $\textup{S}$ while making explicit the role of the settings or Heisenberg cuts. Our main result concerning agents' reasoning in EWFS was to show that this formal version of assumptions regarding agents' use of unitary quantum theory to each others' labs and classical logic applied to observed outcomes, is in fact consistent in any EWFS. Here we provide the formal definitions of these assumptions which utilise methods from epistemic modal logic for defining $\textup{C}$, $\textup{D}$ and $\textup{S}$ (inspired by previous work \cite{Nurgalieva2018, Vilasini_2019}). 

One of the assumptions (or logical axioms) used in the FR argument is of the form ``If Alice is certain that Bob is certain that the outcome $a=1$'', then ``Alice is certain that $a=1$''. This is called the $\textup{C}$ assumption there, and is about the ability of Alice to inherit Bob's knowledge. In \cite{Nurgalieva2018}, this assumption $\textup{C}$ was formalised within the framework of modal logic. 
Here $K_{\A_i}(S)$ is used to denote that ``agent $\A_i$ knows that the statement $S$ is true'', where $K^{\A_i}$ is known as a knowledge operator (see \cite{Nurgalieva2018} for a formal definition of the knowledge operators in terms of Kripke structures in modal logic). The assumption $\textup{C}$ can then be succinctly expressed as an inference of the form 
\begin{equation}
    \label{eq: C_logic}
    K_{\A_i}K_{\A_j}(S) \Rightarrow K_{\A_i}(S),
\end{equation}

More specifically, it is pointed out in \cite{Nurgalieva2018} that such an inference only needs to be made when the agent $\A_i$ \emph{trusts} the agent $\A_j$ (for otherwise $\A_i$ may not believe in everything that $\A_j$ claims to know, and may not want to inherit $\A_j$'s knowledge). In Wigner's Friend scenarios such a that of FR, \cite{Nurgalieva2018} instantiate the trust structure by considering pairs of agents performing compatible measurements, and only apply $\textup{C}$ for such pairs. 
The trust structure will not be relevant for the general solution that we propose here  because, as we will show (c.f. \Cref{theorem: main}), the inclusion of the settings in our framework ensures the general validity of \cref{eq: C_logic} in an augmented EWFS independently of the trust structure.

\begin{definition}[Assumption $\mathbf{C}$]
\label{def: assump_C}
For any two agents $\A_i$ and $\A_j$, \cref{eq: C_logic} holds for all statements $S\in \Sigma^{aug}$.
\end{definition}

As noted in \cite{Nurgalieva2018}, the FR argument also uses the distributive axiom of logic. We instantiate this in our formalism below.

\begin{definition}[Assumption $\mathbf{D}$]
\label{def: assump_D}
For any set of agents reasoning using the augmented circuit, if $S_1\in \Sigma^{aug}_L$, and $S_1\Rightarrow S_2$, then $S_2\in \Sigma^{aug}_L$ holds and we have
\begin{equation}
    \label{eq: logic_axiom1}
    K_{\A_i}(S_1\land (S_1\Rightarrow S_2))\Rightarrow K_{\A_i}(S_2),
\end{equation}
that is if an agent $\A_i$ knows $S_1$ and also that $S_1$ implies $S_2$, then the agent knows $S_2$.
\end{definition}

Finally, we consider the $\textup{S}$ assumption of FR which states that an agent cannot be certain of two opposite values|say $a=0$ and $a=1$|of a measurement outcome. In our framework, this is a weaker version of our general consistency condition of \Cref{def: consistency_pred} as the latter applies to all predictions while the former only to logical predictions.

\begin{definition}[Assumption $\mathbf{S}$]
\label{def: assump_S}
For any subset $\vec{a}_j$ of measurement outcomes, if $P(\vec{a}_j=\vec{\aaa}_j)=1$, then it is impossible to have $P(\vec{a}_j=\vec{\aaa}'_j)=1$ for any outcome values $\vec{\aaa}_j\neq \vec{\aaa}'_j$.
\end{definition}

Now, for the quantum theory dependent assumptions which independently capture the validity of the Born rule and of unitary evolution of closed quantum systems (including agents' labs).

\begin{definition}[Assumption $\mathbf{Q}$]
\label{def: Q_assump}
Consider a statement $S:=$``If the outcomes $\vec{a}_l$ take values $\vec{\aaa}_l$ and the settings take the value $\vec{x}=\xxx$, then the outcomes $\vec{a}_l$ take values $\vec{\aaa}_l$ with a probability $P$.'' Agents in an EWFS can regard such a statement $S$ as true if and only if the corresponding setting-conditioned prediction $P(\vec{a}_j=\vec{\aaa}_j|\vec{a}_l=\vec{\aaa}_l,\vec{x}=\vec{\xxx})=P$ can be derived by applying the Born rule to the EWFS (as detailed in \Cref{definition: setting_prediction}). 
\end{definition}

We note that by construction, the set of all statements $\Sigma^{aug}$ constructed from setting-conditioned predictions in the augmented circuit has this property. In our formalisation of the \hyperref[def: U_assump]{$\mathbf{U}$} assumption, we make explicit another implicit assumption in the previous literature, which relates to quantum control over other agents' labs.

\begin{definition}[Assumption $\mathbf{U}$]
\label{def: U_assump}
 Agents can choose the setting $x_i=0$ (i.e., pure unitary description) for the measurement of any agent $\A_i$ whose outcome they are not logically reasoning about i.e., for every setting-conditioned prediction $P(\vec{a}_j=\vec{\aaa}_j|\vec{a}_l=\vec{\aaa}_l,\vec{x}=\vec{\xxx})$, both choices $x_i\in\{0,1\}$ are allowed for all $i\not\in \{j_1,...,j_p,l_1,...,l_q\}$. Moreover, agents can have full quantum control over the labs of other agents, i.e., the measurement of an agent $\A_i$ can act non-trivially on the total system $\mathtt{S}_j\M_j$ comprising the lab of another agent $\A_j$.
\end{definition}

\subsection{Non-super-agent structure}
\label{appendix:superagent}

The concept of the non-super-agent structure $n\matholdcal{SA}$ of an EWFS was explained and illustrated in \Cref{sec: standard_QT}, and it defines a criterion for when one agent does not act as a super-agent to another agent in an EWFS, and thereby distinguish standard quantum experiments from genuinely Wigner's Friend like experiments (\Cref{def:std_q_exp}). This criterion, defined below, is based on the causal structure and partial order/DAG of an EWFS defined in \Cref{definition: causal_str,def: dir_path} presented earlier in this appendix.  
\begin{definition}[Non-super-agent structure $n\matholdcal{SA}$]
\label{definition: nSA}

We say that an EWFS involving a set $\{A_1,...,A_N\}$ of agents respects a non-super-agent structure $n\matholdcal{SA}$ whose elements are pairs $(\A_i,\A_j)$ of agents, if the given EWFS is operationally equivalent to another EWFS involving a set $\{A'_1,...,A'_N\}$ of agents, such that for all  $(\A_i,\A_j)\in n\matholdcal{SA}$ the following conditions hold.

\begin{enumerate}
    \item If $i=j$, then $\matholdcal{E}'_i$ acts trivially on the memory $M'_i$.
    \item If $i\neq j$, then one of the following holds
    \begin{itemize}
        \item $\A_i'\not\prec \A_j'$
        \item If $\A_i'\prec \A_j'$, then $\matholdcal{E}'_i$ acts trivially on the memory $M_i'$ and $\A_i'\not\prec^{\M'_i}\A_j'$
    \end{itemize}
If $(\A_i,\A_j)\in n\matholdcal{SA}$, we say that $\A_j$ does not act on the memory of $\A_i$ or does not act as a super-agent to $\A_i$.
\end{enumerate} 
\end{definition}

\section{Overview of the Frauchiger-Renner apparent paradox}\label{appendix:FR_review}

\subsection{The FR no-go theorem}
\label{sec: FRreview}

Here we review the assumptions $\textup{Q}$, $\textup{C}$ and $\textup{S}$ of FR's claimed no-go theorem, as well as the additional assumptions $\textup{U}$ and $\textup{D}$ that are also relevant to the FR analysis, as pointed out in \cite{Nurgalieva2018}. The FR protocol and formal statement of their no-go theorem will be reviewed in \Cref{sec: FR_PM_review}.

Assumption ($\textup{Q}$): it asserts that an agent can be certain that a given proposition holds whenever the quantum-mechanical Born rule assigns probability 1 to it. Specifically:\\

\noindent\fbox{%
	\parbox{\linewidth}{%
		Suppose that agent A has established that
		\indent\emph{Statement} A$^\textup{(i)}$: ``System $\Ss$ is in state $\ket{\psi}_\Ss$ at time $t_0$.''
		Suppose furthermore that agent A knows that
		\indent\emph{Statement} A$^\textup{(ii)}$: ``The value $x$ is obtained by a measurement of $\Ss$ w.r.t. the family $\{\pi_x^{t_0}\}_{x\in\chi}$ of Heisenberg operators relative to time $t_0$, which is completed at time $t$.''
		If $\braket{\psi|\pi_\xi^{t_0}|\psi}=1$ for some $\xi\in\chi$, then agent A can conclude that
		\indent\emph{Statement} A$^\textup{(iii)}$: ``I am certain that $x=\xi$ at time $t$.''
	}%
}\\

Assumption ($\textup{C}$): 
It asserts that one agent can inherit the knowledge of another agent who uses the same theory as them to arrive at their conclusions.\\

\noindent\fbox{%
	\parbox{\linewidth}{%
		Suppose that agent A has established that
		\indent\emph{Statement} A$^\textup{(i)}$: ``I am certain that agent A', upon reasoning within the same theory as the one I am using, is certain that $x=\xi$ at time $t$''.
		Then agent A can conclude that
		\indent\emph{Statement} A$^\textup{(ii)}$: ``I am certain that $x =\xi$ at time $t$''.
	}%
}\\

Assumption ($\textup{S}$): from the viewpoint of an agent who carries out a particular measurement, this measurement has one single outcome. Specifically:\\

\noindent\fbox{%
	\parbox{\linewidth}{%
		Suppose that agent A has established that
		\indent\emph{Statement} A$^\textup{(i)}$: ``I am certain that $x=\xi$ at time $t$''
		The agent A must necessarily deny that
		\indent\emph{Statement} A$^\textup{(ii)}$: ``I am certain that $x\neq \xi$ at time $t$''.
	}%
}\\

Note that a violation of $\textup{S}$ can itself be interpreted as a logical paradox, as it would imply that the outcome $x$ is both $\xi$ and not $\xi$ with certainty.

In \cite{Nurgalieva2018}, the authors also identify additional assumption $\textup{U}$ (unitarity) that FR use (as part of $\textup{Q}$) but did not explicitly state in their set of assumptions.

Assumption ($\textup{U}$): \\

\noindent\fbox{%
	\parbox{\linewidth}{%
		An agent $\A$ can model measurements performed by 
		any other agent $\B$ as reversible unitary evolutions  in $\B$'s lab.
	}%
}\\

Furthermore they also note that the FR reasoning involves another basic rule of classical logical inference namely the \emph{distributive axiom}.

Assumption ($\textup{D}$): \\

\noindent\fbox{%
	\parbox{\linewidth}{%
		If an agent $\A$ knows a statement $s_1$ and also knows that $s_1$ implies another statement $s_2$ then agent $\A$ can conclude that they know $s_2$.
	}%
}\\

\noindent They formulate this in terms of a knowledge operator which formally keeps track of which agent knows which statement \cite{Nurgalieva2018}. We will also introduce said notation later. Assumption $\textup{D}$ is implicitly used in FR when the statements $s_1$ and $s_2$ correspond to measurement outcomes, e.g. $s_1$ could be ``I am certain that $x =\xi$ at time $t$', $s_2$ could be ``I am certain that $x'=\xi'$ at time $t'$''. Then if agent $\A$ who knows $s_1$ also knows ``If I am certain that $x =\xi$ at time $t$, then I am certain that $x'=\xi'$ at time $t'$'',  they would use $\textup{D}$ to conclude that agent $A$ is certain of the  statement $s_2$, ``I know that $x' =\xi'$ at time $t'$''.

In~\cite{Nurgalieva2018}, the FR argument is refined by making explicit the additional assumptions $\textup{U}$ and $\textup{D}$. Thus, they suggest the implication that the assumptions $\textup{Q}$, $\textup{U}$, $\textup{C}$, $\textup{D}$ and $\textup{S}$ cannot be simultaneously satisfied in the protocol proposed by FR. We note that this result is proven in \cite{Nurgalieva2018} by formalising these assumptions using the Kripke structure of epistemic modal logic (which is the branch of logic that refers to knowledge of agents) and we refer the reader to \cite{Nurgalieva2018} for the formal statement of this theorem in this mathematical language. We will not review the full modal logic framework here but will refer to aspects of it wherever necessary.

At a broad level, the proof proceeds by considering the FR scenario where agents reason about each other's knowledge using assumptions $\textup{Q}$, $\textup{U}$, $\textup{C}$ and $\textup{D}$ and claims to show that such agents would always arrive at a violation of $\textup{S}$, which as we have explained above can be interpreted as a paradox.

\subsection{Entanglement version of the FR scenario}
\label{sec: FR_ent_review}

In \Cref{sec: resolution_entanglement} and \cref{fig: FR_ent_circuit} of the main text, we provided an overview of of the entanglement version of the FR scenario and the main arguments. Here, we review in more detail the entanglement-based version of the FR scenario and apparent paradox that highlights the proof method typically employed to prove the no-go claim regarding a contradiction between the assumptions $\textup{Q}$, $\textup{U}$, $\textup{C}$ and $\textup{D}$ (reviewed in the subsection above).

The entanglement version of the FR protocol is originally attributed to Lluis Masanes (based on a talk), and was also mentioned by Matthew Pusey in \cite{Pusey2018}. It is much simpler than FR's original (prepare and measure based) protocol, but makes the important/salient features of the FR protocol more readily accessible. The resolution to EWFS paradoxes given by our work is however fully general and applies in particular to both the entanglement and to the original prepare and measure version of FR's arguments. For the interested reader, we provide a review of the original FR scenario in \Cref{sec: FR_PM_review} and apply our framework to resolve it, in \Cref{ssec: resolution_prep}. 

In this section we follow the notation and agent naming conventions of \cite{Nurgalieva2018, Vilasini_2019}. Here we have two agents Alice and Bob who measure individual subsystems of a bipartite system while two super-agents Ursula and Wigner measure the labs (system and memory) of Alice and Bob respectively. The protocol can be broken down into three steps: a bipartite state preparation (pre-selection) at an initial time $t=1$, intermediate local measurements by Alice and Bob on the state at time $t=2$, a final measurement and post-selection by two super-agents Ursula and Wigner at time $t=3$. 
\begin{itemize}
	\item {\bf Pre-selection at time $t=1$: }An initial state $\ket{\psi^{t=1}}_{\R\Ss}:=\frac{1}{\sqrt{3}}(\ket{00}_{\R\Ss}+\ket{10}_{\R\Ss}+\ket{11}_{\R\Ss})$ is prepared and shared between Alice and Bob where $\R$ and $\Ss$ label the subsystems belonging to Alice and Bob respectively. Alice and Bob's memories $\A$ and $\B$ are initialised to $\ket{0}_\A$ and $\ket{0}_\B$. Hence the initial preparation (i.e., pre-selected state) on $\R\A\Ss\B$ is given by
	\begin{equation}
		\ket{\psi^{t=1}}_{\R\A\Ss\B}:=\frac{1}{\sqrt{3}}(\ket{0000}+\ket{1000}+\ket{1010})_{\R\A\Ss\B}
	\end{equation}
	\item {\bf Intermediate operations at time $t=2$: }Alice and Bob measure their respective systems in the computational basis $\{\ket{0},\ket{1}\}$ and store the outcomes $a$ and $b$ of their respective measurements in their memory systems. From the outside perspective, Alice's measurement is modelled as a unitary $\matholdcal{M}^{\A}_{unitary}:=(CNOT)_{\R\A}$ on the joint system $\R\A$ which performs a CNOT operation with $\R$ as control and $\A$ as target, while Bob's measurement is similarly modelled as a unitary $\matholdcal{M}^{\B}_{unitary}:=(CNOT)_{\Ss\B}$.
	\item {\bf Post-selection at time $t=3$:} The super-agents Ursula and Wigner post-select on the following final state $\ket{\phi^{t=3}}_{\R\A\Ss\B}:=\ket{\textup{ok}}_{\R\A}\ket{\textup{ok}}_{\Ss\B}$, which they achieve by measuring $\R\A$ and $\Ss\B$ respectively in the basis $\{\ket{\textup{ok}}:=\frac{1}{\sqrt{2}}(\ket{00}-\ket{11}),\ket{\textup{fail}}:=\frac{1}{\sqrt{2}}(\ket{00}+\ket{11})\}$ and halting when both of them obtain the outcome $u=\textup{ok}$, $w=\textup{ok}$ (where $u$ and $w$ denote Ursula's and Wigner's outcomes, the halting condition is checked by announcing these measurement outcomes in each run).
\end{itemize}

When modelling both measurements as unitaries, the joint state of $\R\A\Ss\B$ just after time $t=2$ is given by
\begin{align}
	\label{eq: state_main}
	\begin{split}
		\ket{\psi^{t=2}}_{\R\A\Ss\B}=(\matholdcal{M}^{\A}_{unitary}\otimes \matholdcal{M}^{\B}_{unitary})  \ket{\psi^{t=1}}_{\R\A\Ss\B}=\frac{1}{\sqrt{3}}(\ket{0000}+\ket{1100}+\ket{1111})_{\R\A\Ss\B}.
	\end{split}
\end{align}

The super-agents can then calculate the probability of success of the post-selection given the pre-selection and intermediate unitary evolution as
\begin{align}
\label{eq:FRpostselProb}
	\begin{split}	P(u=w=\textup{ok}|\psi^{t=2})=|\braket{\phi^{t=3}|\psi^{t=2}}|^2=|\bra{\phi^{t=3}}(\matholdcal{M}^{\A}_{unitary}\otimes \matholdcal{M}^{\B}_{unitary})  \ket{\psi^{t=1}}|^2=\frac{1}{12}.
	\end{split}
\end{align}

By this, they establish that they have a non-zero probability of obtaining $u=w=\textup{ok}$ and can thus repeat the protocol until they succeed. Upon successfully obtaining the desired outcomes, the protocol is halted and the agents reason about each others' knowledge as follows, where we recall that $K_{\A}(S)$ denotes that Agent $\A$ knows the statement $S$.  

\begin{itemize}
	\item Upon obtaining $u=\textup{ok}$ on measuring $\R\A$, Ursula reasons using the joint state $\ket{\psi^{t=2}}$ (\cref{eq: state_main}) that Bob must have certainly obtained the outcome $b=1$ upon measuring $\Ss$, since  $\bra{\textup{ok}}_{\R\A}\bra{00}_{\Ss\B} \ket{\psi^{t=2}}_{\R\A\Ss\B}=0$. This gives 
	\begin{equation}
		\label{eq: chain1}
		K_U(  u=w=\textup{ok} \Rightarrow b=1)
	\end{equation}
	\item  Using the same state, Ursula knows that if Bob obtained $b=1$ on measuring $\Ss$, he would have concluded with certainty that Alice obtained $a=1$ on measuring $\R$ since $\bra{00}_{\R\A}\bra{11}_{\Ss\B}\ket{\psi^{t=2}}_{\R\A\Ss\B}=0$.
	\begin{equation}
		\label{eq: chain2}
		K_U K_B(  b=1 \Rightarrow a=1)
	\end{equation}
	\item Again using the same state, Ursula further reasons that Bob knows that if Alice had obtained $a=1$, she would have concluded with certainty that Wigner would obtain $w=\textup{fail}$, since  $\bra{11}_{\R\A}\bra{\textup{ok}}_{\Ss\B}\ket{\psi^{t=2}}_{\R\A\Ss\B}=0$. This gives 
	\begin{equation}
		\label{eq: chain3}
		K_U K_B K_A(  a=1 \Rightarrow w=\textup{fail}).
	\end{equation}
\end{itemize}

As shown in~\cite{Nurgalieva2018}, the above three statements can be combined using the assumption $\textup{C}$ of the FR paper (in the form of \cref{eq: C_logic}) and the distributive axiom $\textup{D}$ to yield the following paradoxical chain of statements. To obtain this result, the assumption $\textup{C}$ in the form of \cref{eq: C_logic} only needs to be used between the following pairs of agents: Alice and Bob, Alice and Wigner, Ursula and Bob, Ursula and Wigner, as other pairs of agents need not trust each other \cite{Nurgalieva2018, Vilasini_2019} \cite{Note15}.

\begin{equation}
	\label{eq: chain4}
	K_U( u=w=\textup{ok} \Rightarrow b=1 \Rightarrow  a=1 \Rightarrow w=\textup{fail}),
\end{equation}
or in short $K_U(u=w=\textup{ok} \Rightarrow w=\textup{fail})$. This  argument aims to establish that agents reasoning using $\textup{Q}$, $\textup{U}$, $\textup{C}$ and $\textup{D}$ will arrive at a contradiction with $\textup{S}$ as they conclude through such a reasoning that $w=\textup{ok}$ and $w=\textup{fail}$ must both hold with certainty, and is therefore regarded as a proof of the FR no-go theorem (and its refinement as given in \cite{Nurgalieva2018}) regarding the incompatibility of $\textup{Q}$, $\textup{U}$, $\textup{C}$, $\textup{D}$ and $\textup{S}$.

\subsection{Prepare and measure version of the FR scenario}
\label{sec: FR_PM_review}

The scenario in \cite{Frauchiger2018} describes a protocol realised by agents F, \tb{F}, W and \tb{W}. Agents F and \tb F have their own individual labs while W and \tb W are so-called super-agents, i.e. W can perform arbitrary measurements on F and the lab of F, while \tb W can perform arbitrary measurements on \tb F and the lab of \tb F. The labs of F and \tb F are completely isolated from W and \tb W until W and \tb W measure at the end of the protocol. It is also assumed that the labs F and \tb F are initially in pure states. The lab systems of F and \tb F are denoted L and \tb L respectively. \tb L includes everything in \tb F's lab such as the agent \tb F and a random generator R that they use, but excludes a spin qubit S which will start off in the lab of \tb F and move to the lab of F during the protocol. The lab system L of F will include the spin S that arrives to F, as well as the agent F and other devices in their lab which are not explicitly specified.

The protocol is repeated $n$ times. Between each implementation, it is reset to the initial state. There is a halting condition which is examined at the end of each round. When the condition is satisfied, the protocol is stopped and the last round of the experiment is analysed. The $n^\text{th}$ round of the protocol from the perspective of W and \tb W is as follows:

Before time $n\!:\!00$, \tb F tosses a coin in her lab which gives heads with probability 1/3 and tails with probability 2/3. This coin toss is a random variable which is obtained by measuring the following quantum state in the same basis in which it is expressed.
\begin{align}
\label{eq: FR_PM_instate}
\ket{\textup{init}}_\R:=\sqrt{\frac13}\ket{\textup{heads}}_\R+\sqrt{\frac23}\ket{\textup{tails}}_\R.
\end{align}

The lab of \tb F consists of the random generator (or coin) R, a spin S and $\Lbar\backslash$ which represents the rest of the lab. In the following, $\Lbar:=\R\otimes \Lbar\backslash$ is the lab of \tb F, excluding the spin S.
Then the unitary $U_{\mb{F}}$ (from the perspective of F, W and \tb W) that describes F's coin toss (measurement of R) implements the evolution

\begin{align}
U_{\mb{F}}:\ket{\textup{init},\phi_0,S_0}_{\Lbar\Ss}\to \big(\sqrt{\frac{1}{3}}\ket{\mb h}_{\Lbar}+\sqrt{\frac{2}{3}}\ket{\mb t}_{\Lbar}\big)\ket{S_0}_{\Ss},\label{eq:FR_PM_UF}
\end{align}
where $\ket{\textup{init},\phi_0,S_0}_{\Lbar\Ss}:=\ket{\textup{init}}_\R\ket{\phi_0}_{\Lbar\backslash}\ket{S_0}_\Ss$ with $\ket{\phi_0}_{\Lbar\backslash}$ the initial state of \tb F's lab excluding the spin system and the random generator; the latter two being $\ket{S_0}_\Ss$ and $\ket{\textup{init}}_\R$ respectively, where $\ket{S_0}_\Ss$ is some initial state of $S$. The kets $\ket{\mb h}$, $\ket{\mb t}$ represent the state of \tb F's lab (excluding the spin qubit) after her measurement of the random variable ($\ket{\mb h}$ in the case her measurement revealed heads while $\ket{\mb t}$ when she obtained tails).

\begin{figure*}
    \includegraphics[scale=0.73]{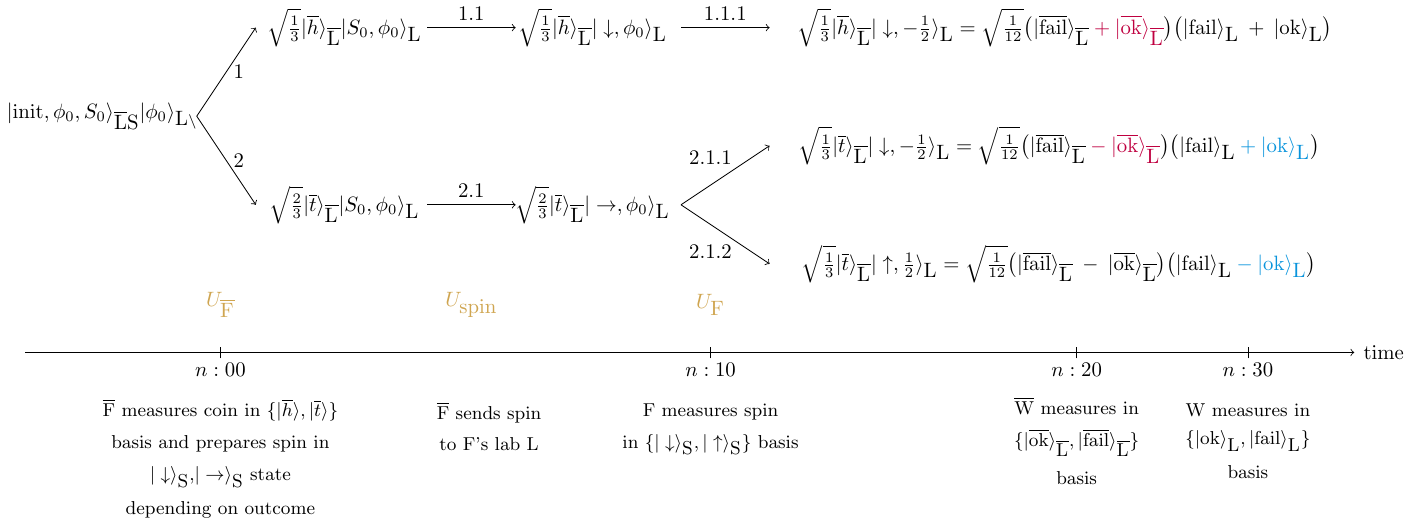}
    \caption{The FR scenario expressed using branching notation. The quantum state as per a unitary description from the viewpoint of W and \tb  W at any time, is the sum of the states of each branched at said time. Arrows indicate the splitting of the amplitudes due to the actions of individual observers. We refer to these as \emph{branches} and each branch has its own number above the corresponding arrow. The two purple kets on branches 1.1.1. and 2.1.1 cancel each other out when said branches are added together, while the two blue kets on branches 2.1.1. and 2.1.2. cancel when said branches are added together.}\label{fig:FR paper experiment}
\end{figure*}
Following the coin toss, \tb F prepares a spin qubit in her lab in state $\ket{\downarrow}_\Ss$ if she gets heads, while $\ket{\rightarrow}_\Ss$ if she gets tails. From the perspective of F, W and \tb W, this is a unitary process since her lab is an isolated system.
The unitary that describes this preparation of $\Ss$ is denoted as $U_{\textup{spin}}$ and together with $U_{\mb{F}}$, this implements the evolution

\begin{align}
U_{\textup{spin}}U_{\mb{F}}:\ket{\textup{init},\phi_0,S_0}_{\Lbar\Ss}\to\sqrt{\frac13}\ket{\mb h,\downarrow}_{\Lbar\Ss}+\sqrt{\frac23}\ket{\mb t,\rightarrow}_{\Lbar\Ss},\label{eq:initial to final ket}
\end{align}

Meanwhile, F waits patiently in her lab, which we denote $\ket{\phi_0}_{\Lab\backslash}$ initially, where $\Lab\backslash$ denotes F's lab without the spin qubit $S$ and we will use $\Lab:=\Ss\otimes \Lab\backslash$. Between times $n\!:\!10$ and $n\!:\!20$, \tb F sends her qubit system to F's lab and F subsequently measures the qubit in the $\ket{\uparrow}_\Ss$, $\ket{\downarrow}_\Ss$ basis, denoting her outcome $+1/2$ and $-1/2$ respectively.
From the perspective of W and \tb W, the labs of F and \tb F are subsequently modelled through the series of evolutions $U_{\mb{F}}$, followed by $U_{\textup{spin}}$ followed by $U_{\textup{F}}$ where $U_{\textup{F}}$ is the unitary evolution corresponding to F's measurement of S.

\begin{align}\label{eq:final state FR version}
\begin{split}
&U_{\textup{F}}U_{\textup{spin}}U_{\mb{F}}: \ket{\textup{init},\phi_0,S_0}_{\Lbar\Ss}\ket{\phi_0}_{\Lab\backslash}
\to\sqrt{\frac13}\Big(\ket{\mb h}_\Lbar\ket{\downarrow,-1/2}_{\Lab}
+\ket{\mb t}_\Lbar\ket{\downarrow,-1/2}_{\Lab}+ \ket{\mb t}_\Lbar\ket{\uparrow,1/2}_{\Lab}\Big),
\end{split}
\end{align}
where $\ket{\downarrow,-1/2}_{\Lab}$ denotes the spin in the down state $\ket{\downarrow}_\Ss$ and the state of the rest of lab L when F obtains measurement outcome $-1/2$; similarly for $\ket{\downarrow,-1/2}_{\Lab}$. 
Between times $n\!:\!20$ and $n\!:\!30$, \tb W measures lab \tb L in the basis $\{ \ket{\mb{\textup{ok}}}_\Lbar, \ket{\mb{\textup{fail}}}_\Lbar\}$, where $\ket{\mb{\textup{ok}}}_\Lbar:=\sqrt{\frac12} (\ket{\mb h}_\Lbar-\ket{\mb t}_\Lbar)$, $\ket{\mb{\textup{fail}}}_\Lbar:=\sqrt{\frac12} (\ket{\mb h}_\Lbar+\ket{\mb t}_\Lbar)$. After time $n=30$,  W measures lab L in the basis $\{ \ket{{\textup{ok}}}_\Lab, \ket{{\textup{fail}}}_\Lab\}$, where $\ket{{\textup{ok}}}_{\Lab}:=\sqrt{\frac12} (\ket{\downarrow,-1/2}_\Lab-\ket{\uparrow,1/2}_\Lab)$, $\ket{{\textup{fail}}}_\Lab:=\sqrt{\frac12} (\ket{\downarrow,-1/2}_\Lab+\ket{\uparrow,1/2}_\Lab$. The final state of the labs (right hand side of \cref{eq:final state FR version}) in this measurement basis takes on the form
\begin{align}
\begin{split}
\sqrt{\frac{1}{12}} \left(\ket{\mb{\textup{fail}}}_\Lbar+\ket{\mb{\textup{ok}}}_\Lbar\right)\left(\ket{\textup{fail}}_\Lab+\ket{\textup{ok}}_\Lab\right)+\sqrt{\frac{1}{6}} \left(\ket{\mb{\textup{fail}}}_\Lbar-\ket{\mb{\textup{ok}}}_\Lbar\right)\ket{\textup{fail}}_\Lab.
\end{split}
\end{align}
In \cref{fig:FR paper experiment}, we summarise this scenario using branching notation.
If the outcomes of W and \tb W are not ok, \tb {ok} respectively, then the protocol is re-set and repeated. If their outcomes are ok, \tb {ok}, then the experiment is finished and F, \tb F, W and \tb W reason about their ok, \tb {ok} measurement outcomes and what they should be able to predict about the measurement outcomes of the other agents in the last round of the protocol. The obtention of ok, \tb {ok} by W and \tb W respectively is the halting condition mentioned previously, and from the above form of the final state, it is clear that this outcome occurs with probability $\frac{1}{12}$.

While reasoning, they assume that all agents are aware of the entire experimental procedure as described above, and that they all employ the same theory. As we will see, all agents, when reasoning from \emph{their} perspective, assume unitary dynamics for all agent's measurements other than those they are reasoning about. In other words, their protocol is a special case of an EWFS. In \cref{fig: FR_PM_aug_circuit1} we present the augmented circuit of the original FR protocol as a simplified version that makes the mapping back to the entanglement version more apparent.

 In FR's reasoning, they make the three assumptions $\textup{Q}$, $\textup{C}$, $\textup{S}$ explicitly, in addition to the assumptions $\textup{U}$ and $\textup{D}$ implicitly, as described in~\Cref{sec: FRreview}.

In particular, the authors then claim to prove the following theorem:

\emph{Theorem 1: Any theory that satisfies assumptions ($\textup{Q}$), ($\textup{C}$), and ($\textup{S}$) yields contradictory statements when applied to their thought experiment of Box 1}.
 By ``Box 1'' the authors are referring to the above protocol.

\section{Derivation of predictions in the augmented EWFS}
\label{appendix: prob_rule_general}
This section serves the purpose of deriving explicit expressions for predictions and setting-conditioned predictions in the augmented circuit of an EWFS from the Born rule and probability theory.

Our sample space $\Omega$ is chosen to be the set of all measurement outcomes under all settings, namely
\begin{align}
    \Omega=\{a_1, x_1, a_2, x_2, \ldots, a_N, x_N\}_{\{a_j,x_j\}_j},
\end{align}
where $a_j\in \perp$ for $x_j=0$ and $a_j\in  \mathtt{O}_j$ for $x_j=1$. The set of events is the power set of $\Omega$. It follows from applying the Born rule to our augmented circuit (see also \cref{fig: genform_circuit_settings}) that if agents $\A_1$, $\A_2$, \ldots, $\A_N$ were to perform measurements under settings $\vec x$ with outcomes corresponding to projectors $\pi_{a_1,x_1}^{\A_1}, \pi_{a_2,x_2}^{\A_2}, \ldots, \pi_{a_N,x_N}^{N}$ respectively, then the probability of these elementary events is

    \begin{align}
    \begin{split}
     &P(a_{1}, a_{2},\ldots, a_{N}, \vec x)=P(a_{1}, a_{2},\ldots, a_{N}|\vec x) P(\vec x)\\
     =&\tr\Bigg[  \matholdcal{E}^N\Bigg(\pi^{\A_N}_{a_{N},x_{N}} \matholdcal{M}_\text{unitary}^{\A_N}\Bigg(\ldots\matholdcal{E}^2\bigg(\pi^{\A_2}_{a_{2},x_{2}}\matholdcal{M}_\text{unitary}^{\A_2}\Big(\matholdcal{E}^1\Big(\pi^{\A_1}_{a_{1},x_{1}}\matholdcal{M}_\text{unitary}^{\A_1}(\rho_0)  \pi^{\A_1}_{a_{1},x_{1}}\Big)\Big)\pi^{\A_2}_{a_{2},x_{2}}\bigg)\ldots \Bigg)\pi^{\A_N}_{a_{N},x_{N}}\Bigg)\Bigg]\, P(\vec x), 
     \end{split}
     \label{eq:prob of elemetary events}
    \end{align}

where $\rho_0$ is the initial ``pre-selected'' state, namely $\rho_{\Ss_1,\ldots,\Ss_m}\otimes\proj{0}^{\otimes N}$ and $P(\vec x)$ the unconditional probability distribution over settings. The other term in \cref{eq:prob of elemetary events} is the probability of outcomes $\{a_1, a_2, \ldots, a_N\}$ conditioned on setting $\vec x$.

The probability of the other events can be derived by marginalising over this distribution. In particular, we will be interested in the probability of events corresponding to agents $\{\A_{j_1}, \A_{j_2},\ldots, \A_{j_p}\}$ obtaining outcomes $\{a_{j_1}, a_{j_2},\ldots, a_{j_p}\}$, conditioned of the setting being $\vec x$. This is given by

\begin{align}
	\begin{split}
		&P(a_{j_1}, a_{j_2},\ldots, a_{j_p}| \, \vec x) \\&=\sum_{\substack{a_1 \text{ if } x_1=1\, \&\, 1\,\notin \,\mathtt{OUT}\\ a_2 \text{ if } x_2=1\, \&\, 2\,\notin \,\mathtt{OUT} \\ \vdots \\ a_{\!N} \!\text{ if } x_N\!=1\, \&\, N\notin \,\mathtt{OUT}}} \tr\Bigg[  \matholdcal{E}^N\Bigg(\pi^{\A_N}_{a_{N},x_{N}}\!\ldots
		\matholdcal{E}^2\bigg(\pi^{\A_2}_{a_{2},x_{2}}\matholdcal{M}_\text{unitary}^{\A_2}\Big(\matholdcal{E}^1\Big(\pi^{\A_1}_{a_{1},x_{1}}\matholdcal{M}_\text{unitary}^{\A_1}(\rho_0)\pi^{\A_1}_{a_{1},x_{1}}\Big)\Big)\pi^{\A_2}_{a_{2},x_{2}}\bigg)\ldots \pi^{\A_N}_{a_{N},x_{N}}\Bigg)\Bigg],
	\end{split}\label{eq:prob of events given x}
\end{align}
$\mathtt{OUT}:=\{j_1,j_2,\ldots,j_p\}$. We have not summed over settings $x_j=0$ since the corresponding observer is modelled unitarily with a deterministic outcome, $a_i=\perp$. Note that if we have a set of projectors $\{\pi_j\}_j$ on a system $\A$ and a quantum channel $\matholdcal{E}$ also on $\A$, from the Stinespring dilation theorem and other elementary properties, it follows that there exists another system $\B$ such that for an arbitrary linear operator $\hat A$ on $\A$, we have
\begin{align}\label{eq:simplifiction}
\begin{split}
		&\sum_j \tr_\A\big[\matholdcal{E}( \pi_j \hat A \pi_j)\big]=\sum_j \tr_\A\big[ \tr_\B[ U_{\A\B} (\pi_j \hat A\pi_j)\otimes \rho_\B  U_{\A\B}^\dag] \big]\\
	&=\sum_j\tr_{\A\B}\big[ U_{\A\B} (\pi_j \hat A\otimes \rho_\B \pi_j)  U_{\A\B}^\dag \big]\\
	&=\tr_{\A\B}\big[(\sum_j \pi_j^2) \hat A\otimes \rho_\B\big]= \tr_\A[\hat A].
\end{split}
\end{align}
Applying this equality iteratively to \cref{eq:prob of events given x} allows us to simplify it.

\begin{align}
	\begin{split}
		&P(a_{j_1}, a_{j_2},\ldots, a_{j_p}| \, \vec x) =\sum_{\substack{a_1 \text{ if } x_1=1\, \&\, 1\,\notin \,\mathtt{OUT}\\ a_2 \text{ if } x_2=1\, \&\, 2\,\notin \,\mathtt{OUT} \\ \vdots \\ a_{K\!-\!1} \!\text{ if } x_{K-\!1}=1\, \&\, K-\!1\notin \,\mathtt{OUT}}} \tr\bigg[**\bigg],\\
  & ** =\pi^{\A_{K}}_{a_{\!K},1}  \matholdcal{E}^{K-\!1}\Bigg(\pi^{\A_{K-\!1}}_{a_{\!K-\!1},x_{\!K-\!1}}\!\!\ldots\matholdcal{E}^2\bigg(\pi^{\A_2}_{a_{2},x_{2}}\matholdcal{M}_\text{unitary}^{\A_2}\bigg(\matholdcal{E}^1\Big(\pi^{\A_1}_{a_{1},x_{1}}\matholdcal{M}_\text{unitary}^{\A_1}(\rho_0)\pi^{\A_1}_{a_{1},x_{1}}\Big)\bigg)\pi^{\A_2}_{a_{2},x_{2}}\bigg)\ldots \pi^{\A_{K-\!1}}_{a_{\!K-\!1},x_{\!K-\!1}}\Bigg).
	\end{split}\label{eq:prob of events given x 2}
\end{align}

where $K:=\max (\mathtt{OUT})$.
We see that this prediction does not depend on any channel $\matholdcal{E}^j$, setting $x_j$, nor any other property of an agent $\A_j$ for which $j>K$.

Using the definition of conditional probability, we can now derive an expression for setting-conditioned predictions (\Cref{definition: setting_prediction}), i.e. for the probability of a set of observers $\{\A_{j_1}, \A_{j_2},\ldots, \A_{j_p}\}$ obtaining outcomes $\{a_{j_1}, a_{j_2},\ldots, a_{j_p}\}$ given $\{\A_{l_1}, \A_{l_2},\ldots, \A_{l_q}\}$ measurement outcomes $\{a_{l_1}, a_{l_2},\ldots, a_{l_q}\}$ and setting $\vec x$. When using the definition of conditional probability, \cref{eq:prob of events given x} and simplifying by means of \cref{eq:prob of events given x 2}, we obtain

		\begin{align}
			\begin{split}
		&P\left(a_{j_1},a_{j_2},\ldots, a_{j_p}\,\big{|}\, a_{l_1}, \ldots  a_{l_q},\vec x\right):=\sum_{\substack{a_1\, \text{ if } x_1=1\, \&\,\, 1\,\notin \,\mathtt{OUT}\,\cup\,\mathtt{IN} \\ a_2\, \text{ if } x_2=1\, \&\,\, 2\,\notin \,\mathtt{OUT}\,\cup\,\mathtt{IN} \\ \vdots \\ a_{\!Q-\!1} \!\text{ if } x_{Q-\!1}=1\, \&\, {Q-\!1}\notin \,\mathtt{OUT}\,\cup\,\mathtt{IN}}} \frac{\textup{Numerator}}{\sum_{\substack{a_1\, \text{ if } x_1=1 \\ a_2\, \text{ if } x_2=1 \\ \vdots \\ a_{L-1} \!\text{ if } x_{\!L-1}\!=1}} \textup{Denominator}} ,\label{eq:general conditional prob rule simplicied 2}
	\end{split}
\end{align}
\begin{align}
&\textup{Numerator}:= \\
&\tr\left[  \pi^{\A_Q}_{a_{Q},1}\matholdcal{M}_\text{unitary}^{\A_Q} \Bigg(\matholdcal{E}^{Q-1} \Bigg( \pi^{\A_{Q-1}}_{a_{Q-1},x_{Q-1}} \!\!\ldots\matholdcal{E}^2\bigg(\pi^{\A_2}_{a_{2},x_{2}}\matholdcal{M}_\text{unitary}^{\A_2}\Big(\matholdcal{E}^1\Big(\pi^{\A_1}_{a_{1},x_{1}}\matholdcal{M}_\text{unitary}^{\A_1}(\rho_0)\pi^{\A_1}_{a_{1},x_{1}}\Big)\Big)\pi^{\A_2}_{a_{2},x_{2}}\bigg)\ldots \pi^{\A_{Q-1}}_{a_{Q-1},x_{Q-1}} \Bigg)\Bigg)\right]\\
&\textup{Denominator}:= \\
&\tr\left[  \pi^{\A_L}_{a_{L},1}\matholdcal{M}_\text{unitary}^{\A_L} \Bigg(\matholdcal{E}^{L-1} \Bigg( \pi^{\A_{L-1}}_{a_{L-1},x_{L-1}} \!\!\ldots\matholdcal{E}^2\bigg(\pi^{\A_2}_{a_{2},x_{2}}\matholdcal{M}_\text{unitary}^{\A_2}\Big(\matholdcal{E}^1\Big(\pi^{\A_1}_{a_{1},x_{1}}\matholdcal{M}_\text{unitary}^{\A_1}(\rho_0)\pi^{\A_1}_{a_{1},x_{1}}\Big)\Big)\pi^{\A_2}_{a_{2},x_{2}}\bigg)\ldots \pi^{\A_{L-1}}_{a_{L-1},x_{L-1}} \Bigg)\Bigg)\right]
\end{align}

where $Q:=\max\,( \mathtt{OUT}\cup\mathtt{IN})$, $\mathtt{IN}=\{l_1,l_2,\ldots, l_q\}$, $L:=\max\,(\mathtt{IN})$.
Furthermore, we have the constraint that  $x_{j_1}= x_{j_2}= \ldots= x_{j_p}= x_{l_1}= x_{l_2}=\ldots=x_{l_q} =1$, since we are reasoning about these measurement outcomes. Note that if we now sum $P\left(a_{j_1},a_{j_2},\ldots, a_{j_p}\,\big{|}\, a_{l_1}, \ldots  a_{l_q},\vec x\right)$ over the elements in the set $\{a_{j_1}, a_{j_2}, \ldots, a_{j_p}\}$ we obtain one, and thus the distribution is normalised. Notice also that we are not summing over outcomes $a_j$ for which $x_j=0$ since these correspond to the case said measurement is modelled unitarily. Also, as mentioned previously, were we have also pre-selected state $\rho_0$. The ``post-selected state'' is merely the post measurement state of the last measurement performed by observers $\A_j$, $j\in\mathtt{OUT}$. We can easily derive an expression for our predictions (\Cref{definition: prediction}) using \cref{eq:general conditional prob rule simplicied 2} and our prior $P(\vec x)$:
\begin{align}
\begin{split}
        P\left(a_{j_1},a_{j_2},\ldots, a_{j_p}\,\big{|}\, a_{l_1}, \ldots  a_{l_q}\right):=\sum_{\vec x} P\left(a_{j_1},a_{j_2},\ldots, a_{j_p}\,\big{|}\, a_{l_1}, \ldots  a_{l_q},\vec x\right) P(\vec x),\label{eq:prediction eq in appendix A}
\end{split}
\end{align}
where for consistency we have defined $P\left(a_{j_1},a_{j_2},\ldots, a_{j_p}\,\big{|}\, a_{l_1}, \ldots , a_{l_q},\vec x\right)=0$ if there exists $k\in\{ j_1, ..., j_p, l_1,..., l_q\}$ s.t. $a_k\neq \perp$, $\&$ $x_k=0$.

The simplification coming from the iterative application of \cref{eq:prob of events given x} to our conditional probability has important physical consequences, namely that the expressions \cref{eq:general conditional prob rule simplicied 2,eq:prediction eq in appendix A} readily do not depend on the channels, measurement schemes nor settings of agents in the future of when the agents in $\{ \A_k|\, k\in\mathtt{OUT}\cup\mathtt{IN} \, \}$ perform their operations, i.e. in the future of the agents who measurement outcomes are being reasoned about.

\section{Reduction of the augmented circuit in standard quantum scenarios}
\label{appendix: standard_circuits}

In \Cref{sec: standard_QT} we defined the subclass of EWFSs which we call standard quantum scenarios (\Cref{def:std_q_exp}), and showed that objective, setting-independent predictions emerge in this case. Here we show that in EWFSs corresponding to standard quantum scenarios, the augmented circuit of our framework reduces to a standard form quantum circuit without the ``Heisenberg-cut'' settings. We define two forms of standard quantum circuit representations below before proving this.

\begin{figure*}[t!]
    \centering
    \begin{tikzpicture}
    \draw[->] (6,-3.5)--(9,-3.5); \node at (9.5,-4) {time};
    
\draw [ultra thick,decorate,
    decoration = {calligraphic brace}] (-0.2,-2.2)--(-0.2,0.2); \node at (-1,-1) {$\rho_{\Ss_1,\ldots,\Ss_m}$};

    \draw (0,0)--(1,0); \node at (0.5,0.2) {$\Ss_1$}; \draw (0,-2)--(1,-2); \node at (0.5,-1.8) {$\Ss_m$}; \node at (0.5,-1) {$\cdot$}; \node at (0.5,-1.3) {$\cdot$}; \node at (0.5,-0.7) {$\cdot$}; \draw (1,1) rectangle node{$\matholdcal{M}^{\A_1}$} (2,-2.3); \draw[blue] (2,0.7)--(3,0.7); \node[blue] at (3.6,0.7) {$a_1\in\mathtt{O}_1$};

     \draw (2,0)--(3,0); \node at (2.5,0.2) {$\Ss_1$}; \draw (2,-2)--(3,-2); \node at (2.5,-1.8) {$\Ss_m$}; \node at (2.5,-1) {$\cdot$}; \node at (2.5,-1.3) {$\cdot$}; \node at (2.5,-0.7) {$\cdot$};
     
      \draw (3,0.3) rectangle node{$\matholdcal{E}^1$} (4,-2.3);

\begin{scope}[shift={(2,0)}]
           \draw (2,0)--(3,0); \node at (2.5,0.2) {$\Ss_1$}; \draw (2,-2)--(3,-2); \node at (2.5,-1.8) {$\Ss_m$}; \node at (2.5,-1) {$\cdot$}; \node at (2.5,-1.3) {$\cdot$}; \node at (2.5,-0.7) {$\cdot$};
\end{scope}


\node at (6,-1) {$\dots$};

\begin{scope}[shift={(5,0)}]
           \draw (2,0)--(3,0); \node at (2.5,0.2) {$\Ss_1$}; \draw (2,-2)--(3,-2); \node at (2.5,-1.8) {$\Ss_m$}; \node at (2.5,-1) {$\cdot$}; \node at (2.5,-1.3) {$\cdot$}; \node at (2.5,-0.7) {$\cdot$};
\end{scope}

 \draw (8,1) rectangle node{$\matholdcal{M}^{\A_N}$} (9,-2.3); \draw[blue] (9,0.7)--(10,0.7); \node[blue] at (10.6,0.7) {$a_N\in\mathtt{O}_N$};

 \begin{scope}[shift={(7,0)}]
           \draw (2,0)--(3,0); \node at (2.5,0.2) {$\Ss_1$}; \draw (2,-2)--(3,-2); \node at (2.5,-1.8) {$\Ss_m$}; \node at (2.5,-1) {$\cdot$}; \node at (2.5,-1.3) {$\cdot$}; \node at (2.5,-0.7) {$\cdot$};
\end{scope}
     \draw (10,0.3) rectangle node{$\matholdcal{E}^N$} (11,-2.3);
    \end{tikzpicture}
    \caption{A $\matholdcal{C}^{sys}$-form standard quantum circuit. Here $\mathtt{O}_i$ is the set of possible non-trivial values of the outcome $a_i$. Each $\text{$\matholdcal{M}$}^{\A_i}:=\{\pi^{\mathtt{S}_i}_{a_i}\}_{a_i}$ implements a projective measurement of the subset $\mathtt{S}_i\subseteq \mathtt{S}:=\{\Ss_1,...,\Ss_m\}$ of systems. }
    \label{fig:circuit_sys_form}
\end{figure*}

\begin{definition}[Standard quantum circuit representations]
\label{definition: std_circuit_forms}
 Consider a quantum protocol involving $N$ agents $\mathtt{A}:=\{\A_1,...,\A_N\}$ and $m$ systems $\mathtt{S}:=\{\Ss_1,..,\Ss_m\}$ where each agent $\A_i$ performs a projective measurement $\text{$\matholdcal{M}$}^{\A_i}:=\{\pi^{\mathtt{S}_i}_{\aaa_i}\}_{\aaa_i\in\mathtt{O}_i}$ at time $t_i$ that acts non-trivially on a subset $\mathtt{S}_i\subseteq \mathtt{S}$ of the systems, obtaining an outcome $a_i$ that can take values in a set $\mathtt{O}_i$, followed by a channel $\matholdcal{E}_i$ that may act on all $\mathtt{S}$. A standard quantum circuit representation of such a protocol corresponds to a circuit of one of the two following types, in one case all measurements are modelled as projectors on the systems $\mathtt{S}$ alone and in another case, all measurements can be equivalently purified to unitaries on the systems and some ancillas. 
 \begin{enumerate}
     \item {\bf $\matholdcal{C}^{sys}$-form quantum circuit (\cref{fig:circuit_sys_form})} A quantum circuit acting on $\mathtt{S}$ which is defined through the composition $\matholdcal{E}_N\circ\text{$\matholdcal{M}$}^{\A_N}\circ...\circ\matholdcal{E}_1\circ \text{$\matholdcal{M}$}^{\A_1}$, where each operation is defined over all of $\mathtt{S}$ but it is given that each $\text{$\matholdcal{M}$}^{\A_i}$ acts non-trivially on some subset $\mathtt{S}_i\subseteq \mathtt{S}$. Outcome probabilities are calculated by applying the Born rule to the circuit, wit the projective measurements $\text{$\matholdcal{M}$}^{\A_i}:=\{\pi^{\mathtt{S}_i}_{\aaa_i}\}_{\aaa_i\in\mathtt{O}_i}$ on $\mathtt{S}_i$.
     
      \item {\bf $\matholdcal{C}^{sys+anc}$-form quantum circuit (\cref{fig:circuit_sys+anc_form})} A quantum circuit acting on $\mathtt{S}\cup \{\M_1,...,\M_N\}$ where $\M_i$ denotes an ancillary quantum system corresponding to the measurement $\text{$\matholdcal{M}$}^{\A_i}$ whose state space is isomorphic to that of the systems $\mathtt{S}_i\subseteq \mathtt{S}$ on which the measurement acts non-trivially. It is defined through the composition $\matholdcal{E}_N\circ\text{$\matholdcal{M}$}^{\A_N}_{unitary}\circ...\circ\matholdcal{E}_1\circ \text{$\matholdcal{M}$}^{\A_1}_{unitary}$ where each measurement $\text{$\matholdcal{M}$}^{\A_i}$ is purified to a unitary interaction $\text{$\matholdcal{M}$}^{\A_i}_{unitary}$ acting on $\mathtt{S}\cup \M_i$ (and non-trivially on $\mathtt{S}_i\cup M_i$), where $\text{$\matholdcal{M}$}^{\A_i}_{unitary}$ corresponds to the unitary that implements a coherent copy from $\mathtt{S}_i$ to $\M_i$ in the orthonormal basis given by the measurement projectors (as defined in \cref{eq:M_unitary_gen}). Each $\matholdcal{E}_i$ acts only on $\mathtt{S}$. Outcome probabilities are calculated by measuring the ancillas $\M_i$ (using isomorphic projectors $\{\pi_{\aaa_i}^{\M_i}:=\ket{\aaa_i}\bra{\aaa_i}_{\M_i}\}_{\aaa_i\in \mathtt{O}_i}$) at a time $t_f>t_N$ at the end of the protocol (using the Born rule).
 \end{enumerate}
\end{definition}

\begin{figure*}[t!]
    \centering
    \begin{tikzpicture}
    \draw[->] (6,-6)--(9,-6); \node at (9.5,-6.5) {time};
    
\draw [ultra thick,decorate,
    decoration = {calligraphic brace}] (-0.2,-2.2)--(-0.2,0.2); \node at (-1,-1) {$\rho_{\Ss_1,\ldots,\Ss_m}$};

    \draw (0,0)--(1,0); \node at (0.5,0.2) {$\Ss_1$}; \draw (0,-2)--(1,-2); \node at (0.5,-1.8) {$\Ss_m$}; \node at (0.5,-1) {$\cdot$}; \node at (0.5,-1.3) {$\cdot$}; \node at (0.5,-0.7) {$\cdot$}; \draw (1,3.3) rectangle node{$\matholdcal{M}^{\A_1}_{unitary}$} (3,-3.3); \draw (3,3)--(13.5,3); \node at (13,3.2) {$\M_1$}; \node at (13,2.3) {$\cdot$}; \node at (13,2) {$\cdot$}; \node at (13,1.7) {$\cdot$};
    \draw (0,-3)--(1,-3); \node at (0.5,-2.8) {$\M_1$}; \node at (-0.5,-3) {$\ket{0}_{\M_1}$};

\draw (13.5,2.5) rectangle node{$\matholdcal{M}^{\M_1}$} (14.5,3.5); \draw[blue] (14.5,3)--(14.9,3); \node[blue] at (15.5,3) {$a_1\in\mathtt{O}_1$};

\draw (13.5,0.5) rectangle node{$\matholdcal{M}^{\M_N}$} (14.5,1.5);  \draw[blue] (14.5,1)--(14.9,1); \node[blue] at (15.5,1) {$a_N\in\mathtt{O}_N$};

     \begin{scope}[shift={(1,0)}]
     \draw (2,0)--(3,0); \node at (2.5,0.2) {$\Ss_1$}; \draw (2,-2)--(3,-2); \node at (2.5,-1.8) {$\Ss_m$}; \node at (2.5,-1) {$\cdot$}; \node at (2.5,-1.3) {$\cdot$}; \node at (2.5,-0.7) {$\cdot$};
     
      \draw (3,0.3) rectangle node{$\matholdcal{E}^1$} (4,-2.3);
   \end{scope}   

\begin{scope}[shift={(3,0)}]
           \draw (2,0)--(3,0); \node at (2.5,0.2) {$\Ss_1$}; \draw (2,-2)--(3,-2); \node at (2.5,-1.8) {$\Ss_m$}; \node at (2.5,-1) {$\cdot$}; \node at (2.5,-1.3) {$\cdot$}; \node at (2.5,-0.7) {$\cdot$};
\end{scope}

\node at (7,-1) {$\dots$};

\begin{scope}[shift={(6,0)}]
           \draw (2,0)--(3,0); \node at (2.5,0.2) {$\Ss_1$}; \draw (2,-2)--(3,-2); \node at (2.5,-1.8) {$\Ss_m$}; \node at (2.5,-1) {$\cdot$}; \node at (2.5,-1.3) {$\cdot$}; \node at (2.5,-0.7) {$\cdot$};
\end{scope}

 \draw (9,1.3) rectangle node{$\matholdcal{M}^{\A_N}_{unitary}$} (11,-5.3); \draw (11,1)--(13.5,1); \node at (13,1.2) {$\M_N$}; 
 \draw (0,-5)--(9,-5); \node at (0.5,-4.8) {$\M_N$}; \node at (-0.5,-5) {$\ket{0}_{\M_N}$}; \node at (0.5,-3.5) {$\cdot$}; \node at (0.5,-3.8) {$\cdot$}; \node at (0.5,-4.1) {$\cdot$};
 \begin{scope}[shift={(9,0)}]
           \draw (2,0)--(3,0); \node at (2.5,0.2) {$\Ss_1$}; \draw (2,-2)--(3,-2); \node at (2.5,-1.8) {$\Ss_m$}; \node at (2.5,-1) {$\cdot$}; \node at (2.5,-1.3) {$\cdot$}; \node at (2.5,-0.7) {$\cdot$};
\end{scope}
     \draw (12,0.3) rectangle node{$\matholdcal{E}^N$} (13,-2.3);
    \end{tikzpicture}
    \caption{A $\matholdcal{C}^{sys+anc}$-form standard quantum circuit where each measurement in \cref{fig:circuit_sys_form} is purified to a unitary using an ancilla. Here $\mathtt{O}_i$ is the set of possible non-trivial values of the outcome $a_i$. The measurements $\text{$\matholdcal{M}$}^{\M_i}:=\{\pi^{\M_i}_{a_i}\}_{a_i}$ applied in the global future, implement the isomorphic projective measurement of the ancilla $\M_i$, as do the measurements $\text{$\matholdcal{M}$}^{\A_i}$ (from \cref{fig:circuit_sys_form}) on the  subset $\mathtt{S}_i\subseteq \mathtt{S}:=\{\Ss_1,...,\Ss_m\}$ of systems. It is immediate to see (and well-known) that the present circuit and that of \cref{fig:circuit_sys_form} are operationally equivalent. }
    \label{fig:circuit_sys+anc_form}
\end{figure*}

These two forms of circuits are illustrated in \cref{fig:circuit_sys_form,fig:circuit_sys+anc_form}.

\begin{restatable}[Recovering standard quantum circuits]{theorem}{StandardQT}
\label{theorem: standardQT}
If an EWFS corresponds to a standard quantum scenario (\Cref{def:std_q_exp}), then its augmented circuit can be
equivalently reduced to a standard quantum circuit, such that the same (non-trivial) predictions are obtained from the original augmented circuit, the $\matholdcal{C}^{sys}$-form standard circuit or the $\matholdcal{C}^{sys+anc}$-form standard circuit. Explicitly, for any disjoint sets $\vec{a}_j=(a_{j_1},...,a_{j_p})$ and $\vec{a}_l=(a_{l_1},...,a_{l_q})$ of outcomes, and any choice of settings $\vec{x}=\vec{\xxx}$ such that $x_i=1$ for all $i\in \{j_1,...,j_p,l_1,...,l_q\}$, we have

\begin{align}
    \begin{split}
        P_{aug}(\vec{a}_j=\vec{\aaa}_j|\vec{a}_j=\vec{\aaa}_j,
        \vec{x}=\vec{\xxx})
        =P_{std}(\vec{a}'_j=\vec{\aaa}_j|\vec{a}'_j=\vec{\aaa}_j),
    \end{split}
\end{align}
where the $P_{aug}$ refers to setting-conditioned predictions in the augmented circuit of the EWFS and $P_{std}$ refers to predictions in an equivalent  $\matholdcal{C}^{sys}$-form or  $\matholdcal{C}^{sys+anc}$-form standard quantum circuit (where no settings are involved). 
\end{restatable}
A proof of this theorem can be found in \Cref{appendix: proofs}.

\section{Detailed analysis of the FR scenario}

\subsection{Entanglement version of the FR scenario}
\label{appendix: ent_FR}

In the main text, we provided a brief overview of the entanglement version of the FR scenario as well as a simple explanation of our resolution of the paradox. Having reviewed this scenario in detail in \Cref{sec: FR_ent_review}, we now provide a more detailed analysis of the same, showing the explicit calculation for every prediction involved.

We reproduce each of the individual statements used in the reasoning of the FR scenario described in \Cref{sec: FR_ent_review}, i.e., those captured by \cref{eq: chain1,eq: chain3}. Additionally, this reasoning occurs only in a round where the super-agents observe the outcomes $u=w=\textup{ok}$, which is associated with a probability $\frac{1}{12}$ in FR's arguments as shown in \cref{eq:FRpostselProb}. We start by reproducing this probability as an explicit setting conditioning prediction in our framework. The probability of \cref{eq:FRpostselProb} is equivalent to conventional prediction $P^{FR}_{conv}(u=w=\text{ok})$ (\cref{def: conv_prediction}) of the FR scenario,  in computing this, FR apply the $\textup{U}$ assumption implicitly to model the measurements of the agents Alice and Bob purely unitarily (using $\matholdcal{M}^A_{unitary}$ and $\matholdcal{M}^B_{unitary}$ as seen in \cref{eq:FRpostselProb}). This corresponds precisely to the setting choices $x_A=x_B=0$ (as expected from the general mapping from conventional to setting-conditioned predictions given in \Cref{theorem: main}). Therefore it follows immediately that 
\begin{align}
\begin{split}
    P(u=w=\textup{ok}|\ket{\psi}^{t=2}):=P^{FR}_{conv}(u=w=\text{ok})=P(u=w=\text{ok}|x_A=x_B=0)
    \end{split}
\end{align}

Note that in our notation for predictions, we don't explicitly condition on the initial state of the scenario as this is taken to be in the common knowledge of all agents. See \Cref{sec: reasoning_rules} for a discussion on how our framework and arguments can generalise to the case where one relaxes this common knowledge assumption. 

Now we proceed to analysing each of the statements that agents make when the above post-selection on $u=w=\textup{ok}$ succeeds. Consider the statement obtained in \cref{eq: chain1}, here Ursula, upon knowing that $u=w=\textup{ok}$ reasons about Bob's outcome $b$ using the state $\ket{\psi^{t=2}}_{\R\A\Ss\B}$ and concludes that $u=w=\textup{ok} \Rightarrow b=1$. This logical statement is equivalently expressed in probabilities through the conventional prediction $P_{conv}(b=1|u=w=\textup{ok})=1$.

We can readily extract the setting choices implicit in the calculation of this probability for the FR protocol. Note that $\ket{\psi^{t=2}}$ is obtained by applying $M^{\A}_{unitary}\otimes M^{\B}_{unitary}$ to the initial state $\ket{\psi^{t=1}}$ and in calculating the above-mentioned probability for $b=1$ using the Born rule, FR apply the projector $\pi^{\B}_{1,1}=\ket{11}\bra{11}_{\Ss\B}$ to $\ket{\psi^{t=2}}$. This precisely corresponds to the assigning $x_2=1$ for Bob's setting. On the other hand, they model Alice's measurement as a purely unitary evolution ($M^{\A}_{unitary}$) as seen by Ursula in this reasoning step, therefore the setting choice used for Alice in this reasoning is $x_1=0$. Making these setting choices explicit, we see that this probability calculated in the FR reasoning is equivalent to the setting-conditioned prediction $P(b=1|u=w=\textup{ok},(x_1,x_2)=(0,1))=1$ in our framework.

Indeed one can calculate this prediction from the augmented circuit of \cref{fig: FR_ent_circuit_aug} for the FR protocol (using the Born rule and the well-known rule for conditional probabilities) and would obtain the same. We demonstrate this below, for further details on how the probability rule for setting-conditioned predictions in augmented circuits is derived, see \Cref{appendix: prob_rule_general}.

\begin{align}
	\begin{split}
		P(b=1|u=w=\textup{ok},(x_1,x_2)=(0,1)) = \frac{P(b=1,u=w=\textup{ok}|(x_1,x_2)=(0,1))}{P(u=w=\textup{ok}|(x_1,x_2)=(0,1))}.
	\end{split}
\end{align}

That this expression evaluates to unit probability is evident from the following calculation of the numerator and denominator of this expression for the FR protocol.

\begin{align}
	\begin{split}
		&P(b=1,u=w=\textup{ok}|(x_1,x_2)=(0,1))=\frac{1}{12}\\
		=&|\bra{\textup{ok}}_{\R\A}\otimes \bra{\textup{ok}}_{\Ss\B}\Big(\id_{RA}\otimes \pi^{\B}_{x_2=1,b=1}\Big) \ket{\psi^{t=2}}|^2\\
		=&|\bra{\textup{ok}}_{\R\A}\otimes \bra{\textup{ok}}_{\Ss\B}\Big(\id_{RA}\otimes \ket{11}\bra{11}_{\Ss\B}\Big) \ket{\psi^{t=2}}|^2.
	\end{split}
\end{align}

\begin{align}
	\begin{split}
		&P(u=w=\textup{ok}|(x_1,x_2)=(0,1))=\frac{1}{12}\\
		=&\sum_{b\in\{0,1\}}|\bra{\textup{ok}}_{\R\A}\otimes \bra{\textup{ok}}_{\Ss\B}\Big(\id_{RA}\otimes \pi^{\B}_{x_2=1,b}\Big) \ket{\psi^{t=2}}|^2\\
		=&\sum_{b\in\{0,1\}}|\bra{\textup{ok}}_{\R\A}\otimes \bra{\textup{ok}}_{\Ss\B}\Big(\id_{RA}\otimes \ket{bb}\bra{bb}_{\Ss\B}\Big) \ket{\psi^{t=2}}|^2.
	\end{split}
\end{align}

Having formalised FR's logical statement $u=w=\textup{ok}\Rightarrow b=1$ as the setting-conditioned prediction $P(b=1|u=w=\textup{ok},(x_1,x_2)=(0,1))=1$ in our framework, we obtain the corresponding explicit version of the statement: 
\begin{equation}
	u=w=\textup{ok} \land (x_1,x_2)=(0,1)\Rightarrow b=1.
\end{equation}

We now proceed to the next statement of the FR reasoning, given by \cref{eq: chain2}, where Ursula reasons about Bob's reasoning of Alice through the statement $b=1\Rightarrow a=1$. This is equivalently expressed in terms of probabilities through the conventional prediction $P_{conv}(a=1|b=1)=1$.

To evaluate this probability using the Born rule as FR do, we must apply the projector $\pi^{\A}_{1,1}\otimes\pi^{\B}_{1,1}=\ket{11}\bra{11}_{\R\A}\otimes \ket{11}\bra{11}_{\Ss\B}$ to $\ket{\psi^{t=2}}_{\R\A\Ss\B}$ and it is evident that the implicit setting choices for Alice and Bob used here are $(x_1,x_2)=(1,1)$. Therefore, making this explicit, we have $P(a=1|b=1, (x_1,x_2)=(1,1))=1$. We can again verify this from the probability rule for our augmented circuit (which is simply the Born rule and standard conditional probability rule).
\begin{align}
	\begin{split}
		P(a=1|b=1, (x_1,x_2)=(1,1))=\frac{P(a=1,b=1|(x_1,x_2)=(1,1))}{P(b=1|(x_1,x_2)=(1,1))}
	\end{split}
\end{align}

That this evaluates to unit probability is immediate from the following expressions for the numerator and denominator.

\begin{align}
	\begin{split}
		&P(a=1,b=1|(x_1,x_2)=(1,1))=\frac{1}{3}\\
		=&|\Big(\pi^{\A}_{x_1=1,a=1}\otimes\pi^{\B}_{x_2=1,b=1}\Big)\ket{\psi^{t=2}}|^2\\
		=&|\Big(\ket{11}\bra{11}_{\R\A}\otimes \ket{11}\bra{11}_{\Ss\B}\Big)\ket{\psi^{t=2}}|^2.
	\end{split}
\end{align}

\begin{align}
	\begin{split}
		&P(b=1|(x_1,x_2)=(1,1))=\frac{1}{3}\\
		=&|\Big(\id_{\R\A}\otimes\pi^{\B}_{x_2=1,b=1}\Big)\ket{\psi^{t=2}}|^2\\
		=&|\Big(\id_{\R\A}\otimes \ket{11}\bra{11}_{\Ss\B}\Big)\ket{\psi^{t=2}}|^2.
	\end{split}
\end{align}

From this, as before, we can extract the explicit version of the logical statement.
\begin{equation}
	b=1 \land (x_1,x_2)=(1,1)\Rightarrow a=1.
\end{equation}

We now turn to the third statement of the FR reasoning given in \cref{eq: chain3}, where Ursula reasons about Bob's reasoning about Alice's reasoning about Wigner, through the statement $a=1\Rightarrow w=\textup{fail}$. This equivalent probabilistic version is given by the conventional prediction $P_{conv}(w=\textup{fail}|a=1)=1$.

Analysing how FR calculate this probability using the Born rule, we see that this involves applying the projector $\pi^{\A}_{1,1}=\ket{11}\bra{11}_{\R\A}$ to the state $\ket{\psi^{t=2}}$ which gives us Alice's setting $x_1=1$. Moreover, Bob is modelled purely unitarily here, through $M^{\B}_{unitary}$ and we have $x_2=0$. Therefore, the explicit setting-conditioned prediction corresponding to this reasoning step of FR is $P(w=\textup{fail}|a=1,(x_1,x_2)=(1,0))=1$. This can be verified within our framework as follows.

\begin{align}
	\begin{split}
		P(w=\textup{fail}|a=1,(x_1,x_2)=(1,0))
		= \frac{P(w=\textup{fail},a=1|(x_1,x_2)=(1,0))}{P(a=1|(x_1,x_2)=(1,0))}
	\end{split}
\end{align}

The numerator and denominator are evaluated below, which makes it evident that the expression above evaluates to unity.

\begin{align}
	\begin{split}
		&P(w=\textup{fail},a=1|(x_1,x_2)=(1,0))=\frac{2}{3}\\
		=&  |\id_{\R\A}\otimes \bra{\textup{fail}}_{\Ss\B}\Big(\pi^{\A}_{x_1=1,a=1}\otimes\id_{\Ss\B}\Big)\ket{\psi^{t=2}}|^2\\
		=&|\id_{\R\A}\otimes \bra{\textup{fail}}_{\Ss\B}\Big(\ket{11}\bra{11}_{\R\A}\otimes\id_{\Ss\B}\Big)\ket{\psi^{t=2}}|^2
	\end{split}
\end{align}

\begin{align}
	\begin{split}
		&P(a=1|(x_1,x_2)=(1,0))=\frac{2}{3}\\
		=&  |\Big(\pi^{\A}_{x_1=1,a=1}\otimes\id_{\Ss\B}\Big)\ket{\psi^{t=2}}|^2\\
		=&|\Big(\ket{11}\bra{11}_{\R\A}\otimes\id_{\Ss\B}\Big)\ket{\psi^{t=2}}|^2
	\end{split}
\end{align}

Then the corresponding, explicit version of the logical statement is,

\begin{equation}
	a=1 \land (x_1,x_2)=(1,0)\Rightarrow w=\textup{fail}.
\end{equation}

Therefore, we have explicitly derived all the statements in \cref{table: resolution_ent} while highlighting the setting choices implicit in each of FR's statements. As we have seen, making explicit these setting choices is sufficient to resolve the apparent paradox.

\subsection{Prepare and measure version of the FR scenario}
\label{ssec: resolution_prep}

We have reviewed the original prepare and measure of the FR scenario in \Cref{sec: FR_PM_review}.
Here we show that the resolution proposed for the entanglement version is also applicable to the original (prepare and measure) version of the FR paradox, ref.~\cite{Frauchiger2018}. We will first show explicitly how original version of the FR protocol can also be modelled as an augmented circuit. Then, we will go through the reasoning of the prepare and measure version on a statement-by-statement basis of their apparent proof of their paradox. For every statement, we reveal the different settings which said statements are contingent on, but not stated by the authors of the FR paper.

{\bf The augmented circuit}
The augmented circuit of the original prepare and measure version of the FR protocol is given in \cref{fig: FR_PM_aug_circuit1}. An equivalent version of this circuit  is given in \cref{fig: FR_PM_aug_circuit2}, that makes the mapping to the entanglement version of the FR scenario more explicit. 

The circuits encode the states and measurements of the original protocol in the computational basis as follows.
The initial state $\sqrt{\frac13}\ket{\textup{heads}}_\R+\sqrt{\frac23}\ket{\textup{tails}}_\R$ of the coin in $\mb{F}$'s lab (\cref{eq: FR_PM_instate}) is represented in the computational basis as $\ket{\psi}_R=\sqrt{\frac13}\ket{0}_\R+\sqrt{\frac{2}{3}}\ket{1}_\R$, all other systems ($\mb{L}\backslash$, $S$ and $L \backslash$) are initialised to $\ket{0}$. Then the measurement of $\mb{F}$ corresponds to a computational basis measurement, with the outcome $r=head$ identified with $r=0$ and $r=\textup{tails}$ identified with $r=1$. The preparation of $S$ carried out by $\mb{F}$, based on the outcome of their measurement on $R$ corresponds to a controlled Hadamard with the states $\ket{\downarrow}_{\Ss}$, $\ket{\uparrow}_{\Ss}$ of the FR scenarios represented as $\ket{0}_{\Ss}$ and $\ket{1}_{\Ss}$ here.
Similarly, the measurement of $\textup{F}$ then also becomes a computational basis measurement with the outcome $z=+\frac{1}{2}$ identified with $z=0$ and $z=-\frac{1}{2}$ identified with $z=1$. The projectors $\mathtt{\Pi}^{\mb{F}}_{x_1}$ and $\mathtt{\Pi}^{F}_{x_2}$ acting on the systems $\R\Lbar\backslash$ and $\Ss\Lab\backslash$ respectively are identical to the projectors $\mathtt{\Pi}^{A}_{x_1}$ and $\mathtt{\Pi}^{B}_{x_2}$ of \cref{eq: FR_aug_proj} acting on the systems $\R\A$ and $\Ss\B$, and the final measurements of \=W and W are the same as the entanglement version of \Cref{sec: FR_ent_review}, i.e., $\{\ket{\textup{ok}}_{\Ss\Lab\backslash}=\frac{1}{\sqrt{2}}(\ket{00}-\ket{11})_{\Ss\Lab\backslash},\ket{\textup{fail}}_{\Ss\Lab\backslash}:=\frac{1}{\sqrt{2}}(\ket{00}+\ket{11})_{\Ss\Lab\backslash}\}$ and similarly for $\{\ket{\mb{\textup{ok}}}_{\R\Lbar\backslash},\ket{\mb{\textup{fail}}}_{\R\Lbar\backslash}\}$.

\begin{figure*}
	\includegraphics[scale=0.8]{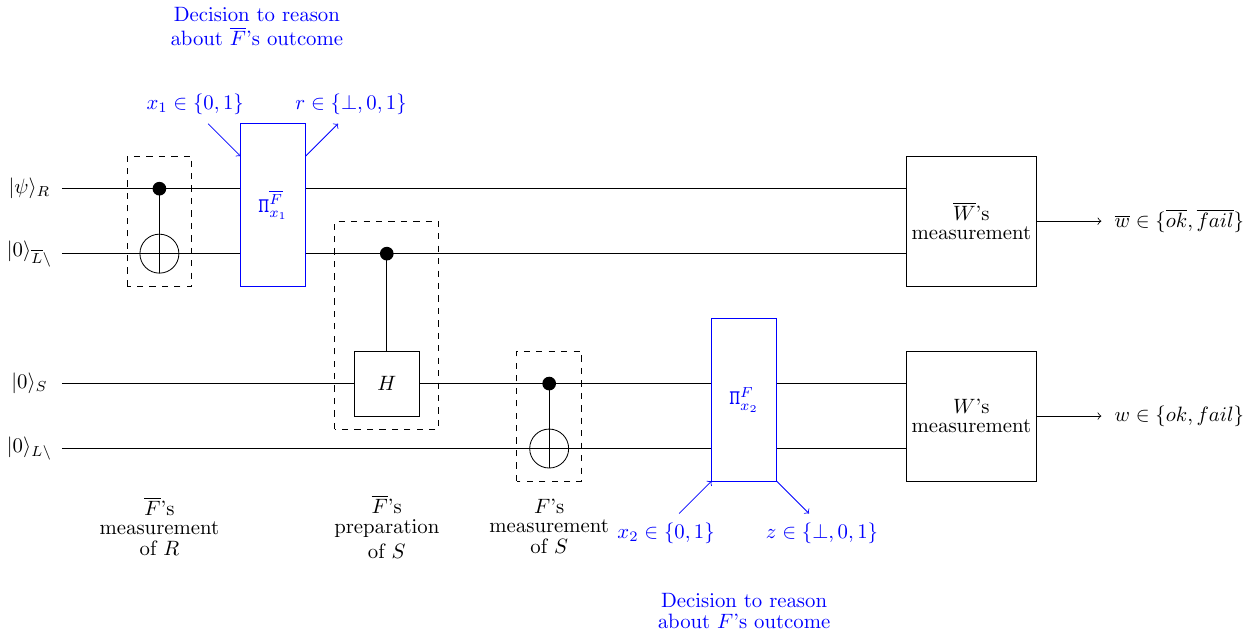}
	\caption{Augmented circuit for the prepare and measure version of the FR scenario expressed in terms of the computational basis, as explained in the main text. }
	\label{fig: FR_PM_aug_circuit1}
\end{figure*}

\begin{figure*}
	\includegraphics[scale=0.8]{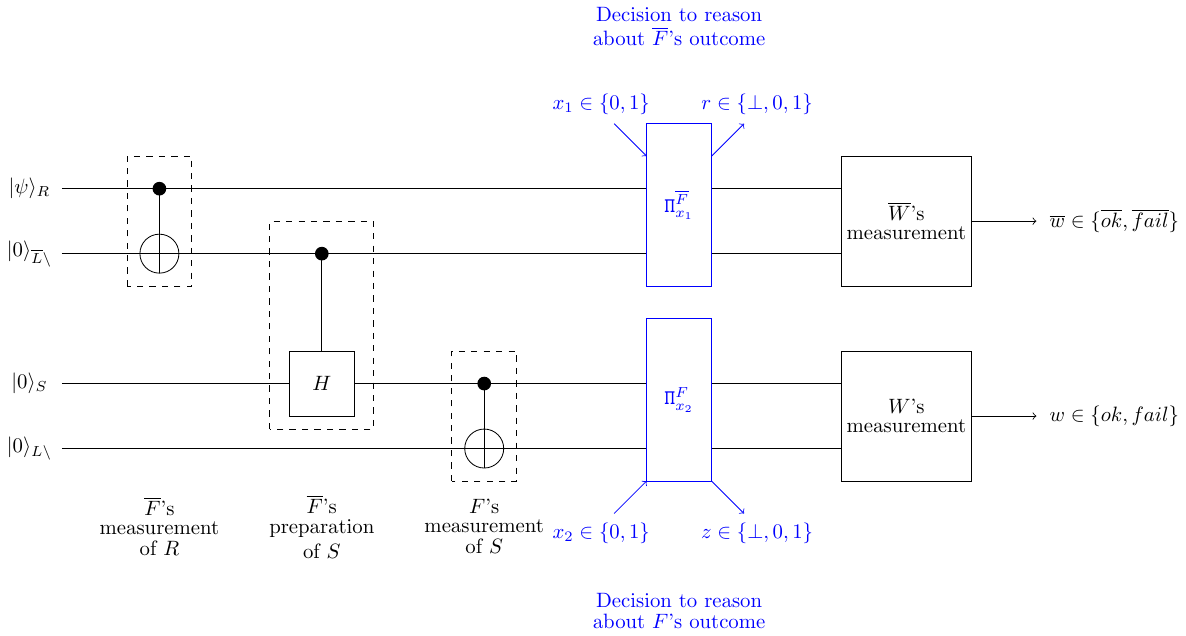}
	\caption{Equivalent version of the augmented circuit (\cref{fig: FR_PM_aug_circuit1}) of the prepare and measure FR protocol. Note that the blue box corresponding to $\mathtt{\Pi}^{\mb{F}}_{x_1}$ commutes with the controlled Hadamard gate since the control is on the same basis as the measurements associated with the box. Then, it is easy to verify that the joint state of $\R\Lbar\backslash\Ss\Lab$ just before the blue boxes (i.e., just after F's measurement of $\Ss$) is precisely the same state as $\ket{\psi^{t=2}}_{\R\A\Ss\B}$ (\cref{eq: state_main}) of the entanglement version  of the FR scenario, with $\Lbar\backslash$ and $\Lab\backslash$ playing the role of $\A$ and $\B$. }
	\label{fig: FR_PM_aug_circuit2}
\end{figure*}

{\bf Statement-by-statement analysis}
We will now analyse their constructive proof step by step and point out what settings are required to reproduce their statements. We will use a notation which is close to, but not identical to, that of the authors to aid comparison.

The authors use the following notation to denote the statements by specific agents at particular times:  $G^{t}$:``k'' where $G\in\{\textup{W}, \textup{\tb  W}, \textup{F}, \textup{\tb  F}\}$ denotes the agent $G$ making statement $k$ at time $t$. Here $x_1$ are for the measurements of \tb F while $x_2$ are for F's. Analogously to in the entanglement scenario, we will not need to consider settings for W nor \tb W since they are super agents. We will use italics when referring to the reasoning of the authors. What's more we will use the same notation as the authors to specify different statements, but with an additional subscript indicating the $(x_2,x_1)$ setting that said observer is using when making the statement in the corresponding augmented circuit: $G^{t}_{(x_1,x_2)}$:``k''. The setting $(0,0)$ corresponds to the case where the agent is not reasoning about F nor \tb F's measurement outcome. Note that the reasoning of a particular observer only depends on their choice of $(x_1,x_2)$ settings and not on those of the other observers. Furthermore, we will say $G^{t}_{(x_1,x_2)}= \mathtt{correct}$ if $G^{t}$ holds under settings $(x_1,x_2)$, $G^{t}_{(x_1,x_2)}= \mathtt{false}$ if statement $G^{t}$ does not hold under settings $(x_1,x_2)$, and $G^{t}_{(x_1,x_2)}=\emptyset$ if statement $G^{t}$ involves reasoning about an observer's measurement outcome when the choice $(x_1,x_2)$ does not allow observer $G$ to reason about said observer.

The first statement is made by \tb  F at time $n\!:\!00$: \emph{$\textup{\tb {F}}^{n:00}$ :``The value $w$ is obtained by a measurement of L w.r.t.  basis $\{ \pi_{w=\textup{ok}}^{n:10},\pi_{w=\textup{fail}}^{n:10} \}$, which is completed at time $n:31$''.} 
Here \phantom{$\pi_{w=\textup{ok}_P}^{n:10^A}$}\hspace{-9mm}$\pi_{w=\textup{ok}}^{n:10}$, $\pi_{w=\textup{fail}}^{n:10}$ are the projectors onto the $ \ket{{\textup{ok}}}_\Lab$, $\ket{{fail}}_\Lab$ basis. This statement is correct according to the framework of this paper for all $(x_1,x_2)$ i.e. $\textup{\tb {F}}^{n:00}_{(0,0)}=\textup{\tb {F}}^{n:00}_{(0,1)}=\textup{\tb {F}}^{n:00}_{(1,0)} = \textup{\tb {F}}^{n:00}_{(1,1)}=\mathtt{correct}$.

\emph{Then, if \tb  F got $r=$ tails in her measurement of the coin flip, she would make the statement  $\textup{\tb {F}}^{n:01}$: ``The spin S is in state $\ket{\rightarrow}_\Ss$ at time $n\!:\!10$''.} In our formalism, this statement would also hold since \tb F is only reasoning about the measurement she made, we thus have  $\textup{\tb {F}}^{n:01}_{(1,0)}=\textup{\tb {F}}^{n:01}_{(1,1)}=\mathtt{correct}$ and $\textup{\tb {F}}^{n:01}_{(0,0)}=\textup{\tb {F}}^{n:01}_{(0,1)}=\emptyset$. Any further statements made by agent \tb  F would have to be pre-selected on the spin being in state $\ket{\rightarrow}_\Ss$. In \cref{fig:FR paper experiment} this would correspond to the selection of branch 1 and multiplying it by $\sqrt{3/2}$ to re-normalise the branch. The authors then go on to make the following claim: \emph{$\textup{\tb {F}}^{n:00}$ and $\textup{\tb {F}}^{n:01}$ inserted into Q imply $w=$ \textup{fail}} (this is claim $\textup{\tb {F}}^{n:02}$ in table 3). This statement only holds if one does not take into account F's measurement at time $n\!:\!10$. In other words, \phantom{$\textup{\tb {F}}^{n:02^A}_{(1,0)_P}$}\hspace{-10mm}$\textup{\tb {F}}^{n:02}_{(1,0)}=\mathtt{correct}$, $\textup{\tb {F}}^{n:02}_{(1,1)}=\mathtt{false}$. To see this, note from \cref{fig:FR paper experiment} that their conclusion follows from noting the cancellation of the two blue $\ket{\textup{ok}}_{\Lab}$ kets when we pre-select on branch 2. However, branch 2 is further split into sub-branches 2.1.1. and 2.1.2. via F's measurement. This splitting causes the blue coloured kets to not cancel each other out from \tb F's perspective when reasoning under settings $(x_1,x_2)=(1,1)$. In terms of equations, under settings $(x_1,x_2)=(1,1)$ we would conclude that $\textup{\tb {F}}^{n:00}$ and $\textup{\tb {F}}^{n:01}$ imply that the probability that $w=\textup{fail}$ is 

\begin{align}
\begin{split}
	&P\big(w=\textup{fail}|r=\textup{tails}, (x_1,x_2)=(1,1)\big)\\
	&=\sum_{z\in\{-1/2,1/2\}}\!\!\!\!\!\!P\big(w=\textup{fail}, z |r=\textup{tails}, (x_1,x_2)=(1,1)\big)\\
	&=\frac{\sum_{z\in\{-1/2,1/2\}}\tr\left[ \pi_\textup{fail}^{\textup{W}} \pi_z^{\textup{F}} \pi_\textup{tail}^{\textup{\tb F}}\rho_\textup{Fi}\pi_\textup{tail}^{\textup{\tb F}}\pi_z^{\textup{F}} \right]}{\sum_{\substack{z\in\{-1/2,1/2\} \\ w\in\{\textup{ok},\textup{fail}\}}}\tr\left[ \pi_w^{\textup{W}} \pi_z^{\textup{F}} \pi_\textup{tail}^{\textup{\tb F}}\rho_\textup{Fi}\pi_\textup{tail}^{\textup{\tb F}}\pi_z^{\textup{F}} \right]}\\
	&=\sum_{z\in\{-1/2,1/2\}}\frac{\tr\left[ \pi_\textup{fail}^{\textup{W}} \pi_z^{\textup{F}} \pi_\textup{tail}^{\textup{\tb F}}\rho_\textup{Fi}\pi_\textup{tail}^{\textup{\tb F}}\pi_z^{\textup{F}} \right]}{\tr\left[ \pi_\textup{tail}^{\textup{\tb F}}\rho_\textup{Fi} \right]}=\frac12,
\end{split}\label{eq:P w given a rightarrow old}
\end{align}	
where $\rho_\textup{Fi}:=  U_{F} U_{\textup{spin}} U_{\mb{F}} \rho_0 ( U_{F} U_{Spin} U_{\mb{F}} )^\dag$ with $\rho_0:=\proj{\psi_0}$, $\ket{\psi_0}$ being the initial state (i.e. l.h.s. of \cref{eq:final state FR version}). The unitaries $U_{F}$, $U_{\textup{spin}}$, $U_{\mb{F}}$ are those of the protocol (see \cref{fig:FR paper experiment}) and give rise to the final state $\rho_\textup{Fi}=\proj{\psi_\textup{Fi}}$ with $\ket{\psi_\textup{Fi}}$ given by the r.h.s. of \cref{eq:final state FR version}. F's measurement outcome at time $n\!:00$ is denoted by $r$ (with $a$ denoting the corresponding random variable) and $\{\pi_z^{\textup{F}}\}_{z\in\{-1/2,1/2\}}$ is the PVM of F's measurement at time $n\!:\!10$. We note that in the notation of the general framework of \Cref{sec:gen_framework}, these projectors would be explicitly written as $\{\pi_{1,z}^{\textup{F}}\}_{z\in\{-1/2,1/2\}}$ as they correspond to the case of choosing the setting to be 1 for that measurement. In order to avoid clutter, here and in the following, we drop the setting subscript ``1'' in all such projectors as the meaning is evident from the context of the protocol at hand. In \cref{eq:P w given a rightarrow old} we have used the fact that the unitary transformations taking us from the initial state to the final state commute with the measurement projectors, i.e. we have used the equivalence between the circuits of~\cref{fig: FR_PM_aug_circuit1,fig: FR_PM_aug_circuit2}.

Since this probability in \cref{eq:P w given a rightarrow old} is less than one, we can conclude that $\textup{\tb {F}}^{n:02}_{(1,1)}=\mathtt{false}$.

Now it is F's turn to make a statement. \emph{$\textup{F}^{n:10}$: ``The value $z$ is obtained by a measurement of spin S with respect to $\{\pi_{z=-1/2}^\textup{n:10},\pi_{z=1/2}^\textup{n:10}\}$, which is completed at time $n\!:\!11$''.} This statement clearly holds for all setting choices for F in which F can reason about her outcome, since it is merely stating one of the rules of the protocol: $\textup{F}^{n:10}_{(0,1)}=\textup{F}^{n:10}_{(1,1)}= \mathtt{correct}$. They now further go on to state: \emph{Suppose now that F observed $z=1/2$ in round $n$. Since $\braket{\downarrow|\pi_{z=-1/2}^{n:10}|\downarrow}=1$, it follows from Q that S was not in state $\ket{\downarrow}_\Ss$, and hence that the random value $r$ was not heads. Therefore $\textup{F}^{n:12}$:``I am certain that \tb  F knows that $r$= tails at time $n\!:\!11$''.} Here F is reasoning about both her measurement outcome and that of \tb F. Therefore, by definition, this statement only makes sense when F chooses setting $(x_1,x_2)=(1,1)$ since in the case $(x_1,x_2)=(0,1)$, F cannot reason about \tb F's measurement outcome since there is no classical outcome to assign to it. We therefore have $\textup{F}^{n:12}_{(0,1)}=\textup{F}^{n:12}_{(1,0)}=\emptyset$.
Meanwhile, the following equation verifies that $\textup{F}^{n:12}_{(1,1)}=\mathtt{correct}$. Using \cref{eq:general conditional prob rule simplicied 2}, if we post-select on F getting $z=1/2$  in round $n$, then we find that the probability that \tb  F got tails is one:
\begin{align}
\begin{split}
	P\big(r=\textup{tails}|z=1/2, (x_1,x_2)=(1,1)\big)=\frac{\tr[ \pi_{1/2}^\textup{F} \pi_\textup{tail}^\textup{\tb  F}\rho_\textup{Fi}\pi_\textup{tail}^\textup{\tb  F}]}{\sum_{r\in\{\textup{tails},\textup{heads}\}}\tr[ \pi_{1/2}^\textup{F} \pi_r^\textup{\tb  F}\rho_\textup{Fi}\pi_r^\textup{\tb  F}]}=1,
\end{split}\label{eq:tails given a half}
\end{align}
where $\pi_{\textup{tails}}^{\textup{\tb  F}}=\proj{\mb t}_\Lbar$, $\pi_{heads}^{\textup{\tb  F}}=\proj{\mb h}_\Lbar$ are the projectors onto the lab of \tb F, corresponding to the two outcomes of the coin toss. The last equality follows from noting that $\tr[ \pi_{1/2}^\textup{F} \pi_\textup{heads}^\textup{\tb  F}\rho_\textup{Fi}\pi_\textup{heads}^\textup{\tb  F}]=0$.

Similarly, the above equation can also be concluded from \cref{fig:FR paper experiment} by noting that if we are in branch 2.1.2 (this corresponds to F getting outcome $z=1/2$), then the only measurement outcome of \tb  F which leads to this branch is $r$ = tails.

The authors of the FR paradox then claim: \emph{Therefore from $\textup{F}^{n:12}$ and invoking Q, we conclude $\textup{F}^{n:13}$: ``I am certain that \tb  F is certain that W will observe $w=\textup{fail}$ at time $n\!:\!31$.''}. Statements $\{\textup{F}^{n:13}_{(x_1,x_2)}\}_{x_1,x_2}$ do not correspond to a single statement in our framework since $\textup{F}^{n:13}$ is a concatenation of two statements ``upon observing $z=1/2$, I know with certainty that $r=\textup{tails}$'' made by F and ``upon observing $r=\textup{tails}$, I know with certainty that $w=\textup{fail}$'' made by $\mb{F}$. Note that these two statements are precisely FR's $\textup{F}^{n:12}$ and $\textup{\tb{F}}^{n:02}$ respectively and they are combined using the assumptions $\textup{C}$ and $\textup{D}$ to give $\textup{F}^{n:13}$. Moreover, as we have seen that $\textup{\tb {F}}^{n:02}_{(1,0)}=\mathtt{correct}$, $\textup{\tb {F}}^{n:02}_{(1,1)}=\mathtt{false}$ and $\textup{F}^{n:12}_{(0,1)}=\emptyset$ and  $\textup{F}^{n:12}_{(1,1)}=\mathtt{correct}$, hence there is no common setting $(x_1,x_2)$ for which both $\textup{\tb F}^{n:02}$ and $\textup{F}^{n:12}$ are correct. Hence there are no setting choice under which $\textup{F}^{n:13}$ can be derived from $\textup{F}^{n:12}$ and $\textup{\tb{F}}^{n:02}$ as FR do. Agent F can always inherit both statements via the knowledge operator \cref{eq: C_logic}, in a similar way to how you, the reader, is ``inheriting'' all the statements in this article when you read them. This, however, by itself poses little value due to the setting mismatch.

Alternatively, one can attempt a more direct derivation of the prediction associated with $\textup{F}^{n:13}$, which tells us something about the outcome $w=$fail based on the observation of the outcome $z=1/2$. This requires the setting $x_2=1$.
Then, from F's perspective under settings $(x_1,x_2)=(1,1)$ and after obtaining measurement outcome $z=1/2$, she would conclude that the probability of $w$= fail is only
\begin{align}
\begin{split}
	&P\big(w=\textup{fail}|z=1/2, (x_1,x_2)=(1,1)\big)\\
	&=\sum_{r\in\{\textup{heads},\textup{tails}\}}P\big(r|z=1/2, (x_1,x_2)=(1,1)\big)P\big(w=\textup{fail}|r, z=1/2, (x_1,x_2)=(1,1)\big)\\
	&=P\big(w=\textup{fail}|r=\textup{tails}\, \&\, z=1/2, (x_1,x_2)=(1,1)\big)\\
	&=\frac{\tr[\pi_\textup{fail}^{\textup{W}} \pi_{1/2}^\textup{F} \pi_\textup{tails}^\textup{\tb  F}\rho_\textup{Fi}\pi_\textup{tails}^\textup{\tb  F} \pi_{1/2}^\textup{F}]}{\sum_{w\in\{\textup{ok},\textup{fail}\}}\tr[\pi_w^{\textup{W}} \pi_{1/2}^\textup{F} \pi_\textup{tails}^\textup{\tb  F}\rho_\textup{Fi}\pi_\textup{tails}^\textup{\tb  F} \pi_{1/2}^\textup{F}]}\\
	&=\frac{\tr[\pi_\textup{fail}^{\textup{W}} \pi_{1/2}^\textup{F} \pi_\textup{tails}^\textup{\tb  F}\rho_\textup{Fi}\pi_\textup{tails}^\textup{\tb  F} \pi_{1/2}^\textup{F}]}{\tr[ \pi_{1/2}^\textup{F} \pi_\textup{tails}^\textup{\tb  F}\rho_\textup{Fi}\pi_\textup{tails}^\textup{\tb  F}]}=\frac12<1,
\end{split}
\end{align}
where we have used \cref{eq:tails given a half} and the last inequality follows from noting that if we are on branch 2.1.2., then the possibility of W measuring fail cannot be one since he can get outcome ok too (since due to F's measurement, the blue cancellation does not take place). Note also that the same probability is obtained under the settings $(x_1,x_2)=(0,1)$.

The authors now proceed to reason from the perspective of the super-agents. The first statement is \emph{$\textup{\tb  W}^{n:21}$: ``System R is initialised to $\ket{\textup{init}}_\R$ at time $n\!:\!00$.} This statement is true for all settings, since it is a statement about the protocol, which all agents are assumed to know, so $\textup{\tb  W}^{n:21}_{(0,1)}=\textup{\tb  W}^{n:21}_{(1,0)}=\textup{\tb  W}^{n:21}_{(1,1)}=\mathtt{correct}$. The authors then point out that \emph{the state $U_{\mb{F}}\ket{\textup{init},\phi_0,S_0}$ (r.h.s of \cref{eq:initial to final ket}) is orthogonal to $\ket{\textup{\tb  {\textup{ok}}}}_\Lbar\ket{\downarrow}_\Ss$}. This is indeed correct, as can be readily seen from \cref{fig:FR paper experiment} by observing that $\proj{\downarrow}_\Ss$ projects onto the superposition of branches 1.1.1. and 2.1.1. and that for these branches, the purple $\ket{\textup{\tb  {\textup{ok}}}}_\Lbar$ terms cancel each other out. The authors then state this in the form of an expectation value of a projector, namely $\braket{\textup{init}| \pi_{(\mb w,z)\neq (\textup{\tb {ok}},-1/2)}^{n:00}|\textup{init}}=1$. Here $\pi_{(\mb w,z)\neq (\textup{\tb {ok}},-1/2)}^{n:00}=\id - \pi_{(\mb w,z)= (\textup{\tb {ok}},-1/2)}^{n:00}$, where $\pi_{(\mb w,z)= (\textup{\tb {ok}},-1/2)}^{n:00}$ is a Heisenberg picture projector that would first transform $\ket{\textup{init}}_{\R}$ to 
$U_{\mb{F}}\ket{\textup{init},\phi_0,S_0}$ (through the appropriate isometry that appends $\ket{\phi_0,S_0}_{\Lbar\backslash \Ss}$ and performs $U_{\mb{F}}$) and then projects onto the outcomes  $(\mb w,z)= (\textup{\tb {ok}},-1/2)$, as $\ket{\textup{\tb  {ok}}}\bra{\textup{\tb  {ok}}}_\Lbar\otimes \ket{\downarrow}\bra{\downarrow}_\Ss\otimes \id_{\Lab\backslash}$ i.e, $\pi_{(\mb w,z)\neq (\textup{\tb {ok}},-1/2)}^{n:00}$ is the Heisenberg projector onto the complement of outcomes $(\mb w,z)= (\textup{\tb {ok}},-1/2)$. They then claim \emph{Agent \tb  W, who uses Q, can hence be certain that $(\mb w,z)\neq (\textup{\tb {ok}},-1/2)$} and that this implies (when $\mb w$=\tb {ok}) the statement \emph{$\textup{\tb  W}^{n:22}$: ``I am certain that F knows that $z=1/2$ at time $n\!:\!11$''.} Now the authors are allowing \tb W to take into account the measurement outcome of F in their reasoning but not the measurement of \tb  F. In other words, they are using settings $(0,1)$ and thus assigning

\begin{align}
\begin{split}
	&P\big(z=1/2|\mb w=\textup{\tb {ok}},(x_1,x_2)=(0,1)\big)\\
	&= \sum_{w\in\{\textup{ok},\textup{fail}\}} P\big(z=1/2,  w|\mb w=\textup{\tb {ok}},(x_1,x_2)=(0,1)\big)\\
	&=\sum_{w\in\{\textup{ok},\textup{fail}\}}\frac{\tr[\pi_w^{\textup{W}}\pi_\textup{\tb  ok}^{\textup{\tb W}} \pi_{1/2}^\textup{F} \rho_\textup{Fi} \pi_{1/2}^\textup{F}]}{\sum_{\substack{z\in\{-1/2,1/2\}\\ w'\in\{\textup{ok},\textup{fail}\}}} \tr[\pi_{w'}^{\textup{W}}\pi_\textup{\tb  ok}^{\textup{\tb W}} \pi_{z}^\textup{F} \rho_\textup{Fi} \pi_{z}^\textup{F}]}\\
	&= \frac{\tr[\pi_\textup{\tb  ok}^{\textup{\tb W}} \pi_{1/2}^\textup{F} \rho_\textup{Fi} \pi_{1/2}^\textup{F}]}{\sum_{z\in\{-1/2,1/2\}} \tr[\pi_\textup{\tb  ok}^{\textup{\tb W}} \pi_{z}^\textup{F} \rho_\textup{Fi} \pi_{z}^\textup{F}]}=1,\label{eq:prob corresponding to statement W 22 (0 1)}
\end{split}
\end{align}
where the last line follows from observing that in~\cref{fig:FR paper experiment} we have that  $ \tr[\pi_\textup{\tb  ok}^{\textup{\tb W}} \pi_{-1/2}^\textup{F} \rho_\textup{Fi} \pi_{-1/2}^\textup{F}]=0$ and thus the r.h.s. of above is one, in accordance with what the authors claim. We thus have $\textup{\tb  W}^{n:22}_{(0,1)}=\mathtt{correct}$.

While they are taking into account F's measurement while reasoning, they are not taking into account \tb F's measurement when reasoning. We can check that the statement $\textup{\tb  W}^{n:22}_{(1,1)}=\mathtt{false}$ since the following probability, which takes into account both \tb F's and F's measurements, is strictly less than one:
\begin{align}\label{eq:second prob diff from 1st}
\begin{split}
	&P\big(z=1/2|\mb w=\textup{\tb {ok}},(x_1,x_2)=(1,1)\big)\\
	&= \sum_{\substack{ w\in\{\textup{ok},\textup{fail}\} \\ r\in\{\textup{heads},\textup{tails}\}}} \!\!\! P\big(z=1/2,w,r|\mb w=\textup{\tb {ok}},(x_1,x_2)=(1,1)\big)\\
	&= \frac{\sum_{\substack{ w\in\{\textup{ok},\textup{fail}\} \\ r\in\{\textup{heads},\textup{tails}\}}}\tr[\pi_w^{\textup{W}}\pi_\textup{\tb  ok}^{\textup{\tb W}} \pi_{1/2}^\textup{F} \pi_{r}^{\textup{\tb  F}} \rho_\textup{Fi} \pi_{r}^{\textup{\tb  F}} \pi_{1/2}^\textup{F}]}{\sum_{\substack{z\in\{-1/2,1/2\}\\ w\in\{\textup{ok},\textup{fail}\}\\ r\in\{\textup{heads},\textup{tails}\}}} \!\!\! \tr[\pi_{w}^{\textup{W}}\pi_\textup{\tb  ok}^{\textup{\tb W}} \pi_{z}^\textup{F} \pi_{r}^{\textup{\tb  F}} \rho_\textup{Fi} \pi_{r}^{\textup{\tb  F}} \pi_{z}^\textup{F}]}\\
	&=\frac{\sum_{ r\in\{\textup{heads},\textup{tails}\}}\tr[\pi_\textup{\tb  ok}^{\textup{\tb W}} \pi_{1/2}^\textup{F} \pi_{r}^{\textup{\tb  F}} \rho_\textup{Fi} \pi_{r}^{\textup{\tb  F}} \pi_{1/2}^\textup{F}]}{\sum_{\substack{z\in\{-1/2,1/2\}\\ r\in\{\textup{heads},\textup{tails}\}}} \tr[\pi_\textup{\tb  ok}^{\textup{\tb W}} \pi_{z}^\textup{F} \pi_{r}^{\textup{\tb  F}} \rho_\textup{Fi} \pi_{r}^{\textup{\tb  F}} \pi_{z}^\textup{F}]}=\frac13<1,
\end{split}
\end{align}

Furthermore $\textup{\tb  W}^{n:22}_{(1,0)}=\emptyset$ since the statement $\textup{\tb  W}^{n:22}$ is about F's measurement outcome. Before we move on, observe that \cref{eq:prob corresponding to statement W 22 (0 1),eq:second prob diff from 1st} take on different values and provides another example of collider bias. In particular, these two equations show that the probability of $\textup{F}$'s outcome does depend on $\textup{\tb  F}$'s setting $x_1$ given the knowledge that the post-selection on  $\mb w$ succeeded. This means that the probability $P\big(z=1/2|\mb w=\textup{\tb {ok}}\big)$ that FR consider, is not well-defined when the setting or the prior over the settings is not specified.

The authors then go on to claim \emph{...because agent \tb  W announces $\mb w$, agent W can be certain about \tb W's knowledge, which justifies statement $\textup{W}^{n:26}$, where $\textup{W}^{n:26}$ is (assuming \tb W announces $\mb w$ = \tb {ok} at time $n\!:\!21$.) ``I am certain that \tb W knows that $\mb w =$ \tb {ok} at time $n \!:\!21$.''}. We agree that W can be sure of any correct announcement made by \tb  W, irrespective of their choice of $(x_1,x_2)$ settings. Thus $\textup{W}^{n:26}_{(0,1)}=\textup{W}^{n:26}_{(1,0)}=\textup{W}^{n:26}_{(1,1)}=\mathtt{correct}$.

Next the authors make the claim \emph{...according to quantum mechanics, agent W can be certain that the outcome ($\mb{w},w$)= (\tb {ok},ok) occurs after finitely many rounds}. They make this based on the fact that
\begin{align} \bra{\psi_{\textup{Fi}}}\pi_{(\mb{w},w)=(\textup{\tb {ok}},\textup{ok})}\ket{\psi_{\textup{Fi}}}=\frac{1}{12},
\end{align}
where $\ket{\psi_{\textup{Fi}}}$ is the final state given on the r.h.s. of \cref{eq:final state FR version}, $\pi_{(\mb{w},w)=(\textup{\tb {ok}},\textup{ok})}=\ket{\textup{\tb {ok}}}
\bra{\textup{\tb {ok}}}_{\Lbar}\otimes
\ket{\textup{ok}}\bra{\textup{ok}}_{\Lab}$. Note that this can also be seen as the case with $(x_1,x_2)=(0,0)$, as we will see later in \cref{eq: prob ok ok 0 0}. 
The authors then claim \emph{$\textup{W}^{n:00}$: ``I am certain that there exists a round $n$ in which the halting condition at time $n\!:\!40$ is satisfied.''}  Noting that 
\begin{align}
\begin{split}
	&P\big(\mb{w}=\mb{\textup{ok}},w=\textup{ok}| (x_1,x_2)=(1,1)\big)\\
	&=\sum_{\substack{z\in\{-1/2,1/2\}\\r\in\{\textup{heads},\textup{tails}\}}}\tr[\pi_\textup{ok}^{\textup{W}}\pi_\textup{\tb  ok}^{\textup{\tb W}} \pi_{z}^\textup{F} \pi_{r}^{\textup{\tb  F}} \rho_\textup{Fi} \pi_{r}^{\textup{\tb  F}} \pi_{z}^\textup{F}]\\
	&=\!\!\!\sum_{\substack{z\in\{-1/2,1/2\}\\r\in\{\textup{heads},\textup{tails}\}}}\!\!\!\!\!\!\!\!\!\!\!\!\big|\! \bra{\textup{\tb {ok}}}\!\bra{\textup{ok}} \pi_z^{\textup{F}}U_{F} U_{Spin} \pi_r^{\textup{\tb  F}} U_{\mb{F}} \ket{\textup{init},\phi_0,S_0}_{\Lbar\Ss}\ket{\phi_0}_\Lab\! \big|=\frac56\label{eq: prob ok ok 1 1}
\end{split}
\end{align}
is positive, we conclude that $\textup{W}^{n:00}_{(1,1)}=\mathtt{correct}$. Similarly, we observe that 
\begin{align}
\begin{split}
	P\big(\mb{w}=\mb{\textup{ok}},w=\textup{ok}|(x_1,x_2)=(0,1)\big)=\sum_{z\in\{-1/2,1/2\}}\tr[\pi_\textup{ok}^{\textup{W}}\pi_\textup{\tb {ok}}^{\textup{\tb W}} \pi_{z}^\textup{F} \rho_\textup{Fi}  \pi_{z}^\textup{F}]=\frac12\label{eq: prob ok ok 0 1}
\end{split}
\end{align}
in positive, thus $\textup{W}^{n:00}_{(0,1)}=\mathtt{correct}$.
Likewise, there exits $r\in \{\textup{tails}, \textup{heads}\}$ such that
\begin{align}
\begin{split}
	P\big(\mb{w}=\mb{\textup{ok}},w=\textup{ok}|(x_1,x_2)=(1,0)\big)=\sum_{r\in\{\textup{heads},\textup{tails}\}}\tr[\pi_\textup{ok}^{\textup{W}}\pi_\textup{\tb  {ok}}^{\textup{\tb W}}  \pi_{r}^{\textup{\tb  F}} \rho_\textup{Fi} \pi_{r}^{\textup{\tb  F}} ] =\frac12\label{eq: prob ok ok 1 0}
\end{split}
\end{align}
is positive thus $\textup{W}^{n:00}_{(1,0)}=\mathtt{correct}$.
Finally, 
\begin{align}
\begin{split}
		P\big(\mb{w}=\mb{\textup{ok}},w=\textup{ok}|(x_1,x_2)=(0,0)\big)=\tr[\pi_\textup{ok}^{\textup{W}}\pi_\textup{\tb  {ok}}^{\textup{\tb W}}  \rho_\textup{Fi}  ]=\frac1{12},\label{eq: prob ok ok 0 0}
\end{split}
\end{align}
hence $\textup{\tb  W}^{n:00}_{(0,0)}=\mathtt{correct}$. Therefore we conclude from \cref{eq: prob ok ok 1 1,eq: prob ok ok 0 1,eq: prob ok ok 1 0,eq: prob ok ok 0 0} that while the statement $\textup{\tb  W}^{n:00}_{(x_1,x_2)}$ is correct for all settings, the value of the corresponding probability which FR calculate is not correct for all settings.

Next, the authors state \emph{Agent F may insert agent's \tb  F's statement $\textup{\tb  F}^{n:02}$ into $\textup{F}^{n:12}$, obtaining statement $\textup{F}^{n:13}$}. As we pointed out above, said statement only holds under certain settings. The authors then go on to say \emph{By virtue of [assumption] $\textup{C}$, she [agent F] may then conclude that statement $\textup{F}^{n:14}$ holds, too.} This statement is \emph{$\textup{F}^{n:14}$: ``I am certain that W will observe $w$= \textup{fail} at time $n\!:\!31$''.} This statement is merely the same as $\textup{F}^{n:13}$ but now with the difference that F \emph{herself} is certain that W will observe $w$= \textup{fail} at time $n\!:\!31$ rather that F being merely certain that \tb F is certain. There are no settings F can select for which this statement is true since
\begin{align}
P\big(w=\textup{fail}|z=1/2, (x_1,x_2)=(x_1,1)\big)<1
\end{align}
for all settings $x_1\in\{0,1\}$. Moreover, F cannot inherit the statement ``I am certain that W will observe $w$= fail at time $n\!:\!31$'' via the transfer of knowledge operator  since this statement does not make reference to agent \tb F and thus it would require F \emph{himself} to use settings $(x_1,x_2)=(1,0)$ just after obtaining the measurement outcome $z=\frac12$; yet the latter requires setting $x_2=1$ (This is distinct to the case of $\textup{F}^{n:13}$).

The remaining statements made by W and \tb  W are derived from the previous statements under the assumptions $\textup{Q}$, $\textup{U}$, $\textup{C}$, $\textup{D}$, and $\textup{S}$. However, they do so disregarding the setting parameters and hence reaching their apparent contradiction.

In \cref{table: resolution_prepmeas}, we compare the original statements of the FR paper used in deriving the apparent contradiction and the explicit version of those statements obtained within our framework. This is analogous to \cref{table: resolution_ent} of the entanglement case and it can be immediately seen that while the original statements (which ignore setting information) can be combined to yield the apparent contradiction, the explicit statements (which specify the setting choice) cannot be combined in this manner even using the standard rules of classical logic. These explicit statements can also be derived directly within the augmented circuit of the EWFS at hand, which is the prepare and measure version of the FR scenario. This augmented circuit is illustrated in \cref{fig: FR_PM_aug_circuit1} and an equivalent version of this circuit (that makes the mapping to the entanglement formulation of the scenario more evident) is given in \cref{fig: FR_PM_aug_circuit2}.

In summary, in this section we have shown that each of the individual statements that FR use to derive a contradiction (see \cref{table: resolution_prepmeas}) hold under some choice of settings in our framework (note this also follows from \Cref{theorem: main} of the general framework). However, each of these statements requires a different setting choice and can no longer be combined using the FR assumptions to yield a paradox.

\begin{table*}[t]
\centering
 \begin{tabularx}{\linewidth}{|| c | L |  L | L | L | L ||} 
 \hline
Agent & Assumed observation  & Statement inferred via $\textup{Q}$ & Original implication obtained from $\textup{Q}$ & Additional implicit assumption & Explicit implication obtained from $\textup{Q}$ \\ [0.5ex] 
 \hline\hline
 $\textup{\tb{F}}$ & $r=\textup{tails}$ at $n:01$&   $\mathbf{\mb{F}^{n:02}}$: I am certain that W will observe $w=\textup{fail}$ at $n:31$ & $K_{\mb{F}}\big(r=\textup{tails} \Rightarrow w=\textup{fail} \big)$&$\textup{\tb{F}}$'s outcome is reasoned about and F's lab is modelled as a closed quantum system & $K_{\mb{F}}\big((x,y)=(1,0) \land r=\textup{tails} \Rightarrow w=\textup{fail} \big)$ \\ 
 \hline
 F & $z= +\frac{1}{2}$ at $n:11$ &  $\mathbf{F^{n:12}}$: I am certain that \tb{F} knows that $r=\textup{tails}$ at $n:01$ &$K_F\big(z=+\frac{1}{2} \Rightarrow K_{\mb{F}}(r=\textup{tails}) \big)$ & \tb{F}'s and F's outcomes are both reasoned about & $K_F\big((x,y)=(1,1) \land z=+\frac{1}{2} \Rightarrow K_{\mb{F}}(r=\textup{tails}) \big)$\\ \hline ${}^{}$
\tb{W} & $\mb{w}=\mb{\textup{ok}}$ at $n:21$ & $\mathbf{\mb{W}^{n:22}}$: I am certain that F knows $z=+\frac{1}{2}$ at $n:11$& $K_{\mb{W}}\big(\mb{w}=\mb{\textup{ok}} \Rightarrow K_F(z=+\frac{1}{2}) \big)$& F's outcome is reasoned about and \tb{F}'s lab is modelled as a closed quantum system & $K_{\mb{W}}\big((x,y)=(0,1) \land \mb{w}=\mb{\textup{ok}} \Rightarrow K_F(z=+\frac{1}{2}) \big)$\\ \hline
W & Announcement by \tb{W} that $\mb{w}=\mb{\textup{ok}}$ at $n:21$ &   $\mathbf{W^{n:26}}$: I am certain that \tb{W} knows that $\mb{w}=\mb{\textup{ok}}$ at $n:21$ & $K_WK_{\mb{W}}\big(\mb{w}=\mb{\textup{ok}}\big )$& (none) &$K_WK_{\mb{W}}\big(\mb{w}=\mb{\textup{ok}}\big )$\\ [1ex] 
 \hline
 \end{tabularx}
 \caption{Table 3 of \cite{Frauchiger2018} with the additional implicit assumptions needed to make each statement. Without the implicit assumptions, the implications drawn from $\textup{Q}$ in the original FR paper are summarised in the modal logic language in column 4. The corresponding explicit version of the same implications that take the additional implicit assumptions into account are given in the last column. While original implications can be combined using $\textup{C}$ and the distributive axiom to yield the further implications listed in Table 3 of the FR paper, the explicit version of these implications cannot be combined even using the standard rules of classical logic such as $\textup{C}$ and the distributive axiom. We therefore see that the paradox never arises even when using $\textup{Q}$, $\textup{C}$ and $\textup{S}$ and modelling agents unitarily as long as we are careful to use the explicit version of the implications. The variables $x$ and $y$ appearing in the last column are the settings appearing in the circuit of \cref{fig: FR_PM_aug_circuit1}, these encode the additional implicit assumptions.}
 \label{table: resolution_prepmeas}
\end{table*}

\subsection{Setting-dependence in FR's scenario}
\label{appendix: setting_dep_FR}

It \Cref{appendix: ent_FR}, we mapped each FR probability ( $P_{conv}(b = 1|u = w =
\textup{ok}) = 1$, $P_{conv} (a = 1|b = 1) = 1$ and $P_{conv} (w = \textup{fail}|a = 1) = 1$) to a unique setting choice $(x_1,x_2)\in\{(1,0),(1,1),(0,1)\}$ that reproduces the same probability in our framework. It is illustrative to go beyond this analysis and consider other possible setting choices we can assign to each statement. Then we find that in contrast to the above-mentioned settings which do reproduce the exact  probabilities of the FR paper, alternate setting choices not longer give the same probabilities, and therefore do not yield the desired logical statements needed in FR's reasoning. For instance, we showed in \Cref{appendix: ent_FR} that $P_{conv}(b = 1|u = w =\textup{ok}) =P(b = 1|u = w =\textup{ok},(x_1,x_2)=(0,1)) = 1$. We can alternatively consider $P(b = 1|u = w =\textup{ok},(x_1,x_2)=(1,1))$. Note that we cannot consider any setting choice with $x_2=0$ here because the prediction involves Bob's outcome $b=1$ and therefore must involve a non-trivial projector on Bob's side in its evaluation.

\begin{align}
	\begin{split}
		&P(b = 1,u = w =\textup{ok}|(x_1,x_2)=(1,1)) = \frac{1}{12}\\
		=& \sum_{a\in\{0,1\}}|\bra{\textup{ok}}_{\R\A}\otimes \bra{\textup{ok}}_{\Ss\B}\Big(\pi^{\A}_{x_1=1,a}\otimes\pi^{\B}_{x_2=1,b=1}\Big)\ket{\psi^{t=2}}|^2\\     =& \sum_{a\in\{0,1\}}|\bra{\textup{ok}}_{\R\A}\!\otimes \bra{\textup{ok}}_{\Ss\B}\!\Big(\ket{aa}\!\bra{aa}_{\R\A}\otimes\ket{11}\bra{11}_{\Ss\B}\Big)\!\ket{\psi^{t=2}}|^2
	\end{split}
\end{align}

\begin{align}
	\begin{split}
		&P(u = w =\textup{ok}|(x_1,x_2)=(1,1)) = \frac{1}{4}\\
		=& \sum_{a,b\in\{0,1\}}|\bra{\textup{ok}}_{\R\A}\otimes \bra{\textup{ok}}_{\Ss\B}\Big(\pi^{\A}_{x_1=1,a}\otimes\pi^{\B}_{x_2=1,b}\Big)\ket{\psi^{t=2}}|^2\\     =& \sum_{a.b\in\{0,1\}}\!\!\!|\bra{\textup{ok}}_{\R\A}\!\otimes \!\bra{\textup{ok}}_{\Ss\B}\!\Big(\ket{aa}\bra{aa}_{\R\A}\otimes\ket{bb}\bra{bb}_{\Ss\B}\Big)\!\ket{\psi^{t=2}}|^2
	\end{split}
\end{align}

Putting this together using the rule for conditional probabilities, we have
\begin{equation}
	P(b = 1|u = w =\textup{ok},(x_1,x_2)=(1,1))=\frac{1}{3}\neq 1    
\end{equation}

The FR probability $P_{conv} (a = 1|b = 1) = 1$ corresponds to $P (a = 1|b = 1, (x_1,x_2)=(1,1)) = 1$ in our framework, but no alternative setting choices are possible here as the prediction refers to both Alice and Bob's classical outcomes. Finally, the FR probability $P_{conv} (w = \textup{fail}|a = 1) = 1$ corresponds to $P (w = \textup{fail}|a = 1,(x_1,x_2)=(1,0)) = 1$ in our framework and we consider the alternative setting choice $(x_1,x_2)=(1,1)$ in this case (again $x_1=0$ is not a possible setting choice here as the prediction involves Alice's classical outcome). The corresponding prediction $P (w = \textup{fail}|a = 1,(x_1,x_2)=(1,1))$ is calculated below, and is also not equal to 1. 

\begin{align}
	\begin{split}
		&P(a = 1,w =\textup{fail}|(x_1,x_2)=(1,1)) = \frac{1}{3}\\
		=& \sum_{b\in\{0,1\}}|\id_{\R\A}\otimes \bra{\textup{fail}}_{\Ss\B}\Big(\pi^{\A}_{x_1=1,a=1}\otimes\pi^{\B}_{x_2=1,b}\Big)\ket{\psi^{t=2}}|^2\\     =& \sum_{b\in\{0,1\}}|\id_{\R\A}\otimes \bra{\textup{fail}}_{\Ss\B}\Big(\ket{11}\bra{11}_{\R\A}\otimes\ket{bb}\bra{bb}_{\Ss\B}\Big)\ket{\psi^{t=2}}|^2
	\end{split}
\end{align}

\begin{align}
	\begin{split}
		&P(a = 1|(x_1,x_2)=(1,1)) = \frac{2}{3}\\
		=&\sum_{b\in\{0,1\}}|\Big(\pi^{\A}_{x_1=1,a=1}\otimes\pi^{\B}_{x_2=1,b}\Big)\ket{\psi^{t=2}}|^2\\     =& \sum_{b\in\{0,1\}}|\Big(\ket{11}\bra{11}_{\R\A}\otimes\ket{bb}\bra{bb}_{\Ss\B}\Big)\ket{\psi^{t=2}}|^2
	\end{split}
\end{align}

This gives us 
\begin{equation}
	P (w = \textup{fail}|a = 1,(x_1,x_2)=(1,1))=\frac{1}{2}\neq 1.  
\end{equation}

In summary, we have found that $P(b = 1|u = w =\textup{ok},(x_1,x_2)=(0,1))\neq P(b = 1|u = w =\textup{ok},(x_1,x_2)=(1,1))$ and $P (w = \textup{fail}|a = 1,(x_1,x_2)=(1,0))\neq P (w = \textup{fail}|a = 1,(x_1,x_2)=(1,1))$. That is the predictions do depend on the setting. This highlights that a core reason for the FR paradox in this scenario is that the FR reasoning involves ignoring the conditioning on setting values in a scenario where the predictions do depend on these value. This illustrates our general result of \Cref{corollary: QUCDSassumptions}, which shows that in any EWFS, any potential inconsistencies can only arise in this manner: ignoring the conditioning on settings, in predictions which are setting-dependent.

Reformulating this insight in a more framework-independent manner, this means that the choice of Heisenberg cuts under which we evaluate the predictions for different sets of outcomes does matter, the probabilities do depend on this choice (captured by the settings in our framework). 
Our framework which explicitly takes into account the Heisenberg cut, provides a natural resolution to any such paradox without forsaking unitary quantum theory or classical logic.

\section{Classical example reproducing certain features of the FR correlations}
\label{appendix: classical_example}

We noted earlier that the FR correlations are an example of collider bias, or in this case, they exhibit signalling under post-selection. Further, we have also seen that the assumption \hyperref[def: I_assump]{$\mathbf{I}$} is necessary for reproducing the apparent paradox. In this section, we provide an example of a classical protocol that reproduces these features of the FR correlations. We note however that the settings in our example do not have the same physical meaning as the settings of the FR augmented circuit (which relate to the choice of Heisenberg cut). Nevertheless, the example gives an intuition for our resolution. We will thus refer to the \hyperref[def: I_assump]{$\mathbf{I}$} applied to the settings of this classical example as $\mathbf{I}_{C}$ to distinguish it from the version of the assumption applied to the settings of our augmented circuit that model Heisenberg cuts.

The augmented circuit \cref{fig: FR_ent_circuit_aug} of the entanglement version of the FR protocol can be represented as shown in \cref{fig:classical_circuit}(a).
Consider now the circuit of \cref{fig:classical_circuit}(b), which has the same configuration as the circuit of \cref{fig:classical_circuit}(a) for the FR protocol and consider the following classical operations in place of the boxes $\matholdcal{A}$, $\matholdcal{B}$, $\matholdcal{U}$ and $\matholdcal{W}$. The initial state is simply a uniformly distributed binary random variable $\Lambda$. The operation $\matholdcal{A}$ of Alice takes the given value $\lambda\in \{0,1\}$ of $\Lambda$ together with the binary setting $x_1$, and generates the outcome $a\in \{0,1\}$ by taking the XOR of the two $a=\lambda\oplus x_1$, and forwards $a$ to Ursula's operation $\matholdcal{U}$. $\matholdcal{U}$ takes $a$ from Alice's operation, internally generates a uniformly distributed bit $k_U\in\{0,1\}$ and outputs $u=a\cdot k_U$ where $\cdot$ denotes the logical AND. Bob's side is identical, his operation $\matholdcal{B}$ takes a value of $\Lambda$ along with $x_2$ and outputs $b=\lambda\oplus x_2$ while forwarding a copy of the classical bit $b$ to Wigner's operation $\matholdcal{W}$. $\matholdcal{W}$ takes $b$ from Bob, generates a uniform bit $k_W$ internally and outputs $w=b \cdot k_W$.

Given this simple model, it is easy to verify that $P(a,b|x_1,x_2)$ is non-signalling i.e., $P(a|x_1,x_2)=P(a|x_1)$ and $P(b|x_1,x_2)=P(b|x_2)$. However, $P(a,b|x_1,x_2,u)$ allows signalling from Alice to Bob i.e., $P(b|x_1,x_2,u)\neq P(b|x_2,u)$ and $P(a,b|x_1,x_2,w)$ allows signalling from Bob to Alice i.e., $P(a|x_1,x_2,w)\neq P(a|x_1,w)$. This is because, for instance, conditioned on the knowledge that $u=1$ (say), we know that since $u=a\cdot k_U$, $a=k_U=1$. Further, since $a=\lambda\oplus x_1=1$ and $b=\lambda\oplus x_2$, we know that $b=x_1\oplus x_2\oplus 1$ and Bob's outcome $b$ clearly depends on Alice's setting $x_1$ now. This means that when reasoning given the knowledge about some post-selections on $u$ and $w$, even in classical probability theory, the agents must be careful to account for any new correlations that this may introduce between their settings, else they may end up making wrong conclusions.

Interestingly, in \cite{Healey2018} Healey pointed out that whether or not Bob's state is collapsed influences the probabilities of Alice's outcome in the FR scenario due to the post-selection. Denying this dependence is what \cite{Healey2018} refers to as the assumption of \emph{intervention insensitivity}. Healey argues that FR implicitly assume intervention insensitivity in order obtain the logical contradiction. While Healey's intervention sensitivity may seem similar in spirit to the setting-dependence in our framework, \cite{Healey2018} appears to regard intervention sensitivity as a special non-local feature of quantum theory. 
In the FR case, the setting-dependence takes the particular form of enabling signalling under post-selection that we have explained here. However, the present classical example shows that this is not a special feature of the quantum correlations of the FR scenario but rather a common feature that arises in classical probability theory and causal inference \cite{Pearl2009} where it often goes by the name of \emph{collider bias}.

In the FR protocol, this feature is not immediately apparent as there are no settings in the actual protocol. However, as we have argued in the main text, these settings are implicit in the reasoning steps, and once made explicit, we see that there is indeed signalling under post-selection also in the FR scenario.  
Note that our general resolution and results do not depend on this aspect of signalling under post-selection or collider bias, although this is a feature of FR correlations. The general results show that imposing the \hyperref[def: I_assump]{$\mathbf{I}$} assumption in a scenario where it fails is necessary for recovering FR type apparent apparent paradoxes (independently of this feature), and we have seen that \hyperref[def: I_assump]{$\mathbf{I}$} does fail in the FR scenario.

This is also the case in our classical example, we can show that imposing $\mathbf{I}_C$ also leads to an apparent paradox here, as the predictions do depend on the settings (the scenario violates $\mathbf{I}_C$). For instance, take $u=1$ and $x_2=1$, then we have $a=k_u=1$ and hence $\Lambda=x_1\oplus 1$, which gives $b=x_1\oplus x_2\oplus 1=x_1$. Therefore we have $P(b=0|x_1=0,u=1,x_2=1)=1$ and $P(b=1|x_1=1,u=1,x_2=1)=1$. If we ignore the conditioning on the settings, as allowed by $\mathbf{I}_C$, we apparently have the paradoxical probability assignments $P(b=0|u=1)=1$ and $P(b=1|u=1)=1$ that correspond to $b$ deterministically being $0$ and $1$ (both with certainty) whenever $u$ is 1, which is an apparent violation of $\textup{S}$ as in the FR scenario. However we see that there is no real paradox once we correctly take into account all the settings and account for the correlations of the outcomes with the settings.

While the property of signalling under postselection holds in this simple classical circuit as it does in the quantum FR scenario, the exact correlations of the FR scenario are Bell non-local (or more generally, contextual) and cannot be reproduced using classical resources alone, once we fix the above circuit configuration where there is no information exchange between Alice and Bob's sides (except the shared initial state). In order to reproduce the FR correlations classically, one must necessarily modify the circuit configuration to allow for additional connections such as allow the settings $x_1$ and $x_2$ to depend on the state preparation, or allow for causal connections between the $\matholdcal{A}$, $\matholdcal{U}$ and $\matholdcal{B}$, $\matholdcal{W}$ operations, analogous to superdeterministic or non-local hidden variable explanations of Bell correlations (see for instance \cite{Wood2015}).
The correspondence between the FR argument and Hardy's logical argument for contextuality of quantum correlations is discussed in the next section.

\begin{figure*}[t!]
    \centering
   \begin{tikzpicture}[scale=0.75, transform shape]
\draw (0,0) rectangle node{\Large{$\matholdcal{A}$}} (1.5,1.5);
\draw[thick, arrows={-stealth}] (0.75,-1.5)--(0.75,0);
\draw[thick, arrows={-stealth}] (0.75,1.5)--(0.75,3);
\draw[thick, arrows={-stealth}] (0.2,1.5) to[out=90,in=-20] (-0.3,2);
\draw[thick, arrows={-stealth}] (-0.3,-0.5) to[out=20,in=-90] (0.2,0);
\node[align=center] at (0.5,-1) {$R$};
\node[align=center] at (0.3,2.5) {$RA$};
\node[align=center] at (-0.6,2) {$a$};
\node[align=center] at (-0.6,-0.5) {$x_1$};

\begin{scope}[shift={(0,3)}]
\draw (0,0) rectangle node{\Large{$\matholdcal{U}$}} (1.5,1.5);
\draw[thick, arrows={-stealth}] (0.2,1.5) to[out=90,in=-20] (-0.3,2);
\node[align=center] at (-0.6,2) {$u$};
\end{scope}

\begin{scope}[shift={(4,0)}]
\draw (0,0) rectangle node{\Large{$\matholdcal{B}$}} (1.5,1.5);
\draw[thick, arrows={-stealth}] (0.75,-1.5)--(0.75,0);
\draw[thick, arrows={-stealth}] (0.75,1.5)--(0.75,3);
\draw[thick, arrows={-stealth}] (1.3,1.5) to[out=90,in=200] (1.8,2);
\draw[thick, arrows={-stealth}] (1.8,-0.5) to[out=160,in=-90] (1.3,0);
\node[align=center] at (1,-1) {$S$};
\node[align=center] at (1.2,2.5) {$SB$};
\node[align=center] at (2.1,2) {$b$};
\node[align=center] at (2.1,-0.5) {$x_2$};

\begin{scope}[shift={(0,3)}]
\draw (0,0) rectangle node{\Large{$\matholdcal{W}$}} (1.5,1.5);
\draw[thick, arrows={-stealth}] (1.3,1.5) to[out=90,in=200] (1.8,2);
\node[align=center] at (2.1,2) {$w$};
\end{scope}

\end{scope}

\draw (0.2,-2.5) rectangle node{$\ket{\psi}_{RS}$} (5.3,-1.5);

\end{tikzpicture}\hspace{5em}
\begin{tikzpicture}[scale=0.75, transform shape]
\draw (0,0) rectangle node{\Large{$\matholdcal{A}$}} (1.5,1.5);
\draw[thick, arrows={-stealth}] (0.75,-1.5)--(0.75,0);
\draw[thick, arrows={-stealth}] (0.75,1.5)--(0.75,3);
\draw[thick, arrows={-stealth}] (0.2,1.5) to[out=90,in=-20] (-0.3,2);
\draw[thick, arrows={-stealth}] (-0.3,-0.5) to[out=20,in=-90] (0.2,0);
\node[align=center] at (-0.6,2) {$a$};
\node[align=center] at (-0.6,-0.5) {$x_1$};

\begin{scope}[shift={(0,3)}]
\draw (0,0) rectangle node{\Large{$\matholdcal{U}$}} (1.5,1.5);
\draw[thick, arrows={-stealth}] (0.2,1.5) to[out=90,in=-20] (-0.3,2);
\node[align=center] at (-0.6,2) {$u$};
\end{scope}


\begin{scope}[shift={(4,0)}]
\draw (0,0) rectangle node{\Large{$\matholdcal{B}$}} (1.5,1.5);
\draw[thick, arrows={-stealth}] (0.75,-1.5)--(0.75,0);
\draw[thick, arrows={-stealth}] (0.75,1.5)--(0.75,3);
\draw[thick, arrows={-stealth}] (1.3,1.5) to[out=90,in=200] (1.8,2);
\draw[thick, arrows={-stealth}] (1.8,-0.5) to[out=160,in=-90] (1.3,0);
\node[align=center] at (2.1,2) {$b$};
\node[align=center] at (2.1,-0.5) {$x_2$};

\begin{scope}[shift={(0,3)}]
\draw (0,0) rectangle node{\Large{$\matholdcal{W}$}} (1.5,1.5);
\draw[thick, arrows={-stealth}] (1.3,1.5) to[out=90,in=200] (1.8,2);
\node[align=center] at (2.1,2) {$w$};
\end{scope}

\end{scope}

\draw (0.2,-2.5) rectangle node{$\Lambda$} (5.3,-1.5);

\end{tikzpicture}
    \caption{Left: A concise representation of the augmented circuit (\cref{fig: FR_ent_circuit_aug}) corresponding to the entanglement version of the FR protocol. Here the boxes $\matholdcal{A}$ and $\matholdcal{B}$ correspond to Alice and Bob's measurements modelled as unitaries $M^{\A}_{unitary}$ and $M^{\B}_{unitary}$ followed by the setting-dependent projectors $\{\pi^A_{x_1}\}$ and $\{\pi^B_{x_2}\}$. The boxes $\matholdcal{U}$ and $\matholdcal{W}$ model to $\textup{ok}, \textup{fail}$ basis measurements of Ursula and Wigner on the joint systems $\R\A$ and $\Ss\B$ producing outcomes $u$, $w$.
    Right: A classical circuit with the same configuration as the left hand side circuit of the FR protocol. Here $\Lambda$, $a$, $b$, $u$, $w$, $x_1$ and $x_2$ are all associated with classical, binary variables and the operations $\matholdcal{A}$, $\matholdcal{B}$, $\matholdcal{U}$ and $\matholdcal{W}$ and initial distribution over $\Lambda$ are described in the text. Here as well, as in the FR scenario, $x_1$ and $b$ although initially independent (due to no-signalling) become correlated when conditioning on the future outcome $w$, and similarly $x_2$ and $a$ become correlated under post-selection on $u$. }
    \label{fig:classical_circuit}
\end{figure*}

\section{Relation to previous works: a more unified picture}
\label{appendix:relation_prev_works}

There are several previous works discussing and analyzing the FR apparent paradox, which can be categorised into three broad groups. The first category includes works (eg. \cite{Relano2018,Relano2020,Kastner2020,Zukowski2021,Biagio2021}) which provide fundamental reasons for rejecting one of the assumptions of the FR no-go theorem. This is a natural response to any no-go theorem, as it necessitates identifying which assumptions to discard.

The second category consists of papers such as \cite{Nurgalieva2018,Bub2019,Narasimhachar2020,Renes2021} which (possibly after identifying implicit assumptions beyond $\textup{Q}$, $\textup{C}$, and $\textup{S}$) conclude that there is indeed a paradox in the FR scenario. Some of these papers suggest that FR's result raises deeper problems for foundations or interpretations of quantum theory e.g., \cite{Bub2019} while others e.g., \cite{Renes2021} suggest possible resolutions by proposing additional reasoning rules.

The third category includes works such as \cite{ScottAronson,Araujo,Drezet2018,Healey2018,Sudbery_2017,Sudbery2019,Fortin2019,Losada2019} that question the correctness of the theorem. They argue that the reasoning in the FR paper includes implicit assumptions, one or more of which are invalid in quantum theory, and conclude that there is no real paradox between quantum theory and logic as FR claim \cite{Note16}.

These responses are often discussed independently and can be interpretation-specific. Here, we provide a more unified picture of these arguments within our framework, also discussing how different interpretations of quantum mechanics could apply our framework to resolve the apparent paradox in any EWFS.

\subsection{Previous works rejecting one of FR’s assumptions}

Here we consider previous arguments which fundamentally reject one of the assumptions $\textup{Q}$, $\textup{U}$, $\textup{C}$, $\textup{D}$, or $\textup{S}$ in the FR no-go theorem \cite{Frauchiger2018, Nurgalieva2018}. We discuss these cases in detail below, relating them to different interpretations of quantum mechanics and demonstrating how these arguments play out within our framework.

\textbf{Rejecting $\textup{Q}$: }  In FR's original paper, the authors interpret the violation of $\textup{Q}$ as a violation of unitary quantum theory. Following \cite{Nurgalieva2018}, we have separated these two assumptions $\textup{Q}$ and $\textup{U}$ explicitly, such that the $\textup{Q}$ assumption refers to the validity of the Born rule (or a weaker possibilistic version  thereof), which can be independently violated without giving up $\textup{U}$.

Rejecting $\textup{Q}$ means at least one prediction of the FR scenario (e.g., \cref{eq: FR_conv_pred} for the entanglement version) does not comply with the Born rule. Not all predictions can be simultaneously tested in a single experiment, leaving room for observable compliance with the Born rule without satisfying $\textup{Q}$. As noted in \cite{schmid2023review}, certain versions of Bohmian mechanics, a non-local hidden variable interpretation, violate $\textup{Q}$ in this manner, by deviating from the Born rule for inaccessible predictions.

In our framework, predictions in an EWFS are given by $P(\vec{a}_j|\vec{a}_l,k)$ (\Cref{definition: prediction}), where $\vec{a}_j$ and $\vec{a}_l$ are sets of measurement outcomes and $k$ are a set of parameters describing the scenario. In hidden variable interpretations, in addition to quantum states, channels, and measurements (\Cref{def:LWFS}), $k$ can include descriptions of the hidden variables. $P(\vec{a}_j|\vec{a}_l,k)$ may be computed using this information, not necessarily following the Born rule.

Such predictions can violate assumption \hyperref[def: Q_assump]{$\mathbf{Q}$} formalised in our work, and they would not correspond to the setting-conditioned predictions of the augmented circuit (\Cref{definition: setting_prediction}) which are derived using the Born rule. However, the settings are still meaningful in such interpretations, they provide different descriptions of the quantum measurement channel that the hidden variables must ``emulate''. It is to be noted that within such theories, settings cannot be interpreted as choices of Heisenberg cuts since classical theories do not have a non-trivial notion of such a cut. But such a fully classical theory could ``emulate'' a unitary measurement channel or setting $x_i=0$ through non-local hidden mechanisms.

Moreover, by applying the general premise of our reasoning rules discussed in \Cref{sec: reasoning_rules}, 
FR type paradoxes can also avoided in any EWFS within such interpretations as well. This would involve using settings to clearly specify the measurement channels being assumed, as well as potentially additional parameters in $k$ and conditioning on these variables in the reasoning (even when the probabilities do not arise through the Born rule).

\textbf{Rejecting $\textup{U}$: }  The $\textup{U}$ assumption, first noted in \cite{Nurgalieva2018}, involves rejecting the idea that agents' labs evolve as unitary closed quantum systems. Interpretations involving objective mechanisms for ``wavefunction collapse" reject this, where unitary quantum theory breaks down at certain macroscopic scales. Such interpretations also violate the \hyperref[def: U_assump]{$\mathbf{U}$} assumption in our work. In this view, one would assign setting $x_i=1$ for all measurements, assuming they involve systems more macroscopic than the scale at which the collapse occurs. Whether such models always avoid FR-type paradoxes depends on whether the predicted objective collapse scale is smaller than the scale at which a quantum system can be regarded as an ``agent”, which remains an open question.

In our formalisation of the $\textup{U}$ assumption as \hyperref[def: U_assump]{$\mathbf{U}$}, we have made explicit an additional aspect that is often implicit in previous works, relating to quantum control over labs of other agents.  Interpretations compatible with unitary quantum theory can still violate \hyperref[def: U_assump]{$\mathbf{U}$} by arguing that full quantum control over an agent's lab is impractical. For instance, in decoherence-based interpretations, measurements are always associated with an inaccessible environment that decoheres the quantum superposition. To regard each $\matholdcal{M}^{\A_i}$ in our EWFS definition (\Cref{def:LWFS}) as a ``measurement" according to this view point, one must restrict to standard quantum scenarios (\Cref{def:std_q_exp}), where one agent cannot have complete quantum control over the labs of other agents. As shown by \Cref{corollary: std_QT_setting_indep}, in these scenarios, a joint distribution for all (non-trivial) measurement outcomes can be assigned, which would be equivalent to the distribution obtained by choosing $x_i=1$ for all settings.

Notably, a recent work \cite{Zukowski2021} makes an argument for objective decoherence based on the concept of pre-measurements. The notion of a pre-measurement coincides with the case of setting 0 in our framework i.e., the modelling of the measurement purely as a unitary evolution. They argue that pre-measurements cannot produce outcomes consistently in quantum theory, and to produce an outcome one requires an irreversible evolution. This aspect is captured within our framework in the fact that that no non-trivial measurement outcome can be assigned to a measurement modelled with setting 0, non-trivial outcomes require setting 1 (applying the projection postulate). The authors of \cite{Zukowski2021} then suggest that this observation resolves the FR apparent paradox as FR's reasoning assigns outcome values to pre-measurements. However, \cite{Zukowski2021} imposes an objective distinction between pre-measurements and measurements, considering the operations of the agents Alice and Bob only as pre measurements (setting 0) and those of the super-agents Ursula and Wigner as measurements (setting 1). Moreover, the criterion for objectively fixing this choice in general EWFSs is not explicitly discussed in prior works. 

The resolution proposed here is significantly more general. It applies to all EWFSs, allowing for generally subjective setting assignments (given by an explicit rule) while consistently accounting for relational interpretations as well.

\textbf{Rejecting $\textup{C}$ or $\textup{D}$: } The assumption $\textup{C}$ allows agents to inherit each other's knowledge. FR suggested that relational approaches, such as relational quantum mechanics \cite{Rovelli1996} and QBism \cite{Caves2002, Fuchs2014}, reject this assumption by allowing different subjective perspectives. However, since the original $\textup{C}$ assumption left unclear the formal modeling of agents' knowledge and choices of Heisenberg cuts, we have proposed a formalisation of this assumption in our framework, as \hyperref[def: assump_C]{$\mathbf{C}$}. 

In subjective/relational approaches to EWFSs, agents can have different choices of Heisenberg cuts. Super-agents model agents as ``inside the cut," assuming unitary evolution of their labs, while agents model themselves as ``outside the cut," associating classical outcomes to their measurements. This is captured in our framework through settings, our reasoning rule (see \Cref{theorem: main} and \Cref{sec: reasoning_rules}) allows the choices of these settings to be subjective and agent-dependent in a precise manner.

Despite this relational aspect and allowing dynamic updates of agents' knowledge (based on new observations), our framework still allows agents to freely inherit each other's knowledge as it satisfies the assumption \hyperref[def: assump_C]{$\mathbf{C}$}. Even when agents use different settings, if assumptions about setting choices are explicitly stated, \hyperref[def: assump_C]{$\mathbf{C}$} and other logical rules, such as the distributive axiom \hyperref[def: assump_D]{$\mathbf{D}$} and transitivity, are always satisfied as shown in \Cref{corollary: QUCDSassumptions}. Ignoring these setting choices leads to apparent breakdowns of these axioms, but, as argued in \Cref{sec: reasoning_rules}, similar logical breakdowns can occur in classical scenarios if assumptions about channels used in reasoning are ignored.

There also exist proposals, such as \cite{Samuel2022}, which reject classical logical axioms in light of FR's work. Our results highlight that such rejection is unnecessary for preserving the consistency of quantum theory or relational interpretations.

\textbf{Rejecting $\textup{S}$: } Previous discussions \cite{Frauchiger2018, Nurgalieva2018, Nurgalieva2020} suggest that many-worlds interpretations \cite{Everett1957} would tend to reject $\textup{S}$ because an outcome $a_i$ could be 0 in one ``branch" of the wavefunction and 1 in another. Treating $a_i$ as a quantum object associated with a quantum state $\ket{a_i}$, the violation of $\textup{S}$ seems less paradoxical since the quantum state not being $\ket{0}$ does not imply it is $\ket{1}$, but could be another non-orthogonal state \cite{Sudbery2019}. Violation of $\textup{S}$ is paradoxical only when $a_i$ is regarded as a classical random variable.

In our framework, there is no ambiguity between classical and quantum objects. Even with subjective points of view, agents' measurement outcomes are classical variables, with the trivial value $a_i=\perp$ when the measurement is unitary. Therefore, rejecting \hyperref[def: assump_S]{$\mathbf{S}$} within our framework would indeed constitute a paradox, leading to an invalid probability distribution $P(a_i)$. However, we have shown this does not happen in our framework.

Thus, many-worlds interpretations would not reject \hyperref[def: assump_S]{$\mathbf{S}$} in our framework. We do not believe any previous proposal or interpretation of quantum theory are pathological enough to fail this assumption. Moreover, it is important to note that selecting setting $x_i=1$ when computing probabilities of $a_i$ does not conflict with many-worlds interpretations where the universe undergoes unitary evolution. The projectors applied for the $x_i=1$ case need not be interpreted as a physical projection or wavefunction collapse, but can be seen as a mathematical tool necessary for computing probabilities through the Born rule. Alternatively, this can also be understood akin to classical probability theory, as conditioning on an agents' knowledge of seeing a particular outcome $a_i=\aaa_i$ in a particular ``branch'' of a many-worlds wavefunction, even when other branches can exist, from the perspective of other agents. It would be interesting to extend such ideas to formally define a perspectival version of the many-worlds interpretation, from an operational approach.

\subsection{Previous suggestions for consistent reasoning}

Another category of previous works is those that propose ways around FR's apparent paradox, either through conceptual discussions or by suggesting additional reasoning rules. 

We have shown that if the FR scenario and its conventional quantum predictions are equivalently described using our augmented circuit, there is no paradox, eliminating the need for additional reasoning rules \cite{Note17}. Our formal version of each FR assumption is satisfied without contradictions. Nevertheless, it is insightful to interpret previous works' proposed reasoning rules.

{\bf Adding context labels} At a more pedagogical level, \cite{Nurgalieva2018} suggests an additional reasoning rule that could incorporated to avoid the paradox, which involves tagging all statements made by agents with certain ``context labels'' and only combining statements with matching labels. However no particular model for these labels or a rigorous formalisation of the rule was proposed there, and the question of efficiency and causal consistency (namely regarding the number of contexts to be checked at each step, and whether one needs to account for the contexts of future measurements) was left open. The settings of our framework can be interpreted as a concrete instantiation of these abstract context labels, but we have seen that no additional reasoning rules beyond quantum theory and classical logic are required once the settings are explicitly specified. 

Importantly, our framework does not impose that only statements derived from the same setting labels can be combined. In fact, we have shown that in certain situations the statements become independent of the certain setting labels and the choice of those setting labels no longer affects the validity of the statement (see \Cref{theorem: main,theorem: setting_independence}). We have also addressed the efficiency and causality issues.

{\bf Parsing rule for measurements} Another recent work \cite{Renes2021} proposes a parsing rule for quantum theory to determine when a unitary/isometry can be considered a measurement with a classical outcome. This rule deems an operation a measurement if no future non-commuting operations act on the memory system, ensuring consistency in the FR scenario by limiting agents' reasoning to such operations. This idea also aligns with the principle of superpositional solipsism in \cite{Narasimhachar2020}. Unlike our framework, which uses settings to describe measurements, \cite{Renes2021} maintains ambiguity in measurement modeling, but the additional parsing rule clarifies when the ambiguity can be safely ignored in agents’ reasoning. 

If an operation $\matholdcal{M}^{\A_i}$ in the EWFS description is deemed a measurement according to this parsing rule, then it would imply that there are no super-agents to the agent $\A_i$ in the scenario (also according to our definition of the non-super-agent structure). As shown in \Cref{theorem: standardQT}, $\A_i$'s setting can then be safely ignored in the predictions, allowing all agents to reason about $\A_i$'s outcome without inconsistencies in that scenario. However, we note that our non-super-agent structure is formalised without reference to commutativity of operations. In particular, a non-commuting operation acting solely on the memory $\M_i$ might fail the parsing rule of \cite{Renes2021} but would still allow $\A_i$'s setting to be safely ignored in our framework.

Moreover our framework does not restrict the ability to reason about $\A_i$'s outcome based on whether there are super-agents to $\A_i$ (i.e., someone ``Hadamarding'' $\A_i$'s brain) in the scenario. As we have seen, each of FR's statements can be reproduced in our formalism. This is in contrast to \cite{Renes2021} and another work \cite{Alexios2022} which suggest that certain statements of FR that refer to outcomes of agents' whose brains will later be Hadamarded, should not be allowed. While such additional rules or restrictions are sufficient to ensure logical consistency, the adherence to causality principles and efficiency of programming remain open question. Our work shows that such additional rules are not necessary for the purpose of ensuring consistency in any EWFS and that logical, causal consistency and efficiency of reasoning are possible with weaker restrictions on the reasoning, though it may be of interest to impose such rules based on other physical considerations or interpretations.

{\bf Consistent histories interpretation} Another notable approach is the consistent histories (CH) interpretation of quantum theory \cite{Griffiths1984, GriffithsCH}. 
In this approach, one specifies a set of possible histories for a given scenario, and a consistency criterion that tells us when a set of histories is consistent, allowing probabilities to be assigned to such consistent sets. The approach has beenn applied to explain a number of quantum paradoxes arising in standard (non-Wigner's Friend like) quantum scenarios, such as contextuality paradoxes, pre and post-selection paradoxes etc.  
In the context of Wigner's Friend scenarios, \cite{Losada2019} shows that the reasoning used in FR's derivation of the paradox requires computing probabilities in an inconsistent family of histories, and the authors then argue that such reasoning is therefore not valid in quantum theory.

At a high level, this bears resemblance to the general result of \Cref{corollary:setting_independence}, where we demonstrated that any EWF paradox is rooted in computing predictions across distinct setting choices and combining such statements while ignoring the setting labels (even though the predictions depend on these labels). However, our approach is distinct from the CH interpretation in key ways, and arguably offers a simpler and more minimal resolution to EWFS paradoxes. 

The CH approach requires considering commutation relations between families of projectors to determine consistency, while our settings are formalised without reference to commutation relations and do not involve such additional rules to ensure consistency. We have seen that conditioning on settings that model a measurement corresponds to conditioning on the choice of channel used in computing a probability, no additional rule is required to forbid combinations of statements made under different settings (as in needed in the CH approach to avoid combining inconsistent histories).

Moreover, our setting independence results, highlight that in certain cases, classical probability theory and classical logic ensure that even statements made under different settings can be consistently combined, as those statements are independent of the setting choice. While non-commuting projectors central to the CH approach are a non-classical aspect, we have discussed how our resolution of the paradox shares strong similarities with how analogous multi-agent inconsistencies are resolved in purely classical theories (see also the previous paragraph on parsing rules, for further discussion on the link between non-commutativity and our resolution).

Moreover, while CH's solution is sufficient to avoid logical inconsistencies, it does not provide an explicit reasoning rule for how to select the set of histories to be used when reasoning about predictions or agents' knowledge in an EWFS \cite{Nurgalieva2020}. 
In particular, \cite{Nurgalieva2020} noted that, if unitary quantum theory were universally valid, the perspectives and predictions of different agents would indeed correspond to different histories and even in a single experiment such as the FR scenario, there is no single objective history of events that is realised.

A natural question that arises is whether there is a unified framework with a concrete set of rules to construct it, where all these perspectives and predictions of an EWFS can be consistently incorporated, while recovering the predictions of real-world quantum experiments performed so far. Here, we have demonstrated that all EWFSs in quantum theory can be completely described within a single consistent quantum circuit framework that is capable of resolving general EWF paradoxes. We also provided an explicit rule for selecting the settings (modelling the Heisenberg cuts) in accordance with universal validity of unitary quantum theory and without assuming the existence of objective notion of observed measurement events, the observations and predictions in our formalism are relative to a choice of such settings.

\subsection{Previous works discussing the validity of FR's claim}

A third category of papers are those which question the validity of FR's theorem, due to additional implicit assumptions (other than $\textup{Q}$, $\textup{U}$, $\textup{C}$, $\textup{D}$ and $\textup{S}$ discussed in \cite{Frauchiger2018, Nurgalieva2018}) which are violated \cite{ScottAronson, Healey2018, Araujo, Sudbery2019, Fortin2019, Losada2019}. Some specific examples include Scott Aaronson's blog post that refers to an additional ``unformalised'' assumption of FR, Healey's assumption of intervention insensitivity \cite{Healey2018}, and the assumption regarding collapse/no-collapse that Araujo points out \cite{Araujo}. To quote Aaronson,

\begin{quote}
 But I reject an assumption that Frauchiger and Renner never formalise. That assumption is, basically: ``it makes sense to chain together statements that involve superposed agents measuring each other's brains in different incompatible bases, as if the statements still referred to a world where these measurements weren't being done."
\end{quote}

In our work, we have concretely shown that an additional assumption \hyperref[def: I_assump]{$\mathbf{I}$} (setting-independence) is violated in the FR scenario, but must be necessarily imposed for reproducing the apparent paradox. We have seen that the assumptions \hyperref[def: Q_assump]{$\mathbf{Q}$}, \hyperref[def: U_assump]{$\mathbf{U}$}, \hyperref[def: assump_C]{$\mathbf{C}$}, \hyperref[def: assump_D]{$\mathbf{D}$} and \hyperref[def: assump_S]{$\mathbf{S}$} about the validity of quantum theory and classical logic are always consistent in any EWFS in our formalism.  In our understanding, the assumption \hyperref[def: I_assump]{$\mathbf{I}$} formally embodies the spirit of the additional implicit assumptions noted in the above examples of previous works. All these previous works argue that the respective assumption is necessary to reproduce the apparent paradox of FR, but suggest that the assumption fails in the FR scenario, due to the incompatible measurements, and/or (Bell) non-locality of the correlations involved.

While the assumption \hyperref[def: I_assump]{$\mathbf{I}$} of our framework is formulated in a much more general manner, and is a priori independent of quantum features such as incompatible measurements and non-locality, when applied to the FR scenario, it captures the features of these previously noted assumptions. Aaronson's assumption is reflected in our framework by noting that the four possible settings $(x_1,x_2)\in\{(0,0),(0,1),(1,0),(1,1)\}$ for the FR protocol are in one-to-one correspondence with the four possible measurement contexts (i.e., sets of compatible measurements) of a bipartite Bell experiment with binary choice of measurements on each side.

The settings also tell us whether or not we have ``collapsed'' the state of an agents' system and memory (or rest of the lab) by applying a projector corresponding to their outcome, and provide a way to formalise Araujo's assumption. Araujo also noted issues with FR's treatment of post-selection, which are accounted for in our analysis by explicitly computing probabilities conditioned on the post-selection and avoiding collider bias. 
We have also seen that in the FR scenario (\Cref{ssec: resolution_prep}), the prediction $P(w=\textup{fail}|r=\textup{tails},(x_1,x_2))$ does depend on the setting $x_2$ of $\textup{F}$. This is analogous to the property that Healey calls intervention sensitivity.

While our results which establish the general consistency of quantum theory are certainly contrary to FR's popular summary that quantum theory cannot consistently justify the use of itself, the validity of FR's claimed theorem depends on how the assumptions are interpreted. In \Cref{sec:The choices of Heisenberg cuts do matter in FR}, we discussed a refined interpretation of FR's theorem in which it is correct, this would be to say that a version of quantum theory that additionally allows Heisenberg cuts to be freely ignored leads to contradictions with classical logic. However, in this interpretation, the result is not as surprising as the popular summary claims it to be, although it has undoubtedly fuelled an intriguing research program on EWFSs in quantum foundations inspiring other no-go theorems (such as \cite{Brukner2018, Bong2020}) based on a similar set-up where the underlying  assumptions are formalised more rigorously.

\section{Proofs of all results}
\label{appendix: proofs}

\subsection{Proofs of results from the main text}
\ConsistencyBasic*
\begin{proof}
   Let $S\in \Sigma $ be an arbitrary statement, by construction this of the form: ``If the outcomes $\vec{a}_l$ take values $\vec{\aaa}_l$ and the additional parameters of the scenario take the value $k=\kkk$, then the outcomes $\vec{a}_l$ take values $\vec{\aaa}_l$ with a probability $P(\vec{a}_j=\vec{\aaa}_j|\vec{a}_l=\vec{\aaa}_l,k=\kkk)$.'' Then the negation of $S$ is ``If the outcomes $\vec{a}_l$ take values $\vec{\aaa}_l$ and the additional parameters of the scenario take the value $k=\kkk$, then the outcomes $\vec{a}_l$ {\bf do not} take values $\vec{\aaa}_l$ with a probability $P(\vec{a}_j=\vec{\aaa}_j|\vec{a}_l=\vec{\aaa}_l,k=\kkk)$.'' More precisely, this is a set of statements:\\

   $\neg S:=\{$ ``If the outcomes $\vec{a}_l$ take values $\vec{\aaa}_l$ and the additional parameters of the scenario take the value $k=\kkk$, then the outcomes $\vec{a}_l$ take values $\vec{\aaa}_l$ with a probability $P'(\vec{a}_j=\vec{\aaa}_j|\vec{a}_l=\vec{\aaa}_l,k=\kkk)$.''$\}_{P\neq P'}$. \\

It is then immediate from \Cref{def: consistency_pred} that if $S\in \Sigma$ then no $S'\in \neg S$ can be such that $S'\in \Sigma$. This completes the proof. 
\end{proof}

\SameChoicesSamePredictions*
\begin{proof}

The statement immediately follows from noting that fixing the setting choice for each measurement fully specifies the circuit of \cref{fig: genform_circuit} and therefore the joint state of all systems and memories $\Ss_1,\ldots\Ss_m,\M_1,\ldots,\M_N$ at each time-step. If $\A$ and $\B$ make the same choice of settings for all measurements, then they fully agree on the circuit (initial states and all channels) and therefore assign the same joint state to all systems at each time step.

Since all agents apply the Born rule to calculate the probabilities (\Cref{definition: setting_prediction}), and agree on the states and measurements for which the probabilities are calculated, they obtain the same probabilities and therefore make the same predictions for all measurements.

More explicitly, \Cref{appendix: prob_rule_general} shows that in any EWFS, given a choice of settings, every setting-conditioned prediction can be uniquely computed by applying the Born rule to our augmented circuit. Therefore if all agents in an EWFS use the augmented quantum circuit to reason, picking the same setting choice or prior distribution $P(\vec{x})$, then they agree on all the predictions made in that scenario. 
\end{proof}

\QUCDSassumptions*

\begin{proof}

The proof proceeds by showing that all 5 assumptions are satisfied by statements $\Sigma^{aug}$ derived in our augmented circuit framework, and the consistency of our framework shown in \Cref{theorem: main} guarantees that all 5 can be consistently applied to reason without any contradictions. 

\hyperref[def: Q_assump]{$\mathbf{Q}$}, \hyperref[def: U_assump]{$\mathbf{U}$} are by construction satisfied for all the statements $\Sigma^{aug}$, since these are associated with predictions computed using the Born rule and consider all possible setting choices (including the unitary modelling). Moreover, all our results apply to general EWFS (\Cref{def:LWFS}) which allows agents to have full quantum control over the labs of others, in the precise sense described in the formalisation of \hyperref[def: U_assump]{$\mathbf{U}$} .

\hyperref[def: assump_S]{$\mathbf{S}$} holds for all statements $\Sigma^{aug}$ in our framework because the statements are obtained from setting-conditioned predictions, which are well-defined normalised probabilities (see \Cref{appendix: prob_rule_general}) and the set of statements $\Sigma^{aug}$ is consistent as proven in \Cref{theorem: main}. No single valid normalised probability distribution can assign $P(\vec{a}_j=\vec{\aaa}_j)=1$ and $P(\vec{a}_j=\vec{\aaa}'_j)=1$ and the consistency result forbids the possibility of two distributions in the same scenario with $P(\vec{a}_j=\vec{\aaa}_j)=1$ and $P'(\vec{a}_j=\vec{\aaa}'_j)=1$. 

\hyperref[def: assump_C]{$\mathbf{C}$} holds for all statements in our framework because if an agent $\A_i$ knows that an agent $\A_j$ knows a statement $S\in \Sigma^{aug}$, then $\A_i$ can directly compute the setting-conditioned prediction $P(\vec{a}_j=\vec{\aaa}_j|\vec{a}_l=\vec{\aaa}_l,\vec{x}=\vec{\xxx})$ that defines the statement $S$ and thereby inherit the knowledge. Consistency as shown in  \Cref{theorem: main} guarantees that there is only one such probability assignment any user of our framework can arrive at for a given choice $\vec{x}=\vec{\xxx}$ and outcome values $\vec{a}_j=\vec{\aaa}_j$, $\vec{a}_l=\vec{\aaa}_l$ \cite{Note18}.

Finally the distributive axiom \hyperref[def: assump_D]{$\mathbf{D}$} is a rather basic axiom of logical inference, that it holds in our framework can be seen as follows. Firstly, since $\Sigma^{aug}_L\subseteq\Sigma^{aug}$, consistency of the superset as shown in \Cref{theorem: main}  implies consistency of the subset. Then consider a logical statement $S_1\in \Sigma_L^{aug}$ associated with a logical setting-conditioned prediction $P(\vec{a}_j=\vec{\aaa}_j| \vec{a}_l=\vec{\aaa}_l,\vec{x}=\vec{\xxx})\in \{0,1\}$, it can be of the following form where the set $\vec{a}_l$ of outcomes can be empty

\begin{align}
   \begin{split}
       \vec{a}_l=\vec{\aaa}_l \land \vec{x}=\vec{\xxx} &\Rightarrow \neg (\vec{a}_j=\vec{\aaa}_j),\\
       \vec{a}_l=\vec{\aaa}_l \land \vec{x}=\vec{\xxx}  &\Rightarrow \vec{a}_j=\vec{\aaa}_j.
   \end{split}
\end{align}

Let $S_2$ be another statement of the same form, but relative to a potentially different set of outcomes $\vec{a}_m$ and $\vec{a}_n$ and setting values $\vec{x}=\vec{\xxx'}$ i.e.,
\begin{align}
   \begin{split}
       \vec{a}_n=\vec{\aaa}_n \land \vec{x}=\vec{\xxx'} &\Rightarrow \neg (\vec{a}_m=\vec{\aaa}_m),\\
       \vec{a}_n=\vec{\aaa}_n \land \vec{x}=\vec{\xxx'}  &\Rightarrow \vec{a}_m=\vec{\aaa}_m.
   \end{split}
\end{align}
This corresponds to $P(\vec{a}_m=\vec{\aaa}_m| \vec{a}_n=\vec{\aaa}_n,\vec{x}=\vec{\xxx'})\in \{0,1\}$. Denote the probabilities associated with $S_1$ and $S_2$ as $P_1$ and $P_2$ in short. If $S_1\Rightarrow S_2$, then $P_2$ can be derived from $P_1$ through the rules of classical probability theory and usual manipulation of quantum circuits (in this case for the augmented circuit).

This implies that if one were to directly compute the probabilities for $\vec{a}_m$, $\vec{a}_n$ under the setting choice $\vec{x}=\vec{\xxx'}$ in the augmented circuit, one must arrive at $P_2$ (otherwise there can be two distinct probability assignments to the same outcomes given these settings, which is not possible due to consistency, \Cref{def: consistency_pred} and \Cref{theorem: main}). This shows that $S_2\in \Sigma_L^{aug}$. As this holds for all agents using our framework to reason, it follows that \cref{eq: logic_axiom1} holds.

\end{proof}

\FRtheorem*

\begin{proof}

The existence claim follows from the main results of our framework \Cref{ssec: general_results} and \Cref{corollary: QUCDSassumptions}. The consistent description satisfying \hyperref[def: Q_assump]{$\mathbf{Q}$}, \hyperref[def: U_assump]{$\mathbf{U}$}, \hyperref[def: assump_C]{$\mathbf{C}$}, \hyperref[def: assump_D]{$\mathbf{D}$} and \hyperref[def: assump_S]{$\mathbf{S}$} is given by the augmented circuit. The violation of \hyperref[def: I_assump]{$\mathbf{I}$} for certain logical setting-conditioned prediction follows from the analysis of \Cref{appendix: setting_dep_FR}, where we have shown that the logical predictions  $P(b = 1|u = w =\textup{ok},(x_1,x_2)=(0,1))=1$ and $P(w = \textup{fail}|a=1,(x_1,x_2)=(1,0))=1$ are setting-dependent (thus violating \hyperref[def: I_assump]{$\mathbf{I}$}), since 

\begin{align}
	\begin{split}
		&P(b = 1|u = w =\textup{ok},(x_1,x_2)=(0,1)) \neq P(b = 1|u = w =\textup{ok},(x_1,x_2)=(1,1)),\\
		&P(w = \textup{fail}|a=1,(x_1,x_2)=(1,0))  \neq P(w = \textup{fail}|a=1,(x_1,x_2)=(1,1))
	\end{split}
\end{align}

In the prepare and measure version analysed in \Cref{ssec: resolution_prep}, the logical predictions $P(w=\textup{fail}|r=\textup{tails}, (x_1,x_2)=(1,0))=1$ and $P(z=+\frac{1}{2}|\overline{w}=\overline{\textup{ok}}, (x_1,x_2)=(0,1))=1$ are setting-dependent and violate \hyperref[def: I_assump]{$\mathbf{I}$}, since

\begin{align}
	\begin{split}
		&P(w=\textup{fail}|r=\textup{tails}, (x_1,x_2)=(1,0))\neq P(w=\textup{fail}|r=\textup{tails}, (x_1,x_2)=(1,1)),\\
		&P(z=+\frac{1}{2}|\overline{w}=\overline{\textup{ok}}, (x_1,x_2)=(0,1))\neq P(z=+\frac{1}{2}|\overline{w}=\overline{\textup{ok}}, (x_1,x_2)=(1,1)).
	\end{split}
\end{align}

These precisely correspond to the three statements $\textup{\tb F}^{n:02}$ and $\textup{\tb W}^{n:22}$ of FR that are combined to yield the apparent paradox (as shown in \cref{table: resolution_prepmeas}).

The fact that assuming \hyperref[def: I_assump]{$\mathbf{I}$} is necessary to reproduce the apparent paradox in both versions of the protocol, follows from \Cref{corollary:setting_independence} and the sufficiency follows from noting that assuming \hyperref[def: I_assump]{$\mathbf{I}$} for all logical setting-conditioned predictions allows us to ignore the setting information on all such predictions and consequently on all associated logical statements. It is clear from \cref{table: resolution_ent,table: resolution_prepmeas} that this recovers the original statements of FR (in both versions) and therefore the apparent paradox.
\end{proof}

\SettingIndep*

\begin{proof}

Given any set $\A_{\matholdcal{K}}$ of agents and another agent $\A_i$ such that $(\A_i,\A_k)\in n\matholdcal{SA}$ $\forall \A_k\in\A_{\matholdcal{K}}$, we will first show that the joint probability of the outcomes of all agents in $\A_{\matholdcal{K}}$ (these outcomes will be denoted using the vector $\vec{a}_{\matholdcal{K}}$) is independent of the setting $x_i$ i.e., 
\begin{align}
    \label{eq: setting_indep_proof1}
    \begin{split}
       P(\vec{a}_{\matholdcal{K}}|(x_1,...,x_N))=P(\vec{a}_{\matholdcal{K}}|(x_1,...,x_{i-1},x_{i+1},...x_N)).     
    \end{split}
\end{align}
Here we have shortened the usual notion using value assignments $a=\aaa$ in probabilities to just the variable $a$, since the meaning of the equation is clear from context. From this, the desired result will follow easily through the rules of conditional probability.

To establish \cref{eq: setting_indep_proof1}, we consider the operationally equivalent augmented circuit (involving agents $\{\A'_1,...,\A'_N\}$) to the original augmented circuit (involving agents $\{\A_1,...,\A_N\}$), whose causal structure more explicitly reflects non-super-agent structure $n\matholdcal{SA}$ of the original EWFS. Recall that such an equivalent circuit is guaranteed to exist by \Cref{definition: nSA}. The equivalence implies in particular, the equivalence of predictions in the two circuits, that is, for all disjoints sets of outcomes $\vec{a}_j$ and $\vec{a}_l$ in the original EWFS,  

\begin{align}
    \begin{split}
    \label{eq: setting_indep_proof1p}
        P(\vec{a}_j=\vec{\aaa}_j|\vec{a}_l=\vec{\aaa}_l,\vec{x}=\vec{\xxx})=P(\vec{a}'_j=\vec{\aaa}_j|\vec{a}'_l=\vec{\aaa}_l,\vec{x}'=\vec{\xxx}).
    \end{split}
\end{align}

We establish \cref{eq: setting_indep_proof1} for the primed EWFS, and by the above equivalence of predictions, obtain the same for the original EWFS. We do so by dividing the problem into different cases, depending on the causal structure of the (equivalent) augmented circuit. According to \Cref{definition: nSA}, $(\A_i,\A_k)\in n\matholdcal{SA}$ (along with the fact that we have $i\neq k$) implies that in the equivalent EWFS, one of the following cases must hold

\begin{itemize}
    \item {\bf Case 1:} $\A'_i\not\prec\A'_k$
    \item {\bf Case 2:}  $\A'_i\prec \A'_k$, and $\matholdcal{E}'_i$ acts trivially on $\M'_i$ and $\A_i \not\prec^{\M_i'} \A_k$
\end{itemize}

{\bf Case 1} In this case \cref{eq: setting_indep_proof1} immediately follows from the causality result of \Cref{theorem: main}.

{\bf Case 2} In this case we must have the following property: for any agent $\A'_j$ with $\A'_i \prec^{\M'_i} \A'_j$ we must have $\A'_j \not\prec \A'_k$ for all $k\in \matholdcal{K}$. Otherwise $\A'_i \prec^{\M'_i} \A'_j$ along with the directed path $\A'_j \prec \A'_k$ would immediately imply by \Cref{definition: causal_str} that $\A'_i \prec^{\M'_i} \A'_k$, which would contradict the condition of Case 2 and hence the fact that $(\A_i,\A_k)\in n\matholdcal{SA}$ \Cref{definition: nSA}.

Now, the prediction of interest from \cref{eq: setting_indep_proof1} can be calculated explicitly by applying \cref{eq:prob of events given x} as detailed in \Cref{appendix: prob_rule_general}. The above property guarantees that for all $\A'_j$ such that $\A'_i \prec^{\M'_i} \A'_j$, the operations of $\A'_j$ and those of any agent $\A'_k\in \A'_{\matholdcal{K}}$ must act of disjoint sets of systems \Cref{definition: causal_str}. Note that $\A'_i \prec^{\M'_i} \A'_j$ implies that $\A'_j$ acts after $\A'_i$ in the augmented circuit. Then by the non-signalling property of quantum theory, and as we can explicitly see from \cref{eq:prob of events given x}, the trace in the probability calculation will act on all the outputs of $\A'_j$'s operation. Since $\A'_j$ is not in the set $\A'_{\matholdcal{K}}$, we sum over their outcomes in the probability, therefore all operations of such agents are completely positive and trace preserving maps (CPTPMs). For CPTPMs tracing the output after applying the map is equivalent to tracing the inputs directly $\tr\circ\matholdcal{E}=\tr$.

The above implies that all the operations of the agents $\A'_j$ such that $\A'_i \xrightarrow{\M_i'} \A'_j$ will drop out of the probability expression for \cref{eq: setting_indep_proof1}. This probability expression will then have the trace over $\M'_i$ acting directly on the operation $\matholdcal{E}'_i\circ\text{$\matholdcal{M}$}'^{\A_i}$ of the agent $\A'_i$. From the definition of Case 2, we also know that $\matholdcal{E}'_i$ acts trivially on $\M'_i$, which implies that the trace over $\M'_i$ will act directly on the measurement $\matholdcal{M}^{\A_i}$.

It is easy to see that $\tr_{\M'_i}\circ \matholdcal{M}^{\A_i}$ is the same, independently of the setting $x_i\in \{0,1\}$, because the setting dictates whether we coherently ($x_i=0$) or incoherently ($x_i=1$) copy the system state onto the memory in the basis of the measurement, and the post-measurement state on the system alone (obtained by tracing out the memory) is the same in both cases. A more explicit proof of this fact can be found in the proof of \Cref{theorem: standardQT}. This is sufficient to establish the required \cref{eq: setting_indep_proof1} for this case.

Having established \cref{eq: setting_indep_proof1} for all sets $\A_{\matholdcal{K}}$ of agents of the form required by the theorem statement, it is immediate that the same setting independence of \cref{eq: setting_indep_proof1} also holds for all subsets of agents  $\A_{\matholdcal{K}}$ (this can be seen by computing the relevant marginals of \cref{eq: setting_indep_proof1}). Now, for any partition $A_{\matholdcal{K}}=\{\A_{j_1},...,\A_{j_p}\}\cup \{\A_{l_1},...,\A_{l_q}\}$ of $\A_{\matholdcal{K}}$, we can compute prediction $P(\vec{a}_j|\vec{a}_l,\vec{x})$ by applying the conditional probability rule
\begin{align}
    \begin{split}
        P(\vec{a}_j|\vec{a}_l,\vec{x})
        =\frac{P(\vec{a}_j,\vec{a}_l|\vec{x})}{P(\vec{a}_l|\vec{x})}.
    \end{split}
\end{align}

The fact that both the numerator and denominator of this expression are independent of the component $x_i$ of the setting vector $\vec{x}$ are then immediate from \cref{eq: setting_indep_proof1}. This establishes the theorem.

\end{proof}

\subsection{Proofs of results from the Appendix}
\label{ssec: proofs_appendix}

\MainTheorem*

\begin{proof}

\begin{enumerate}
	\item \emph{Completeness: } 
Recall that a conventional prediction $P_{conv}(\vec{a}_j=\vec{\aaa}_j|\vec{a}_l=\vec{\aaa}_l)$ in an EWFS (\Cref{def: conv_prediction}) is computed by modelling the measurements $\matholdcal{M}^{\A_i}$ for all $i\in\{j_1,...,j_p,l_1,...,l_q\}$ as purely unitary evolutions of the agents' labs (\cref{eq:M_unitary_gen}).  By construction of the augmented circuit for the EWFS, this corresponds to the case where $x_i=0$ for all $i\in\{j_1,...,j_p,l_1,...,l_q\}$. 

    For $i\not\in\{j_1,...,j_p,l_1,...,l_q\}$, in a conventional prediction, a projective measurement of $\{\mathtt{\Pi}^{\mathtt{S}_i}_{\aaa_i}=\ketbra{\aaa_i}{\aaa_i}_{\mathtt{S}_i}\}_{\aaa_i\in\mathtt{O}_i}$ is performed followed by a CNOT in the same basis with the system $\mathtt{S}_i$ as control and memory $\M_i$ as target (capturing the memory update after measurement). 
 This CNOT is precisely the unitary $\text{$\matholdcal{M}$}^{\A_i}_{unitary}$ (\cref{eq:M_unitary_gen}).

Now, following \Cref{sec: mapping_augcircuit} of the main text, consider an initial state $\ket{\psi}_{\mathtt{S}_i}\otimes \ket{0}_{\M_i}$ of the system and memory  (on which the measurement $\matholdcal{M}^{\A_i}$ acts) where
$\ket{\psi}_{\mathtt{S}_i}=\sum_{\aaa_i\in \mathtt{O}_i}c_{\aaa_i}\ket{\aaa_i}_{\mathtt{S}_i}$. That section of the main text shows that the following two procedures are operationally equivalent: (1) a projective measurement $\{\pi_{\aaa_i}^{\mathtt{S}_i}=\ketbra{\aaa_i}{\aaa_i}_{\mathtt{S}_i}\}_{\aaa_i\in \mathtt{O}_i}$ is applied on the system and then the unitary channel $\text{$\matholdcal{M}$}^{\A_i}_{unitary}$ is applied on the system and memory (2) the unitary channel $\text{$\matholdcal{M}$}^{\A_i}_{unitary}$  is applied on the system and memory, and then the projective measurement $\{\pi_{\aaa_i}^{\mathtt{S}_i\M_i}=\ketbra{\aaa_i\aaa_i}{\aaa_i\aaa_i}_{\mathtt{S}_i\M_i}\}_{\aaa_i\in \mathtt{O}_i}$ is performed on the system and memory.  The operational equivalence of (1) and (2) entails that they yield the same transformation on the initial state, and also that the measurements in both cases yield the same probabilities for an outcome $\aaa_i$.

Generally by linearity, the argument about the equivalence of (1) and (2) extends to all initial states $\rho_{\mathtt{S}_i}\otimes \ketbra{0}{0}_{\M_i}$ of the system and memory. (1) is the procedure used for dealing with a measurement $\matholdcal{M}^{\A_i}$ when computing conventional predictions involving the outcome $\aaa_i$ (\Cref{def: conv_prediction}) while (2) is the procedure used in the augmented EWFS for calculating the setting conditioned predictions (\Cref{definition: setting_prediction}) involving the outcome $\aaa_i$ where by default we will have $x_i=1$.

This shows that the conventional prediction for the probability of $\vec{a}_j=\vec{\aaa}_j$ given $\vec{a}_l=\vec{\aaa}_l$ and the augmented circuit prediction for the same yield the same answer when using the setting assignment $\vec{x}=\vec{\xxx^*}$ defined as: $x_i=1$ for $i\in \{j_1,...,j_p,l_1,...,l_q\}$ and $x_i=0$ otherwise. The only difference being that the corresponding prediction in the augmented circuit explicitly conditions on the setting choice $\vec{x}=\xxx^*$, and we have $P_{conv}(\vec{a}_j=\vec{\aaa}_j|\vec{a}_l=\vec{\aaa}_l)=P(\vec{a}_j=\vec{\aaa}_j|\vec{a}_l=\vec{\aaa}_l,\vec{x}=\vec{\xxx^*})$.

This establishes the claim that all the conventional predictions of any given EWFS can be derived as particular cases of setting-conditioned predictions in the single augmented circuit of the EWFS, showing the completeness of our formalism.

\item {\emph Consistency: } The fact that the set $\Sigma^{aug}$ of statements associated with an EWFS in our framework satisfies the consistency condition of \Cref{def: consistency_pred} immediately follows from the procedure through which these statements are derived.

As defined in \Cref{def: all_aug_statements}, each statement $S\in \Sigma^{aug}$ is associated with a setting-conditioned prediction $P(\vec{a}_j=\vec{\aaa}_j|\vec{a}_l=\vec{\aaa}_l,\vec{x}=\vec{\xxx})$ of the given EWFS. Such predictions are computed by applying the Born rule to a single, well-defined quantum circuit (explicitly detailed in \Cref{appendix: prob_rule_general}), which implies that for any given setting choice $\vec{x}=\vec{\xxx}$ and sets of outcome values $\vec{a}_j=\vec{\aaa}_j$, $\vec{a}_l=\vec{\aaa}_l$ in a given EWFS, a unique setting-conditioned prediction $P(\vec{a}_j=\vec{\aaa}_j|\vec{a}_l=\vec{\aaa}_l,\vec{x}=\vec{\xxx})$ can be computed, which corresponds to a valid, well-defined and normalised conditional probability distribution.

Since $\Sigma^{aug}$ only contains statements associated with setting-conditioned predictions in the augmented circuit, and it is impossible to have $P(\vec{a}_j=\vec{\aaa}_j|\vec{a}_l=\vec{\aaa}_l,\vec{x}=\vec{\xxx})$ and $P'(\vec{a}_j=\vec{\aaa}_j|\vec{a}_l=\vec{\aaa}_l,\vec{x}=\vec{\xxx})$ in the same augmented EWFS for $P\neq P$, the consistency of $\Sigma^{aug}$ according to \Cref{def: consistency_pred} follows.

\item {\emph Causality: } Our framework provides a single well-defined quantum circuit (the augmented circuit) from which all 
 setting-conditioned predictions in an EWFS can be derived, and in which all the operations are applied in an acyclic order (given by the DAG $G$). From this, the desired result follows immediately by applying well-known results on quantum causal networks or causal models, such as the $d$-separation theorem \cite{Pearl2009,Tucci2007, Henson2014, Barrett2020A}. However, in the interest of not introducing new concepts, we describe the proof in terms of concepts and results introduced in this paper.

Consider the setting-conditioned prediction $P(\vec{a}_j=\vec{\aaa}_j|\vec{a}_l=\vec{\aaa}_l,\vec{x}=\vec{\xxx})$, and the condition (C) $\A_i\not\prec \A_k$ for all $k\in\{j_1,...,j_p,l_1,...,l_q\}$. Notice that an $\A_i$ satisfies this condition, in particular when $t_i>$ max$(t_{j_1},...,t_{j_p},t_{l_1},...,t_{l_q})$. In \Cref{appendix: prob_rule_general}, we have established that any setting-conditioned prediction $P(\vec{a}_j=\vec{\aaa}_j|\vec{a}_l=\vec{\aaa}_l,\vec{x}=\vec{\xxx})$ in an augmented EWFS can be simplified to the form of \cref{eq:general conditional prob rule simplicied 2}, where the components $x_i$ of the setting vector $\vec{x}$ which are associated with a time $t_i>$ max$(j_1,...,j_p,l_1,...,l_q)$ do not feature in the probability expression, which implies the required independence \cref{eq: causality_indep} for all such settings $x_i$.

For cases where we have $t_i<t_k$ for some $k\in\{j_1,...,j_p,l_1,...,l_q\}$ but $\A_i\not\prec \A_k$, by construction, this would only happen when the circuit contains no directed path of wires from the measurement $\matholdcal{M}^{\A_i}$ at $t_i$ to the measurement $\matholdcal{M}^{\A_k}$ at $t_k>t_i$ (absence of directed paths between corresponding nodes in $G$). In this case, the two measurements act on disjoint sets of systems and the non-signalling property of quantum theory would then guarantee the required independence of the outcome of one measurement from the setting of the other.

Alternatively, in such cases, we can always transform to an equivalent circuit (with the same channels, systems, states and same connectivity between channels) where $\A_i$ and $\A_k$ are assigned time $t'_i$ and $t'_k$ with the opposite time order, $t'_i>t_k$, allowing us to apply the argument from the previous paragraph to establish the required independence \cite{Note19}.

\end{enumerate}

\end{proof}

\StandardQT*

\begin{proof}

In a standard quantum scenario, by \Cref{def:std_q_exp}, we must have $(\A_i,\A_j)\in n\matholdcal{SA}$ for all $i,j\in \{1,...,N\}$. This means that the augmented circuit of the given EWFS over agents $\{\A_1,...,\A_N\}$ can be reduced to an operationally equivalent augmented circuit over corresponding agents $\{\A'_1,...,\A'_N\}$ such that for any $\A'_i$ and $\A'_j$ we have: either $\A'_j$ acts before $\A'_i$ in time $t'_i>t'_j$ or in the alternative cases where $i=j$ or $t'_j>t'_i$, no operation that acts after the measurement $\matholdcal{M}^{\A'_i}$ in the circuit (including the operation $\matholdcal{E}'_i$) acts non-trivially on the memory $\M'_i$. 
Here we have used the fact that $\A'_i\prec\A'_j$ according to the operational causal structure (\Cref{definition: causal_str}) implies $t'_i<t'_j$.

Applying this argument to every pair of agents in the scenario, this implies that for each measurement $\matholdcal{M}^{\A'_i}$, the corresponding set of systems $\mathtt{S}'_i$ of systems (excluding the memory $\M'_i$) on which it acts non-trivially is a subset of the systems $\mathtt{S}'$ and does not include any of the memories $\M'_j\in \mathtt{M}'$ of other agents (in contrast to a general EWFS of \Cref{def:LWFS} where the ``system'' for one agent's measurement may include the `` memories'' of other agents). All the following arguments will refer to the equivalent augmented circuit over the primed agents.

    Let the measurement outcomes $a_i$ of $\A'_i$'s measurement 
    take values in the same set $\aaa_i\in \mathtt{O}_i$ as that of the original scenario w.l.o.g. \cite{Note20}. Then $\A'_i$'s measurement on the systems $\mathtt{S}'_i$ is associated with the projectors $\{\ket{\aaa_i}\bra{\aaa_i}_{\mathtt{S}'_i}\}_{\aaa_i\in \mathtt{O}_i}$. Let $\rho_{\mathtt{S}'_i}$ be an arbitrary state, which corresponds to the state of the system $\mathtt{S}'_i$ just before $\A'_i$'s measurement. We can express this state is the basis of the measurement as follows, for some coefficients $c_{\aaa_i,\bar{\aaa}_i}$.
    $$\rho_{\mathtt{S}'_i}=\sum_{\aaa_i,\bar{\aaa}_i}c_{\aaa_i,\bar{\aaa}_i}\ket{\aaa_i}\bra{\bar{\aaa}_i}_{\mathtt{S}'_i},$$ 

where $\{\ket{\aaa}_i\}$ and $\{\ket{\bar{\aaa}}_i\}$ both correspond to the same measurement basis. Recall that the memory $\M'_i$ is initialised to $\ket{0}\bra{0}_{\M'_i}$. Now consider the case where the setting $x'_i=0$. Then we model $\A'_i$'s measurement as the unitary evolution $\text{$\matholdcal{M}$}^{\A'_i}_{unitary}$ which as we have seen before, is simply a unitary implementing a coherent copy in the measurement basis with $\mathtt{S}'_i$ being the control and $\M'_i$ being the target. The post-measurement state of $\mathtt{S}'_i\M'_i$ in this case is given as
    \begin{align}
        \begin{split}
        \label{eq: standardQT_proof1}
          \text{$\matholdcal{M}$}^{\A'_i}_{unitary} (\rho_{\mathtt{S}'_i}\otimes \ket{0}\bra{0}_{\M'_i})\text{$\matholdcal{M}$}^{\A'_i,\dagger}_{unitary} =\sum_{\aaa_i,\bar{\aaa}_i}c_{\aaa_i,\bar{\aaa}_i}\ket{\aaa_i\aaa_i}\bra{\bar{\aaa}_i\bar{\aaa}_i}_{\mathtt{S}'_i\M'_i}.
        \end{split}
    \end{align}
   Then it is easy to see that the following holds, where we use the notation $\pi_{\aaa_i}^{\mathtt{S}'_i}:=\ket{\aaa_i}\bra{\aaa_i}_{\mathtt{S}'_i}$, $\pi_{\aaa_i}^{\M'_i}:=\ket{\aaa_i}\bra{\aaa_i}_{\M'_i}$, $\pi_{\aaa_i}^{\mathtt{S}'_i\M'_i}=\ket{\aaa_i\aaa_i}\bra{\aaa_i\aaa_i}_{\mathtt{S}'_i\M'_i}$. 
 
    \begin{align}
        \begin{split}
            \label{eq: standardQT_proof2}   & (\id_{\mathtt{S}'_i}\otimes \pi_{\aaa_i}^{\M'_i})\big( \text{$\matholdcal{M}$}^{\A'_i}_{unitary} (\rho_{\mathtt{S}'_i}\otimes \ket{0}\bra{0}_{\M'_i})\text{$\matholdcal{M}$}^{\A'_i,\dagger}_{unitary}\big)(\id_{\mathtt{S}'_i}\otimes \pi_{\aaa_i}^{\M'_i}) \\
           =& \pi_{\aaa_i}^{\mathtt{S}'_i\M'_i} \text{$\matholdcal{M}$}^{\A'_i}_{unitary} (\rho_{\mathtt{S}'_i}\otimes \ket{0}\bra{0}_{\M'_i})\text{$\matholdcal{M}$}^{\A'_i,\dagger}_{unitary} \pi_{\aaa_i}^{\mathtt{S}'_i\M'_i}
        \end{split}
    \end{align}

We now show that the post-measurement on the system $\mathtt{S}'_i$ alone is independent of the setting $x'_i$, and this indeed the post-measurement state one would get by directly applying the projective measurement of $\matholdcal{C}^{sys}$ on $\mathtt{S}'_i$. This will allows us to reduce our augmented circuit to a  $\matholdcal{C}^{sys}$-form standard quantum circuit as required. For this, first consider this state under the setting $x'_i=0$. This is given as $\tr_{\M'_i}\big[\text{$\matholdcal{M}$}^{\A'_i}_{unitary} (\rho_{\mathtt{S}'_i}\otimes \ket{0}\bra{0}_{\M'_i})\text{$\matholdcal{M}$}^{\A'_i,\dagger}_{unitary} \big]$. Noting that $\sum_{\aaa_i\in \mathtt{O}_i}\pi_{\aaa_i}^{\M'_i}=\id_{\M'_i}$, we can expand this as follows

         \begin{align}
        \begin{split}
            \label{eq: standardQT_proof3}
          & \tr_{\M'_i}\big[\text{$\matholdcal{M}$}^{\A'_i}_{unitary} (\rho_{\mathtt{S}'_i}\otimes \ket{0}\bra{0}_{\M'_i})\text{$\matholdcal{M}$}^{\A'_i,\dagger}_{unitary} \big]  \\
           =& \sum_{\bar{\aaa}_i\in \mathtt{O}_i} (\id_{\mathtt{S}'_i}\otimes \bra{\bar{\aaa}_i}_{\M'_i})\big(\text{$\matholdcal{M}$}^{\A'_i}_{unitary} (\rho_{\mathtt{S}'_i}\otimes \ket{0}\bra{0}_{\M'_i})\text{$\matholdcal{M}$}^{\A'_i,\dagger}_{unitary}\big)(\id_{\mathtt{S}'_i}\otimes \ket{\bar{\aaa}_i}_{\M'_i})\\
           =& \sum_{\bar{\aaa}_i\in \mathtt{O}_i} (\id_{\mathtt{S}'_i}\otimes \bra{\bar{\aaa}_i}_{\M'_i})(\id_{\mathtt{S}'_i}\otimes \sum_{\aaa_i\in \mathtt{O}_i}\pi_{\aaa_i}^{\M'_i})\big(\text{$\matholdcal{M}$}'^{\A'_i}_{unitary} (\rho_{\mathtt{S}'_i}\otimes \ket{0}\bra{0}_{\M'_i})\text{$\matholdcal{M}$}^{\A'_i,\dagger}_{unitary}\big)(\id_{\mathtt{S}'_i}\otimes \sum_{\aaa_i\in \mathtt{O}_i}\pi_{\aaa_i}^{\M'_i})(\id_{\mathtt{S}'_i}\otimes \ket{\bar{\aaa}_i}_{\M'_i})\\
           =& \sum_{\bar{\aaa}_i\in \mathtt{O}_i} \sum_{\aaa_i\in \mathtt{O}_i}(\id_{\mathtt{S}'_i}\otimes \bra{\bar{\aaa}_i}_{\M'_i})(\id_{\mathtt{S}'_i}\otimes \pi_{\aaa_i}^{\M'_i})\big(\text{$\matholdcal{M}$}^{\A'_i}_{unitary} (\rho_{\mathtt{S}'_i}\otimes \ket{0}\bra{0}_{\M'_i})\text{$\matholdcal{M}$}^{\A'_i,\dagger}_{unitary}\big)(\id_{\mathtt{S}'_i}\otimes \pi_{\aaa_i}^{\M'_i})(\id_{\mathtt{S}'_i}\otimes \ket{\bar{\aaa}_i}_{\M'_i})\\
           =& \tr_{\M'_i}\big[\sum_{\aaa_i\in\mathtt{O}_i}(\id_{\mathtt{S}'_i}\otimes \pi_{\aaa_i}^{\M'_i})\big(\text{$\matholdcal{M}$}^{\A'_i}_{unitary} (\rho_{\mathtt{S}'_i}\otimes \ket{0}\bra{0}_{\M'_i})\text{$\matholdcal{M}$}^{\A'_i,\dagger}_{unitary}\big)(\id_{\mathtt{S}'_i}\otimes \pi_{\aaa_i}^{\M'_i})\big].
        \end{split}
    \end{align}

In the second equality above, we have summed over the same indices $\sum_{\aaa_i\in \mathtt{O}_i}$ in both occurences of the projectors $\pi_{\aaa_i}^{\M'_i}$ as the cross terms disappear due to the trace over $\M'_i$. Using this and \cref{eq: standardQT_proof2}, we immediately obtain the following
\begin{align}
    \begin{split}
       & \tr_{\M'_i}\big[\text{$\matholdcal{M}$}^{\A'_i}_{unitary} (\rho_{\mathtt{S}'_i}\otimes \ket{0}\bra{0}_{\M'_i})\text{$\matholdcal{M}$}^{\A'_i,\dagger}_{unitary} \big] \\
       =& \tr_{\M'_i}\big[\sum_{\aaa_i\in \mathtt{O}_i}\pi_{\aaa_i}^{\mathtt{S}'_i\M'_i}\text{$\matholdcal{M}$}^{\A'_i}_{unitary} (\rho_{\mathtt{S}'_i}\otimes \ket{0}\bra{0}_{\M'_i})\text{$\matholdcal{M}$}^{\A'_i,\dagger}_{unitary} \pi_{\aaa_i}^{\mathtt{S}'_i\M'_i}\big]\\
       =& \sum_{\aaa_i\in \mathtt{O}_i}c_{\aaa_i,\aaa_i}\ket{\aaa_i}\bra{\aaa_i}_{\mathtt{S}'_i}.
    \end{split}
\end{align}
 The right hand side is precisely the expression for the post-measurement state of $\mathtt{S}'_i$ under the setting $x_i=1$ when using the trace preserving form of the evolution for this setting (c.f. \cref{eq: proj_settings}). Thus we have shown that $\A'_i$'s measurement implements the following map on an arbitrary input state $\rho_{\mathtt{S}'_i}$ of the system $\mathtt{S}'_i$ irrespective of the setting $x'_i$.
\begin{equation}
  \rho_{\mathtt{S}'_i}= \sum_{\aaa_i,\bar{\aaa}_i\in \mathtt{O}_i}c_{\aaa_i,\bar{\aaa}_i}\ket{\aaa_i}\bra{\bar{\aaa}_i}_{\mathtt{S}'_i} \mapsto \sum_{\aaa_i\in \mathtt{O}_i}c_{\aaa_i,\aaa_i}\ket{\aaa_i}\bra{\aaa_i}_{\mathtt{S}'_i}
\end{equation}

This map can equivalently be described as
\begin{align}
    \begin{split}
      \rho_{\mathtt{S}'_i} \mapsto    \sum_{\aaa_i\in \mathtt{O}_i}\pi_{\aaa_i}^{\mathtt{S}'_i}(\rho_{\mathtt{S}'_i})\pi_{\aaa_i}^{\mathtt{S}'_i}.
    \end{split}
\end{align}

Or, for the case of a particular outcome, we have the following trace non-increasing map from pre to post measurement state
\begin{align}
    \begin{split}
            \rho_{\mathtt{S}'_i} \mapsto    \frac{\pi_{\aaa_i}^{\mathtt{S}'_i}(\rho_{\mathtt{S}'_i})\pi_{\aaa_i}^{\mathtt{S}'_i}}{\tr\Big[\pi_{\aaa_i}^{\mathtt{S}'_i}(\rho_{\mathtt{S}'_i})\pi_{\aaa_i}^{\mathtt{S}'_i}\Big]}.  
    \end{split}
\end{align}

These correspond precisely to the how the measurements $\text{$\matholdcal{M}$}^{\A'_i}=\{\pi_{\aaa_i}^{\mathtt{S}'_i}\}_{\aaa_i\in \mathtt{O}'_i}$ of a $\matholdcal{C}^{sys}$-form circuit (\Cref{definition: std_circuit_forms}) over $\{A'_1,...,A'_N\}$ act.

To show the full reduction of the augmented circuit to a $\matholdcal{C}^{sys}$-form circuit, let $\mathtt{S}_i^{'c}:=\mathtt{S}'\backslash \mathtt{S}'_i$ denote the complement of  $\mathtt{S}_i\subseteq \mathtt{S}'$ ( the systems that $\matholdcal{M}^{\A'_i}$ acts non-trivially on). The above argument, immediately generalises to the case where we consider $\rho_{\mathtt{S}'_i}$ to be a reduced state of a larger state $\rho_{\mathtt{S}_i^{'c}\mathtt{S}'_i}$ over all systems $\mathtt{S}'$ in the EWFS. Since under both settings, the measurement operation of $\A'_i$ only acts locally on $\mathtt{S}'_i$ by construction, the above argument implies that the joint state on  $\mathtt{S}_i^{'c}\mathtt{S}'_i$ is also independent of the settings. It follows that for all agents $\A'_i$, we can replace their measurement in the (primed) augmented circuit with the setting independent projective measurement $\{\ket{\aaa_i}\bra{\aaa_i}_{\mathtt{S}'_i}\}_{\aaa_i\in \mathtt{O}'_i}$ on the system $\mathtt{S}'_i\subseteq\mathtt{S}'$ alone while preserving all the predictions i.e., the augmented circuit reduces to an equivalent $\matholdcal{C}^{sys}$-form standard circuit. Since the primed circuit is operationally equivalent to the original augmented EWFS, it follows that the  $\matholdcal{C}^{sys}$-form standard circuit obtained here is also operationally equivalent to the original augmented circuit.

To obtain an equivalent $\matholdcal{C}^{sys+anc}$-form circuit, we can keep the memories (even if they are not acted upon after $\A'_i$'s measurement) and model all measurements with $x'_i=1$ (recalling that we have established full setting independence of predictions, therefore an arbitrary fixing of settings will not change the predictions). Then, it follows from \cref{eq: standardQT_proof2} that this is equivalent to modelling each measurement in its unitary form $\text{$\matholdcal{M}$}^{\A'_i}_{unitary}$ at the time $t'_i$ of the measurement followed by a projective measurement $\{\pi^{\M'_i}_{\aaa_i}\}_{\aaa_i\in \mathtt{O}_i}$ on the memory $\M'_i$ alone. Since the memory is not subsequently acted upon until the time $t'_N$ at which the protocol ends, the measurement $\{\pi^{\M'_i}_{\aaa_i}\}_{\aaa_i\in \mathtt{O}_i}$ acting on the memory alone, can be equivalently performed at any time $t_f$ after $t'_N$. This immediately yields $\matholdcal{C}^{sys+anc}$-form standard circuit which is equivalent to the original augmented circuit, where the memories $\{\M'_1,...,\M'_N\}$ act as the ancillas. This completes the proof of the current theorem.

\end{proof}


\begin{thebibliography}{10}

\bibitem{Wigner1967}
Eugene~P. Wigner.
\newblock {\em Remarks on the Mind-Body Question}, pages 171--184.
\newblock Indiana University Press, 1967.
\newblock \url{https://link.springer.com/content/pdf/10.1007/978-3-642-78374-6_20.pdf}.

\bibitem{Baumann_2018}
Veronika Baumann and Stefan Wolf.
\newblock On formalisms and interpretations.
\newblock {\em Quantum}, 2:99, 2018.
\newblock \url{http://dx.doi.org/10.22331/q-2018-10-15-99}.

\bibitem{Frauchiger2018}
Daniela Frauchiger and Renato Renner.
\newblock Quantum theory cannot consistently describe the use of itself.
\newblock {\em Nature Communications}, 9:3711, 2018.
\newblock \url{https://doi.org/10.1038/s41467-018-05739-8}.

\bibitem{Nurgalieva2018}
Nuriya Nurgalieva and L{\'{\i}}dia del Rio.
\newblock Inadequacy of modal logic in quantum settings.
\newblock {\em In Proceedings QPL 2018, EPTCS 287}, pages pp. 267--297, 2019.
\newblock \url{https://arxiv.org/abs/1804.01106v2}.

\bibitem{brukner2015}
\v{C}aslav Brukner.
\newblock On the quantum measurement problem, 2015.
\newblock \url{https://arxiv.org/abs/1507.05255}.

\bibitem{Brukner2018}
\v{C}aslav Brukner.
\newblock A no-go theorem for observer-independent facts.
\newblock {\em Entropy}, 20(5), 2018.
\newblock \url{https://www.mdpi.com/1099-4300/20/5/350}.

\bibitem{Bong2020}
Kok-Wei Bong, An{\'\i}bal Utreras-Alarc{\'o}n, Farzad Ghafari, Yeong-Cherng Liang, Nora Tischler, Eric~G. Cavalcanti, Geoff~J. Pryde, and Howard~M. Wiseman.
\newblock {A strong no-go theorem on the Wigner's friend paradox}.
\newblock {\em Nature Physics}, 16(12):1199--1205, 2020.
\newblock \url{https://doi.org/10.1038/s41567-020-0990-x}.

\bibitem{Wiseman_2023}
Howard~M. Wiseman, Eric~G. Cavalcanti, and Eleanor~G. Rieffel.
\newblock A “thoughtful” {L}ocal {F}riendliness no-go theorem: a prospective experiment with new assumptions to suit.
\newblock {\em Quantum}, 7:1112, 2023.
\newblock \url{http://dx.doi.org/10.22331/q-2023-09-14-1112}.

\bibitem{Zeng2025}
William~J. Zeng, Farrokh Labib, and Vincent Russo.
\newblock Towards violations of {L}ocal {F}riendliness with quantum computers.
\newblock {\em Quantum}, 9:1851, 2025.
\newblock \url{http://dx.doi.org/10.22331/q-2025-09-05-1851}.

\bibitem{Pusey2018}
Matthew Pusey.
\newblock An inconsistent friend.
\newblock {\em Nature Communications}, 2018.
\newblock \url{https://www.nature.com/articles/s41567-018-0293-7}.

\bibitem{Nurgalieva2020}
Nuriya Nurgalieva and Renato Renner.
\newblock Testing quantum theory with thought experiments.
\newblock {\em Contemporary Physics}, 61(3):193--216, 2020.
\newblock \url{https://doi.org/10.1080/00107514.2021.1880075}.

\bibitem{Renes2021}
Joseph~M. Renes.
\newblock Consistency in the description of quantum measurement: Quantum theory can consistently describe the use of itself, 2021.
\newblock \url{https://arxiv.org/abs/2107.02193}.

\bibitem{Losada2019}
Marcelo Losada, Roberto Laura, and Olimpia Lombardi.
\newblock Frauchiger-{R}enner argument and quantum histories.
\newblock {\em Phys. Rev. A}, 100:052114, 2019.
\newblock \url{https://link.aps.org/doi/10.1103/PhysRevA.100.052114}.

\bibitem{Alexios2022}
Alexios~P. Polychronakos.
\newblock Quantum mechanical rules for observed observers and the consistency of quantum theory.
\newblock {\em Nature Communications}, 15(1), 2024.
\newblock \url{http://dx.doi.org/10.1038/s41467-024-47170-2}.

\bibitem{ScottAronson}
{Aaronson, Scott}.
\newblock It’s hard to think when someone {H}adamards your brain.
\newblock \url{https://scottaaronson.blog/?p=3975}, 2018.

\bibitem{Healey2018}
Richard Healey.
\newblock Quantum theory and the limits of objectivity.
\newblock {\em Foundations of Physics}, 48(11):1568--1589, 2018.
\newblock \url{https://doi.org/10.1007/s10701-018-0216-6}.

\bibitem{Araujo}
Mateus Ara\'{u}jo.
\newblock The flaw in {F}rauchiger and {R}enner's argument.
\newblock \url{https://mateusaraujo.info/2018/10/24/the-flaw-in-frauchiger-and-renners-argument/}, 2018.

\bibitem{Sudbery2019}
Anthony Sudbery.
\newblock The hidden assumptions of frauchiger and renner.
\newblock {\em International Journal of Quantum Foundations}, 5(3):98--109, 2019.
\newblock \url{https://ijqf.org/archives/5470}.

\bibitem{Cavalcanti2021}
Eric~G. Cavalcanti and Howard~M. Wiseman.
\newblock Implications of {L}ocal {F}riendliness violation for quantum causality.
\newblock {\em Entropy}, 23(8), 2021.
\newblock \url{https://www.mdpi.com/1099-4300/23/8/925}.

\bibitem{Ying2023}
Yìlè Yīng, Marina~Maciel Ansanelli, Andrea~Di Biagio, Elie Wolfe, David Schmid, and Eric~Gama Cavalcanti.
\newblock Relating {W}igners friend scenarios to nonclassical causal compatibility, monogamy relations, and fine tuning.
\newblock {\em Quantum}, 8:1485, 2024.
\newblock \url{http://dx.doi.org/10.22331/q-2024-09-26-1485}.

\bibitem{Renner_Challenge}
L{\'\i}dia del Rio and Renato Renner.
\newblock Reply to: Quantum mechanical rules for observed observers and the consistency of quantum theory.
\newblock {\em Nature Communications}, 15(1):3024, 2024.
\newblock \url{https://doi.org/10.1038/s41467-024-47172-0}.

\bibitem{LF_Vilasini_Woods}
V.~Vilasini and Mischa~P. Woods.
\newblock In preparation.

\bibitem{Deutsch1985}
David Deutsch.
\newblock Quantum theory as a universal physical theory.
\newblock {\em International Journal of Theoretical Physics}, 24(1):1–41, 1985.
\newblock \url{http://dx.doi.org/10.1007/BF00670071}.

\bibitem{Note1}
In the following equation, we have used a pure initial state $\mathinner {|{\psi }\rangle }_{\protect \mathtt {S}_i}$ to make the equations more concise, but it is easy to verify that analogous equations and the same arguments hold for mixed initial states $\rho _{\protect \mathtt {S}_i}$.

\bibitem{Note2}
Collapse theories, for instance, would violate this assumption and prescribe different probabilities and rules for computing them.

\bibitem{Note3}
One can formally write down the probability expressions obtained in steps 3 and 4 through the Born rule. However, in the interests of conciseness, we will decline from doing so at this point, as we will show that these conventional predictions since are recovered as a special case of our to-be-described and rigorously defined general formalism (\protect \Cref {theorem: main}).

\bibitem{Note4}
FR refer to $C$ as an inter-agent ``consistency'' assumption, hence the name. Following the formulation of this in terms of epistemic modal logic \cite {Nurgalieva2018,Vilasini_2019} we find it more appropriate to view it as an inheritance of agents' knowledge. However we stick to the original acronym for simplicity. $D$ refers to the distributive axiom of logic and $S$ corresponds to an assumption about measurements yielding ``single outcomes'' in every experimental run.

\bibitem{Note5}
Recall that in our framework, the settings $x_i=1$ for all outcomes $a_i$ belonging to the outcome sets $\protect \vec {a}_j$ or $\protect \vec {a}_l$ appearing in the given prediction $P(\protect \vec {a}_j=\protect \vec {\protect \mathcal {a}}_j|\protect \vec {a}_l=\protect \vec {\protect \mathcal {a}}_l,\protect \vec {x}=\protect \vec {\xi })$, the remaining components of $\protect \vec {x}$ can be varied.

\bibitem{Note6}
This aligns with a previous work which also argued that FR's conclusions only hold under a specific way of applying quantum theory \cite {Sudbery_2017}. However our formalisation of this point via the assumption \protect \hyperref [def: I_assump]{$\protect \mathbf {I}$} allows to further achieve a full resolution of FR-type paradoxes meeting the desiderata outlined in \protect \Cref {Introduction}.

\bibitem{Note7}
We need not consider the case where Ursula and Wigner's measurements are modelled unitarily (setting 0) as this is only the case if there were further super-super agents who measured Ursula and Wigner's labs.

\bibitem{Note8}
See \cite {NurgalievaVilasini} for further details on a general, theory independent formalisation of this class of Wigner's Friend (apparent) paradoxes, derivable only by combining statements about compatible measurements, in addition to imposing \protect \hyperref [def: I_assump]{$\protect \mathbf {I}$} which we have shown to be necessary.

\bibitem{Note9}
Keep in mind that each agent in our formalism is associated with a single time step, and a physical agent acting at many time steps would be modelled as multiple agents acting here, thus one agent has to act later in time than another in order to act on the latter's memory.

\bibitem{Jones2025}
Caroline~L. Jones and Markus~P. M\"{u}ller.
\newblock On the significance of {W}igner's friend in contexts beyond quantum foundations, 2025.
\newblock \url{https://arxiv.org/abs/2402.08727}.

\bibitem{Note10}
These structures and the setting-independences they entail follow only from the properties of the EWFS protocol (and not agents' subjective choices of cuts) and can be included in the initial protocol specification given to the agents.

\bibitem{Griffiths1984}
Robert~B. Griffiths.
\newblock Consistent histories and the interpretation of quantum mechanics.
\newblock {\em Journal of Statistical Physics}, 36:219–272, 1984.
\newblock \url{https://doi.org/10.1007/BF01015734}.

\bibitem{GriffithsCH}
Robert~B. Griffiths.
\newblock {The Consistent Histories Approach to Quantum Mechanics}.
\newblock The Stanford Encyclopedia of Philosophy (Summer 2024 Edition), Edward N. Zalta and Uri Nodelman (eds.).
\newblock \url{https://plato.stanford.edu/entries/qm-consistent-histories}.

\bibitem{Zukowski2021}
Marek \ifmmode~\dot{Z}\else \.{Z}\fi{}ukowski and Marcin Markiewicz.
\newblock Physics and metaphysics of {W}igner's friends: Even performed premeasurements have no results.
\newblock {\em Phys. Rev. Lett.}, 126:130402, 2021.
\newblock \url{https://link.aps.org/doi/10.1103/PhysRevLett.126.130402}.

\bibitem{Note11}
This holds even when no projection postulate is used to set post-measurement states; see \protect \Cref {remark: no_collapse}.

\bibitem{Knott1911}
Cargill~G. Knott.
\newblock {\em Life and Scientific Work of Peter Guthrie Tait}.
\newblock Cambridge University Press, London, 1911.
\newblock Reproduced in \textit{The Scientific Letters and Papers of James Clerk Maxwell II: 1862--1873}, edited by P. M. Harman, Cambridge University Press, Cambridge, 1995. \url{https://archive.org/details/lifescientificwo00knotuoft/page/212/mode/2up?view=theater}.

\bibitem{Landauer1961}
Rolf Landauer.
\newblock Irreversibility and heat generation in the computing process.
\newblock {\em IBM Journal of Research and Development}, 5(3):183--191, 1961.
\newblock \url{https://ieeexplore.ieee.org/document/5392446}.

\bibitem{Bennett1982}
Charles~H. Bennett.
\newblock The thermodynamics of computation—a review.
\newblock {\em International Journal of Theoretical Physics}, 21:905--940, 1982.
\newblock \url{https://link.springer.com/article/10.1007/BF02084158}.

\bibitem{Note12}
Commonly called indefinite causal order \cite {Oreshkov2012,Chiribella2013}, though this terminology is often debated.

\bibitem{Pearl2009}
Judea Pearl.
\newblock {\em Causality: Models, Reasoning, and Inference}.
\newblock Second edition, Cambridge University Press, New York, NY, USA, 2009.

\bibitem{Henson2014}
Joe Henson, Raymond Lal, and Matthew~F. Pusey.
\newblock Theory-independent limits on correlations from generalized {B}ayesian networks.
\newblock {\em New Journal of Physics}, 16(11):p. 113043, 2014.
\newblock \url{https://iopscience.iop.org/article/10.1088/1367-2630/16/11/113043}.

\bibitem{Costa2016}
Fabio Costa and Sally Shrapnel.
\newblock {Quantum causal modelling}.
\newblock {\em New Journal of Physics}, 18(6):63032, 2016.
\newblock \url{https://iopscience.iop.org/article/10.1088/1367-2630/18/6/063032}.

\bibitem{Allen2017}
John-Mark~A. Allen, Jonathan Barrett, Dominic~C. Horsman, Ciar{\'{a}}n~M. Lee, and Robert~W. Spekkens.
\newblock Quantum common causes and quantum causal models.
\newblock {\em Physical Review X}, 7(3):p. 031021, 2017.
\newblock \url{https://journals.aps.org/prx/abstract/10.1103/PhysRevX.7.031021}.

\bibitem{Barrett2020A}
Jonathan Barrett, Robin Lorenz, and Ognyan Oreshkov.
\newblock Quantum causal models, 2020.
\newblock \url{https://arxiv.org/abs/1906.10726}.

\bibitem{Hardy2005}
Lucien Hardy.
\newblock {Probability Theories with Dynamic Causal Structure: A New Framework for Quantum Gravity}, 2005.
\newblock \url{http://arxiv.org/abs/gr-qc/0509120}.

\bibitem{Oreshkov2012}
Ognyan Oreshkov, Fabio Costa, and {\v{C}}aslav Brukner.
\newblock {Quantum correlations with no causal order}.
\newblock {\em Nature Communications}, 3:1092, 2012.
\newblock \url{https://www.nature.com/articles/ncomms2076}.

\bibitem{Chiribella2013}
Giulio Chiribella, Giacomo~Mauro D'Ariano, Paolo Perinotti, and Benoit Valiron.
\newblock {Quantum computations without definite causal structure}.
\newblock {\em Physical Review A}, 88(2):022318, 2013.
\newblock \url{https://link.aps.org/doi/10.1103/PhysRevA.88.022318}.

\bibitem{Araujo_2017}
Mateus Ara{\'u}jo, Philippe~Allard Gu{\'e}rin, and {\"A}min Baumeler.
\newblock Quantum computation with indefinite causal structures.
\newblock {\em Physical Review A}, 96(5), 2017.
\newblock \url{http://dx.doi.org/10.1103/PhysRevA.96.052315}.

\bibitem{Barrett2020}
Jonathan Barrett, Robin Lorenz, and Ognyan Oreshkov.
\newblock Cyclic quantum causal models.
\newblock {\em Nature Communications}, 12(1), 2021.
\newblock \url{http://dx.doi.org/10.1038/s41467-020-20456-x}.

\bibitem{ferradini2025c}
Carla Ferradini, Victor Gitton, and V.~Vilasini.
\newblock Cyclic functional causal models beyond unique solvability with a graph separation theorem, 2025.
\newblock \url{https://arxiv.org/abs/2502.04171}.

\bibitem{ferradini2025q}
Carla Ferradini, Victor Gitton, and V.~Vilasini.
\newblock Cyclic quantum causal modelling with a graph separation theorem, 2025.
\newblock \url{https://arxiv.org/abs/2502.04168}.

\bibitem{VilasiniColbeckPRA}
V.~Vilasini and Roger Colbeck.
\newblock General framework for cyclic and fine-tuned causal models and their compatibility with space-time.
\newblock {\em Phys. Rev. A}, 106:032204, 2022.
\newblock \url{https://link.aps.org/doi/10.1103/PhysRevA.106.032204}.

\bibitem{VilasiniColbeckPRL}
V.~Vilasini and Roger Colbeck.
\newblock Impossibility of superluminal signaling in minkowski spacetime does not rule out causal loops.
\newblock {\em Phys. Rev. Lett.}, 129:110401, 2022.
\newblock \url{https://link.aps.org/doi/10.1103/PhysRevLett.129.110401}.

\bibitem{VilasiniRennerPRA}
V.~Vilasini and Renato Renner.
\newblock Embedding cyclic information-theoretic structures in acyclic space-times: No-go results for indefinite causality.
\newblock {\em Phys. Rev. A}, 110:022227, 2024.
\newblock \url{https://link.aps.org/doi/10.1103/PhysRevA.110.022227}.

\bibitem{VilasiniRennerPRL}
V.~Vilasini and Renato Renner.
\newblock Fundamental limits for realizing quantum processes in spacetime.
\newblock {\em Phys. Rev. Lett.}, 133:080201, 2024.
\newblock \url{https://link.aps.org/doi/10.1103/PhysRevLett.133.080201}.

\bibitem{Vilasini_2019}
V.~Vilasini, Nuriya Nurgalieva, and L{\'{\i}}dia del Rio.
\newblock Multi-agent paradoxes beyond quantum theory.
\newblock {\em New Journal of Physics}, 21(11):113028, 2019.
\newblock \url{https://doi.org/10.1088/1367-2630/ab4fc4}.

\bibitem{ormrod2023}
Nick Ormrod, V.~Vilasini, and Jonathan Barrett.
\newblock Which theories have a measurement problem?
\newblock {\em New Journal of Physics}, 28:024506, 2026.
\newblock \url{https://iopscience.iop.org/article/10.1088/1367-2630/ae131e}.

\bibitem{NurgalievaVilasini}
Nuriya Nurgalieva and V.~Vilasini.
\newblock {Any theory that admits an extended {W}igner's Friend-type paradox is logically contextual}, 2025.
\newblock \url{https://arxiv.org/abs/2502.03874}.

\bibitem{PopescuRohrlich1994}
Sandu Popescu and Daniel Rohrlich.
\newblock Quantum nonlocality as an axiom.
\newblock {\em Foundations of Physics}, 24:379--385, 1994.
\newblock \url{https://doi.org/10.1007/BF02058098}.

\bibitem{Note13}
The memory $\protect \textup {M}_i$ of an agent $\protect \textup {A}_i$ is chosen to be a system of at least the same dimension as the system $\protect \mathtt {S}_i$ that the agent $\protect \textup {A}_i$ measures.

\bibitem{AsherPeres}
Asher Peres.
\newblock {\em Quantum Theory: Concepts and Methods}, volume~72.
\newblock Kluwer Academic Publishers, New York, 1993.
\newblock Fundamental Theories of Physics.

\bibitem{Note14}
In simple terms, a DAG is a graph with additional structure{:} 1) The edges have a direction associated with them. 2) By ``following the edges in the indicated direction'' one can never get back to the starting point, i.e. no ``directed loops''.

\bibitem{Note15}
In \cite {Nurgalieva2018}, the assumption C is replaced by what they call the trust axiom which asserts that an agent $A^i$ can inherit the knowledge of another agent $A^j$ as per \protect \cref {eq: C_logic} only if $A^i$ trusts $A^j$. In the FR setup, it is precisely these pairs of agents who can be said to trust each other. Other pairs such as Alice and Ursula need not trust each other as one agent Hadamard's the memory of the other.

\bibitem{Wood2015}
Christopher~J. Wood and Robert~W. Spekkens.
\newblock {The lesson of causal discovery algorithms for quantum correlations: causal explanations of {B}ell-inequality violations require fine-tuning}.
\newblock {\em New Journal of Physics}, 17(3):p. 33002, 2015.
\newblock \url{https://iopscience.iop.org/article/10.1088/1367-2630/17/3/033002}.

\bibitem{Relano2018}
Armando Rela\~no.
\newblock Decoherence allows quantum theory to describe the use of itself, 2018.
\newblock \url{https://arxiv.org/abs/1810.07065}.

\bibitem{Relano2020}
Armando Rela\~no.
\newblock {Decoherence framework for Wigner's-friend experiments}.
\newblock {\em Phys. Rev. A}, 101:032107, 2020.
\newblock \url{https://link.aps.org/doi/10.1103/PhysRevA.101.032107}.

\bibitem{Kastner2020}
Ruth~E. Kastner.
\newblock Unitary-only quantum theory cannot consistently describe the use of itself: On the {F}rauchiger--{R}enner paradox.
\newblock {\em Foundations of Physics}, 50(5):441--456, 2020.
\newblock \url{https://doi.org/10.1007/s10701-020-00336-6}.

\bibitem{Biagio2021}
Andrea Di~Biagio and Carlo Rovelli.
\newblock Stable facts, relative facts.
\newblock {\em Foundations of Physics}, 51(1):30, 2021.
\newblock \url{https://doi.org/10.1007/s10701-021-00429-w}.

\bibitem{Bub2019}
Jeffrey Bub.
\newblock 'two dogmas' redux, 2019.
\newblock \url{https://arxiv.org/abs/1907.06240}.

\bibitem{Narasimhachar2020}
Varun Narasimhachar.
\newblock Agents governed by quantum mechanics can use it intersubjectively and consistently, 2020.
\newblock \url{https://arxiv.org/abs/2010.01167}.

\bibitem{Drezet2018}
Aur\'{e}lien Drezet.
\newblock {About Wigner Friend's and Hardy's paradox in a Bohmian approach: a comment of `Quantum theory cannot' consistently describe the use of itself'}, 2018.
\newblock \url{https://arxiv.org/abs/1810.10917}.

\bibitem{Sudbery_2017}
Anthony Sudbery.
\newblock Single-world theory of the extended {W}igner’s friend experiment.
\newblock {\em Foundations of Physics}, 47(5):658–669, 2017.
\newblock \url{http://dx.doi.org/10.1007/s10701-017-0082-7}.

\bibitem{Fortin2019}
Sebastian Fortin and Olimpia Lombardi.
\newblock Wigner and his many friends: A new no-go result?, 2019.
\newblock \url{https://arxiv.org/abs/1904.07412}.

\bibitem{Note16}
The cited works are only representatives of the extensive research generated by the FR paper. For further references on previous responses to FR, see \cite {Nurgalieva2020}.

\bibitem{schmid2023review}
David Schmid, Yìlè Yīng, and Matthew Leifer.
\newblock {A review and analysis of six extended Wigner's friend arguments}.
\newblock \url{https://arxiv.org/abs/2308.16220}, 2023.

\bibitem{Rovelli1996}
Carlo Rovelli.
\newblock Relational quantum mechanics.
\newblock {\em International Journal of Theoretical Physics}, 35(8):1637--1678, 1996.
\newblock \url{https://doi.org/10.1007/BF02302261}.

\bibitem{Caves2002}
Carlton~M. Caves, Christopher~A. Fuchs, and R\"udiger Schack.
\newblock Quantum probabilities as {B}ayesian probabilities.
\newblock {\em Phys. Rev. A}, 65:022305, 2002.
\newblock \url{https://link.aps.org/doi/10.1103/PhysRevA.65.022305}.

\bibitem{Fuchs2014}
Christopher~A. Fuchs, N.~David Mermin, and R\"{u}diger Schack.
\newblock An introduction to {QB}ism with an application to the locality of quantum mechanics.
\newblock {\em American Journal of Physics}, 82(8):749--754, 2014.
\newblock \url{https://doi.org/10.1119/1.4874855}.

\bibitem{Samuel2022}
Stuart Samuel.
\newblock The {F}rauchiger-{R}enner gedanken experiment: an interesting laboratory for exploring some topics in quantum mechanics, 2022.
\newblock \url{https://arxiv.org/abs/2208.00060}.

\bibitem{Everett1957}
Hugh Everett.
\newblock ``{R}elative state" formulation of quantum mechanics.
\newblock {\em Rev. Mod. Phys.}, 29:454--462, 1957.
\newblock \url{https://link.aps.org/doi/10.1103/RevModPhys.29.454}.

\bibitem{Note17}
Although additional steps can be employed to simplify the reasoning by dropping redundant settings and parameters, as is also the case in classical multi-agent scenarios (cf. \protect \Cref {sec: reasoning_rules}).

\bibitem{Note18}
This fact is independent of whether or not different agents agree on the setting choices to model their perspective of the experiment, as for instance, even if Wigner models the Friend's lab as a unitarily evolving closed quantum system ($x_F=0$), both Wigner and the Friend can still use our augmented circuit to compute predictions for the case where $x_F=1$ and will arrive at the same answer.

\bibitem{Tucci2007}
Robert~R. Tucci.
\newblock {Factorization of quantum density matrices according to {B}ayesian and {M}arkov networks}, 2007.
\newblock \url{https://arxiv.org/abs/quant-ph/0701201}.

\bibitem{Note19}
Physically if we regard the circuit as embedded in space-time such that the absence of directed paths between $\protect \textup {A}_i$ and $\protect \textup {A}_k$ correspond to space-like separation, then the above transformation can be seen as transforming to another reference frame where the time order relative to the co-ordinate time is reversed.

\bibitem{Note20}
Since the scenarios are operationally equivalent with a one-one-one correspondence between the outcome sets, we can use the same labels without loss of generality.

\end{thebibliography}

\end{document}